\documentclass[12pt,draftcls,onecolumn]{IEEEtran}

\usepackage{amsmath}
\usepackage{amssymb}
\usepackage{graphicx}
\usepackage{latexsym}
\usepackage{url}
\usepackage{subfigure}
\usepackage{stfloats}

\newcommand{\e}{{e^{\hat{\xi}\theta}}}

\newcommand{\me}{{e^{-\hat{\xi}\theta}}}

\newcommand{\ewoi}{{e^{\hat{\xi}\theta_{wo_i}}}}
\newcommand{\ewoj}{{e^{\hat{\xi}\theta_{wo_j}}}}
\newcommand{\ewol}{{e^{\hat{\xi}\theta_{wo_{h}}}}}

\newcommand{\mewoi}{{e^{-\hat{\xi}\theta_{wo_i}}}}
\newcommand{\mewoj}{{e^{-\hat{\xi}\theta_{wo_j}}}}
\newcommand{\eioi}{{e^{\hat{\xi}\theta_{io_i}}}}

\newcommand{\ewoiota}{{e^{\hat{\xi}\theta_{w o_{\iota}}}}}
\newcommand{\mewoiota}{{e^{-\hat{\xi}\theta_{w o_{\iota}}}}}

\newcommand{\ebioj}{{e^{\hat{\bar{\xi}}\bar{\theta}_{io_j}}}}
\newcommand{\ebioi}{{e^{\hat{\bar{\xi}}\bar{\theta}_{io_i}}}}
\newcommand{\mebioi}{{e^{-\hat{\bar{\xi}}\bar{\theta}_{io_i}}}}
\newcommand{\mebwojs}{{e^{-\hat{\bar{\xi}}\bar{\theta}_{wo_{j^*}}}}}

\newcommand{\ebiotaoiota}{{e^{\hat{\bar{\xi}}\bar{\theta}_{\iota o_{\iota}}}}}

\newcommand{\ebwoiota}{{e^{\hat{\bar{\xi}}\bar{\theta}_{w o_{\iota}}}}}
\newcommand{\mebwoiota}{{e^{-\hat{\bar{\xi}}\bar{\theta}_{w o_{\iota}}}}}
\newcommand{\debioi}{{\dot{e}^{\hat{\bar{\xi}}\bar{\theta}_{io_i}}}}
\newcommand{\debwoi}{{\dot{e}^{\hat{\bar{\xi}}\bar{\theta}_{wo_i}}}}

\newcommand{\ebwoj}{{e^{\hat{\bar{\xi}}\bar{\theta}_{wo_j}}}}
\newcommand{\des}{{\dot{e}^{\hat{\xi}\theta^*}}}
\newcommand{\ebwoi}{{e^{\hat{\bar{\xi}}\bar{\theta}_{wo_i}}}}
\newcommand{\mebwoi}{{e^{-\hat{\bar{\xi}}\bar{\theta}_{wo_i}}}}
\newcommand{\mebwoj}{{e^{-\hat{\bar{\xi}}\bar{\theta}_{wo_j}}}}
\newcommand{\esi}{{e^{\hat{\xi}\theta_{i}^*}}}

\newcommand{\mesiota}{{e^{-\hat{\xi}\theta_{\iota}^*}}}
\newcommand{\mesi}{{e^{-\hat{\xi}\theta_{i}^*}}}
\newcommand{\es}{{e^{\hat{\xi}\theta^*}}}
\newcommand{\mes}{{e^{-\hat{\xi}\theta^*}}}

\newcommand{\tr}{{\rm tr}}
\newcommand{\proj}{{\rm Proj}}

\newcommand{\sk}{{\rm sk}}
\newcommand{\sym}{{\rm sym}}

\newcommand{\N}{\mathcal{N}}
\newcommand{\SO}{SO}

\newcommand{\V}{{\mathcal V}}
\newcommand{\E}{{\mathcal E}}

\newcommand{\R}{{\cal{R}}}

\newtheorem{prop}{Proposition}
\newtheorem{fact}{Fact}
\newtheorem{definition}{Definition}
\newtheorem{remark}{Remark}
\newtheorem{lemma}{Lemma}

\newtheorem{theorem}{Theorem}
\newtheorem{assumption}{Assumption}

\begin{document}
\title{Cooperative Estimation of 3D Target Motion
via Networked Visual Motion Observer}

\author{Takeshi Hatanaka,~\IEEEmembership{Member,~IEEE} and
Masayuki Fujita,~\IEEEmembership{Member,~IEEE} 
\thanks{Takeshi Hatanaka(corresponding author)
and Masayuki Fujita
are with the Department of Mechanical and Control Engineering,
Tokyo Institute of Technology, Tokyo 152-8550,
JAPAN, {\sf hatanaka@ctrl.titech.ac.jp}, {\sf fujita@ctrl.titech.ac.jp}}
}


\maketitle

\begin{abstract}
This paper investigates cooperative estimation 
of 3D target object motion
for visual sensor networks.
In particular, we consider the situation where multiple smart vision cameras 
see a group of target objects.
The objective here is to meet 
two requirements simultaneously:
averaging for static objects
and tracking to moving target objects.
For this purpose, we 
present a cooperative estimation mechanism 
called networked visual motion observer.
We then derive an upper bound of the ultimate error between
the actual average and the estimates produced by the present networked 
 estimation mechanism.
Moreover, we also analyze 
the tracking performance of the estimates to moving target objects.
Finally the effectiveness of the networked visual motion observer
is demonstrated through simulation.
\end{abstract}

\IEEEpeerreviewmaketitle

\section{Introduction}
\label{sec:1}

A visual sensor network \cite{SH09,ZM_CDC09_1} is 
a kind of wireless sensor network consisting of 
spatially distributed smart cameras with 
communication and computation capability.
Unlike other sensors measuring values such as 
temperature and pressure, vision sensors do not provide
explicit data but combining image processing techniques or
human operators gives rich information on situation awareness such as
what happens, what a target is, where it is and where it bears.
Due to their nature,
visual sensor networks are useful in environmental monitoring,
surveillance, target tracking and entertainment
and are expected as a component of sustainable infrastructures. 

A lot of research works have been devoted to
fusing control techniques with visual information
so-called visual feedback control or images in the loop
\cite{vision}--\cite{HF_MSC10}.
The motivating scenarios of the fusion currently spread
over the robotic systems into security and surveillance
systems, medical imaging procedures, human-in-the-loop
systems and even understanding biological perceptual
information processing. 
Driven by the technological innovations of the smart wearable
cameras, the aforementioned networked vision system also emerges
as a challenging new application field of the visual 
feedback control and estimation.

In this paper, we focus on estimation of
3D rigid body motion as in \cite{DFZD_ACC10}--\cite{HF_MSC10},
and reconsider the problem not for
a single camera system but for the networked vision systems.
In particular, we aim at an extension of \cite{TCST07}
from the single camera to visual sensor networks,
where the paper \cite{TCST07} presents a vision-based observer
called {\it visual motion observer} \cite{HF_MSC10}
estimating 3D target object motion from 2D vision data.
In visual sensor networks, it is expected that not only
an estimate is produced but also the vision cameras 
cooperate with each other in an efficient manner,
which brings us new theoretical challenges.
The advantages of cooperation are: 
(i) accurate estimation by integrating rich information, 
(ii) tolerance against
obstruction, misdetection in image processing
and sensor failures
and (iii) wide vision and elimination of blind areas
by fusing images of a scene from a variety of viewpoints. 
To tackle such distributed estimation problems,
cooperative control  as in
\cite{OFM_IEEE07}--\cite{TCST09}
provides useful methodologies.
In this paper, we especially focus on passivity-based
cooperative control schemes
investigated  in \cite{A_TAC06}--\cite{TCST09}.

Cooperative estimation for sensor networks has been addressed
in \cite{XBL}--\cite{KKJM_CDC10}. 
The main objective of these researches is averaging the local 
measurements or local estimates among sensors 
in a distributed fashion in order to improve estimation accuracy.
For this purpose, most of the
works utilize the consensus protocol \cite{OFM_IEEE07}
in the update of the local estimates.
While \cite{XBL,TVT_CDSC08} assume that 
parameters to be estimated are fixed,
\cite{OS_CDC05a}--\cite{KKJM_CDC10} address estimation
of dynamic parameters assuming that
the parameters follow some dynamical system.
Among them, \cite{OS_CDC05a}--\cite{KM_TSP08} execute a large number of
consensus iterations between each update of estimates,
which is hardly applicable to dynamic estimation problems
except for the case of slow dynamics.
Meanwhile, \cite{BFL_CDC10} and \cite{KKJM_CDC10} present estimation algorithms
without using such iterations.
Unfortunately, however, most of these algorithms are not applicable to our problem
since the object's pose takes values in a non-Euclidean space
and the consensus scheme on a vector space \cite{OFM_IEEE07} does 
not work there. 

Meanwhile, average computation in the group of rotations
is tackled by 
\cite{TVT_CDSC08}, \cite{M_SIAM02}, \cite{M_CARVC04}.
The paper \cite{M_SIAM02} defines two types average rotations, Euclidean and Riemannian means,
and derives their fundamental properties.
Reference \cite{M_CARVC04} presents a computational algorithm
of the Riemannian mean and analyzes its convergence.
The paper \cite{TVT_CDSC08} presents a distributed version of the
algorithm in \cite{M_CARVC04}
based on the consensus protocol \cite{OFM_IEEE07},
which is motivated by the visual sensor networks.
However, \cite{TVT_CDSC08} focuses on averaging by assuming 
that the target orientations are obtained {\it a priori}
and the scheme cannot be essentially extended 
to dynamic estimation problems.


In this paper, we present a novel cooperative estimation 
mechanism called {\it networked visual motion observer}.
We consider the situation where multiple smart vision cameras 
capture a group of target objects.
Under the situation, the objective of the present estimation mechanism is to meet 
two requirements simultaneously: {\it averaging for static objects}, which means 
gaining estimates close to an average of multiple target objects' poses, 
and {\it tracking to moving target objects}, which means that
the estimates track the moving average within a bounded error.
Namely, the present mechanism deals with both static and dynamic 
estimation problems.
For this purpose, we first 
present the {\it networked visual motion observer}, which
consists of the visual feedback 
and mutual feedback from neighboring vision cameras,
based on the 
passivity-based visual motion observer \cite{TCST07}
and the 
passivity-based pose synchronization law presented in \cite{TCST09}.

We next evaluate the averaging performance attained
by the networked visual motion observer.
For this purpose, we define a notion of approximate averaging
by using the ultimate error between
the actual average and the estimates produced by the present observer. 
Then, we derive an upper bound of the ultimate error, whose
partial solution is already given in \cite{ACC11,CDC11} and this paper provides its generalized version.
The result gives us an insight into the gain selection such that
average estimation becomes accurate if mutual feedback 
is much stronger than visual feedback.

We moreover evaluate the tracking performance
of the estimates to moving target objects.
Here, we view the body velocities of the target objects
as a disturbance of the total networked system
and evaluate the ultimate distance from the estimates
to the average.
We see from the result an insight
that choosing a large visual feedback gain
results in a good tracking performance.

Finally, we demonstrate the effectiveness of the present networked visual
motion observer
and validity of the theoretical results
through simulation.

The organization of this paper is as follows.
Section \ref{sec:2} explains the situation under consideration
in this paper and formulates the visual sensor networks 
together with the objective to be met.
In Section \ref{sec:3}, after introducing the visual motion observer
\cite{TCST07}, we present the networked visual motion observer.
Section \ref{sec:4}  clarifies accuracy of the average estimation
when the present estimation mechanism is applied to the network of 
vision cameras.
Section \ref{sec:6} clarifies the
tracking performance of the estimates when the target objects are moving.
Verifications through simulation are shown in Section \ref{sec:7}.
Finally, Section \ref{sec:8} draws conclusions.

We finally give some notations used in this paper,
where the readers are recommended to refer to
\cite{vision} for details on the terminologies.
Throughout this paper, we use the notation
$e^{\hat{\xi}_{ab}\theta_{ab}}\in\R^{3\times 3}$ to represent the rotation 
matrix of a frame $\Sigma_{b}$ relative to a frame $\Sigma_{a}$,
which is orthogonal with unit determinant and hence an element of
the Lie group $SO(3):= \{R\in \R^{3\times 3}|\ R^TR=I_3\mbox{ and }\det(R)=+1\}$.
The vector $\xi_{ab}\in \R^{3}$ specifies the rotation axis and
$\theta_{ab} \in \R$ is the rotation angle.
For simplicity we use ${\xi}\theta_{ab}$ to denote 
${\xi}_{ab}\theta_{ab}$. 
The configuration space of the rigid body motion is the product space 
$SE(3) := \R^3 \times SO(3)$.
We use the $4 \times 4$ matrix
$g_{ab}=\left[\begin{array}{cc}
e^{\hat{\xi}\theta_{ab}}&p_{ab}\cr
0&1
\end{array}
\right]$
as the homogeneous representation of
$g_{ab}=(p_{ab},e^{\hat{\xi}\theta_{ab}})\in SE(3)$
describing the configuration of $\Sigma_{b}$
relative to $\Sigma_{a}$.
The notation {\rm `$\wedge$'} is the operator 
such that $\hat{a}b=a\times b$ for the vector cross-product $\times$,
i.e. $\hat{a}$ is a $3\times3$ skew-symmetric matrix. 
The vector space of all $3\times 3$ skew-symmetric matrices is 
denoted by $so(3)$. 
The notation {\rm `$\vee$'} denotes the inverse operator to
{\rm `$\wedge$'}.
Similarly to the definition of $so(3)$, we define 
$se(3):=\{(v,\hat{\omega}):v\in\R^{3},\hat{\omega}\in so(3)\}$. 
In homogeneous representation, we write an element $V:=(v,\omega)$ 
as 
$\hat{V}=\left[\begin{array}{cc}
\hat{\omega}&v\cr
0&0
\end{array}\right]$.

\section{Preparation for Visual Sensor Networks}\label{sec:2}

Let us consider the situation where $n$ vision cameras
$\V := \{1, \cdots, n\}$ with
communication and computation capability
see a group of target objects $\{o_i\}_{i \in \V}$ 
(Fig.~\ref{fig:Ad_VIM}), 
where each vision camera $i\in \V$ captures object $o_i$ on its image plane.
Throughout this paper, we use the pinhole-type
vision cameras with perspective projection \cite{vision}
as in Fig.~\ref{fig:camera}.
Note however that 
all of the subsequent discussions are applicable to panoramic  cameras
through the modifications in \cite{ACC10}.

In this paper, we address estimation of average motion 
of the objects $\{o_i\}_{i \in \V}$.
The problem includes a scenario such that all the cameras
see a common single target object but the pose consistent with vision data
differs from camera to camera due to incomplete localization
and parametric uncertainties.
Under such a situation, averaging the contaminated poses is
a way to improve estimation accuracy \cite{OS_CDC07}.


\begin{figure}[t]
\centering
\begin{minipage}{6.7cm}
\includegraphics[width=6.7cm]{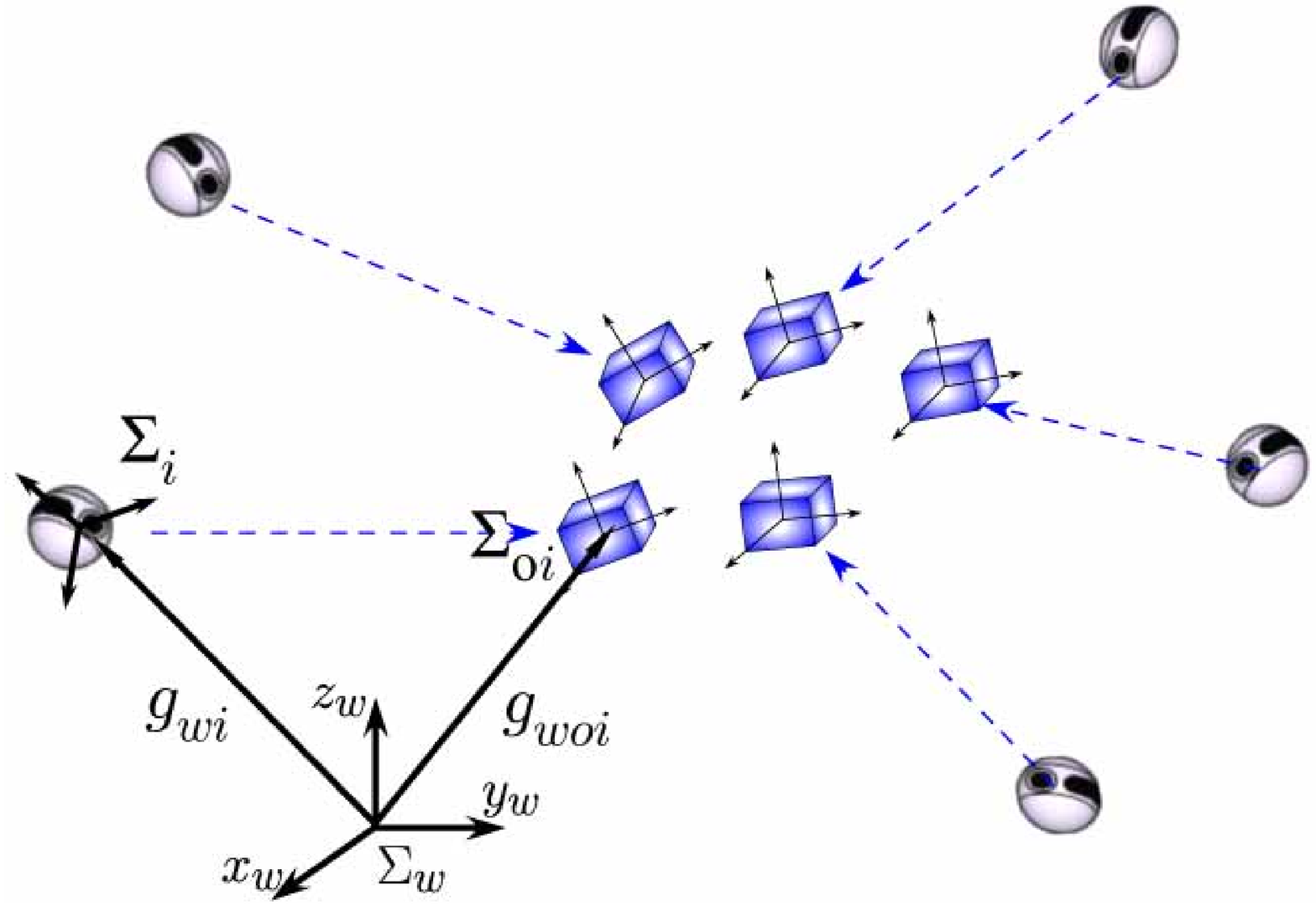}
\caption{Visual Sensor Networks}
\label{fig:Ad_VIM}
\end{minipage}
\hspace{1cm}
\begin{minipage}{5.5cm}
\centering
{\includegraphics[width=5.5cm]{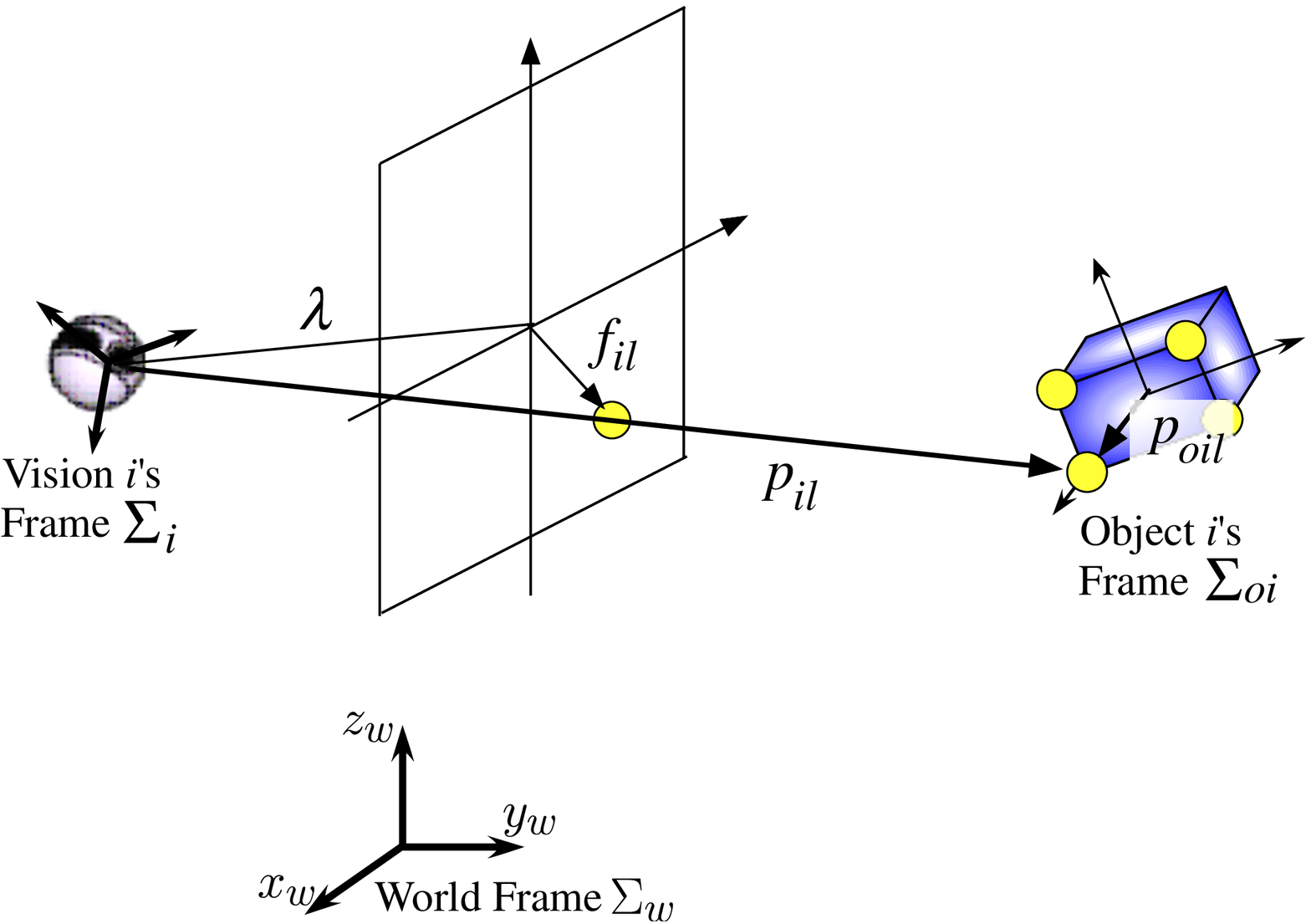}}
\caption{Vision Camera Model}
\label{fig:camera}
\end{minipage}
\end{figure}

\subsection{Rigid Body Motion}
\label{sec:2.1}

Let the coordinate frames $\Sigma_{w}$, $\Sigma_{i}$ and $\Sigma_{o_i}$
represent the world frame, the $i$-th vision camera frame, and
the frame of object $o_i$, respectively.
The pose of vision camera $\Sigma_i$ and object $\Sigma_{o_i}$
relative to the world frame $\Sigma_w$
are denoted by $g_{wi}=(p_{wi},e^{\hat{\xi}\theta_{wi}})\in SE(3)$
and $g_{wo_i}=(p_{wo_i},e^{\hat{\xi}\theta_{wo_i}})\in SE(3)$.
Then, the pose of $\Sigma_{o_i}$ relative to $\Sigma_{i}$, denoted
by $g_{io_i}=(p_{io_i},e^{\hat{\xi}\theta_{io_i}})\in SE(3)$,
can be represented as $g_{io_i} = g_{wi}^{-1}g_{wo_i}$.

We next define the body velocity of object ${o_i}$ relative to the world
frame $\Sigma_{w}$ as $V_{wo_i}^{b} = (v_{wo_i}, \omega_{wo_i})\in \R^6$,
where $v_{wo_i}$ and $\omega_{wo_i}$ respectively
represent the linear and angular velocities
of the origin of $\Sigma_{o_i}$
relative to $\Sigma_{w}$ \cite{vision}.
Similarly, vision camera $i$'s body velocity relative to $\Sigma_{w}$
will be denoted as $V_{wi}^{b} = (v_{wi}, \omega_{wi})\in \R^6$.

By using the body velocities $V_{wi}^{b}$ and $V_{wo_i}^{b}$,
the motion of the relative pose $g_{io_i}$ is written as
\begin{eqnarray}
\dot{g}_{io_i}=-\hat{V}_{wi}^{b}g_{io_i}+g_{io_i}\hat{V}_{wo_i}^{b}\label{eqn:RRBM}
\end{eqnarray}
\cite{vision}.
Equation (\ref{eqn:RRBM}) is called relative rigid body motion
whose
block diagram is depicted in Fig. \ref{fig:RRBM}.


\subsection{Visual Measurement}
\label{sec:2.2}

In this subsection, we define visual measurements
of each vision camera which is available for estimation.
We assume that each target object 
has $m$ feature points and each vision camera can extract them from the
vision data by using some techniques like \cite{SURF}.
The position vectors of the target object $i$'s $l$-th feature point
relative to $\Sigma_{o_i}$ and $\Sigma_{i}$ are
denoted by $p_{o_il}\in\R^{3}$ and $p_{il}\in\R^{3}$ respectively.
Using a transformation of the coordinates, we have
$p_{il}=g_{io_i}p_{o_il}$,
where $p_{o_il}$ and $p_{il}$ should be regarded with a slight abuse
of notation as $[p_{o_il}^{T}\ 1]^{T}$ and $[p_{il}^{T}\ 1]^{T}$.

Let the $m$ feature points of object $o_i$ on
the image plane coordinate 
be the measurement $f_i$ of camera $i$, which is
given by the perspective projection 
\cite{vision} with  a focal length
$\lambda_i$ as
\begin{eqnarray}
f_{i}:=[f_{i1}^{T}\ \cdots\ f_{im}^{T}]^{T}\in\R^{2m},\ \
f_{il}=\frac{\lambda_i}{z_{il}}
[{x_{il}}\ \
y_{il}]
,\
p_{il}=[x_{il}\ y_{il}\ z_{il}]^{T},
\label{eqn:f_i}
\end{eqnarray}
Under the assumption that each camera $i$
knows the location of feature points $p_{o_il}\in\R^{3}$, 
the visual measurement $f_{i}$ depends
only on the relative pose $g_{io_i}$ from (\ref{eqn:f_i}) and $p_{il}=g_{io_i}p_{o_il}$.
Fig.~\ref{fig:RRBMtoImage} shows the block diagram of 
the relative rigid body motion with the camera model.
\begin{figure}[t]
\centering
\begin{minipage}{7cm}
\centering
\includegraphics[width=6cm]{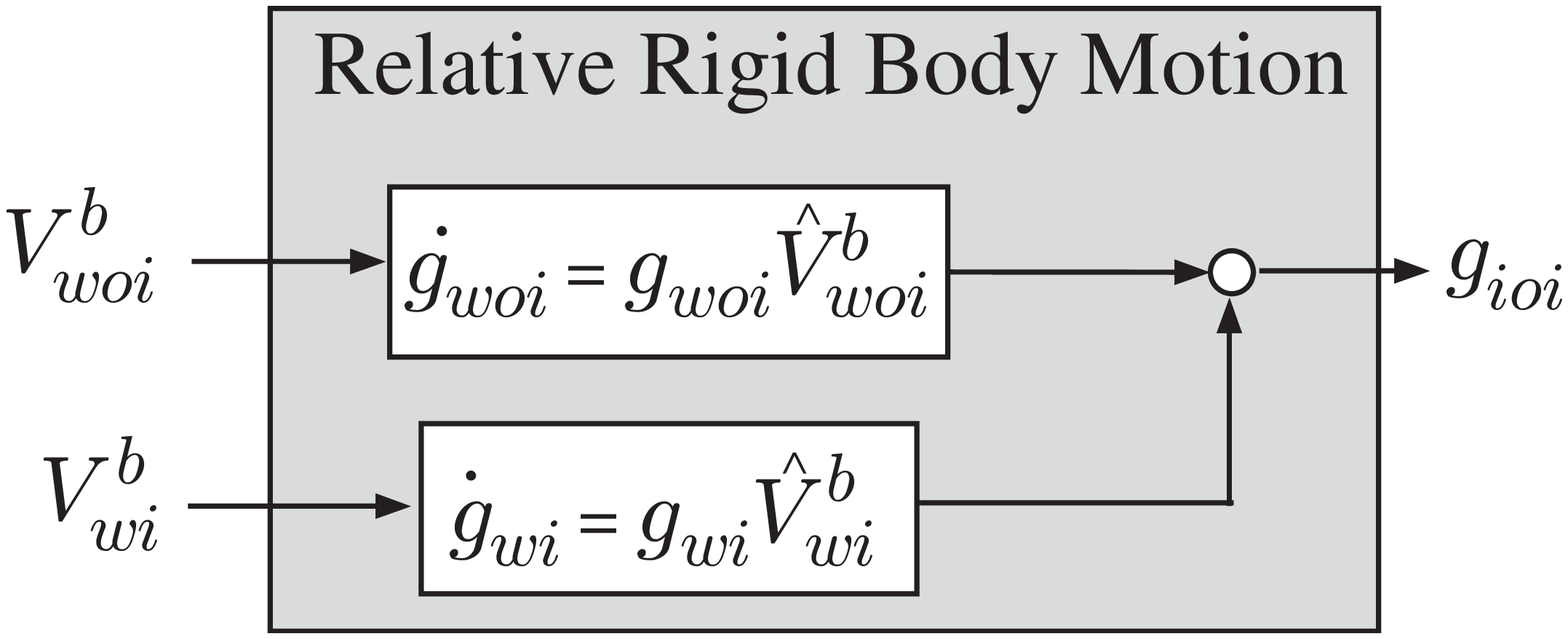}
\caption{Block Diagram of 
Relative Rigid Body Motion}
\label{fig:RRBM}
\end{minipage}
\hspace{.5cm}
\begin{minipage}{8.5cm}
\centering
\includegraphics[width=8.5cm]{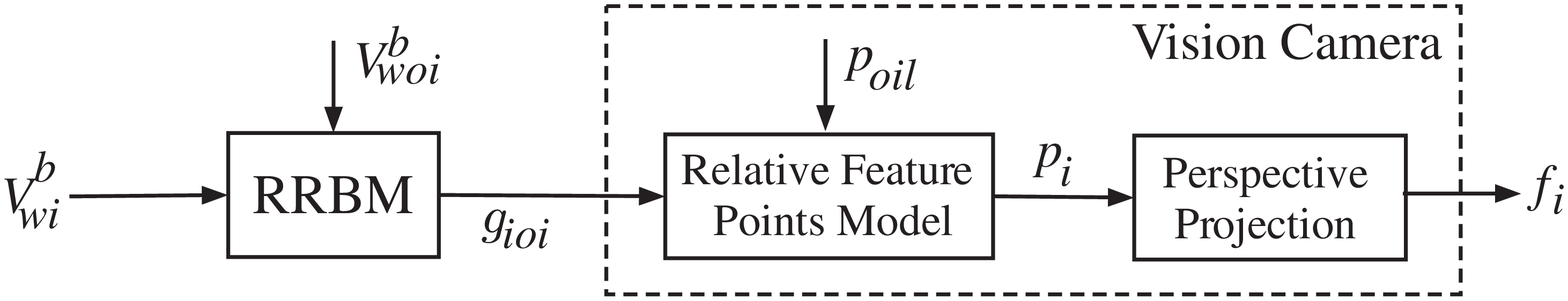}
\caption{Block Diagram of the RRBM with Vision Camera
(RRBM is an acronym for Relative Rigid Body Motion)}
\label{fig:RRBMtoImage}
\end{minipage}
\end{figure}

\subsection{Communication}
\label{sec:2.3}

The vision cameras have communication capability with the neighboring cameras
and form a network.
The communication is modeled by a digraph $G=(\V,\E)$, where
$\E \subset \V \times \V$ as in the left figure of Fig. \ref{fig:graph1-2}.
Namely, vision camera $i$ can get some information
from $j$ if $(j,i) \in \E$.
In addition, we define the neighbor set $\N_i$ of vision camera $i\in \V$ as
\begin{equation}
\N_i := \{j\in\V |\ (j,i) \in \E\}.
\label{eqn:neighbor}
\end{equation}
Let us now employ the following assumption on the graph $G$.
\begin{assumption}
\label{ass:1}
The communication graph $G$ is fixed, balanced and
strongly connected.
\end{assumption}
The balanced and strongly connected graph is a graph
such that there exists at least one directed path 
between any pair of nodes 
and the in-degree and out-degree are equal for
all nodes \cite{bullo}.
\begin{figure}[t]
\begin{center}
\begin{minipage}{4.5cm}
\begin{center}
\includegraphics[width = 4.5cm]{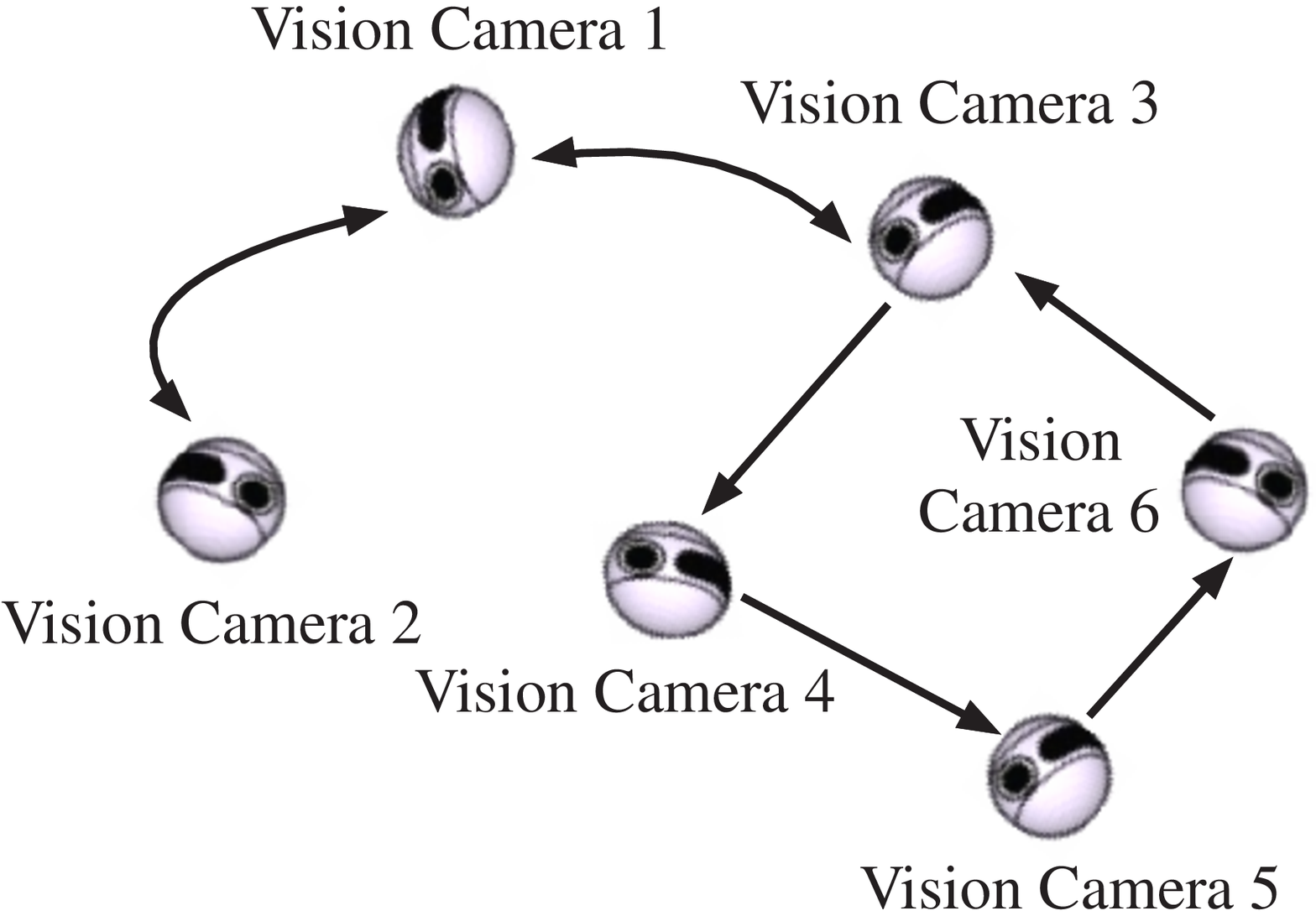}
\end{center}
\end{minipage}
\hspace{.5cm}
\begin{minipage}{5.5cm}
\begin{center}
\includegraphics[width = 5.5cm]{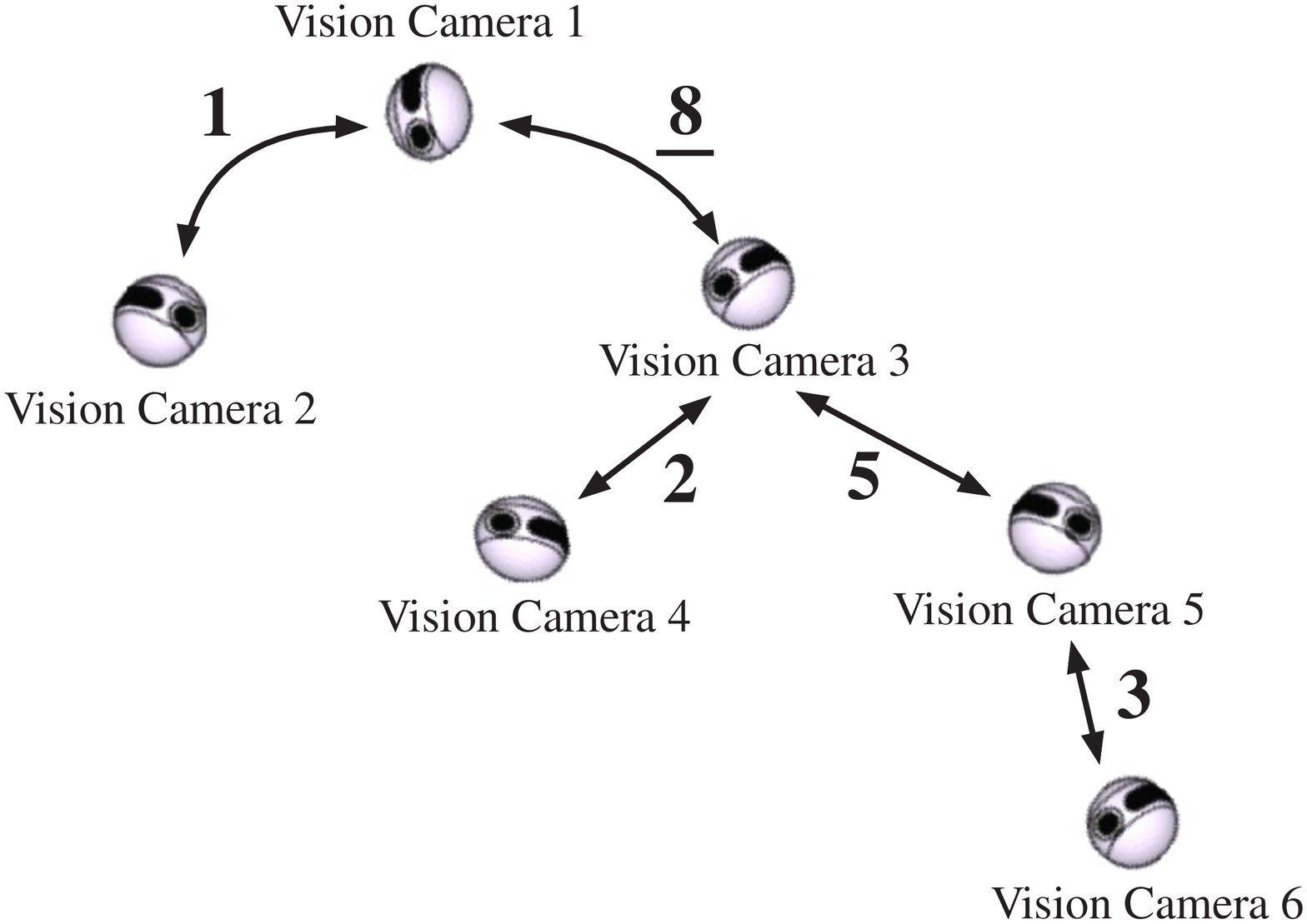}
\end{center}
\end{minipage}
\hspace{.5cm}
\begin{minipage}{4.7cm}
\begin{center}
\includegraphics[width = 4.7cm]{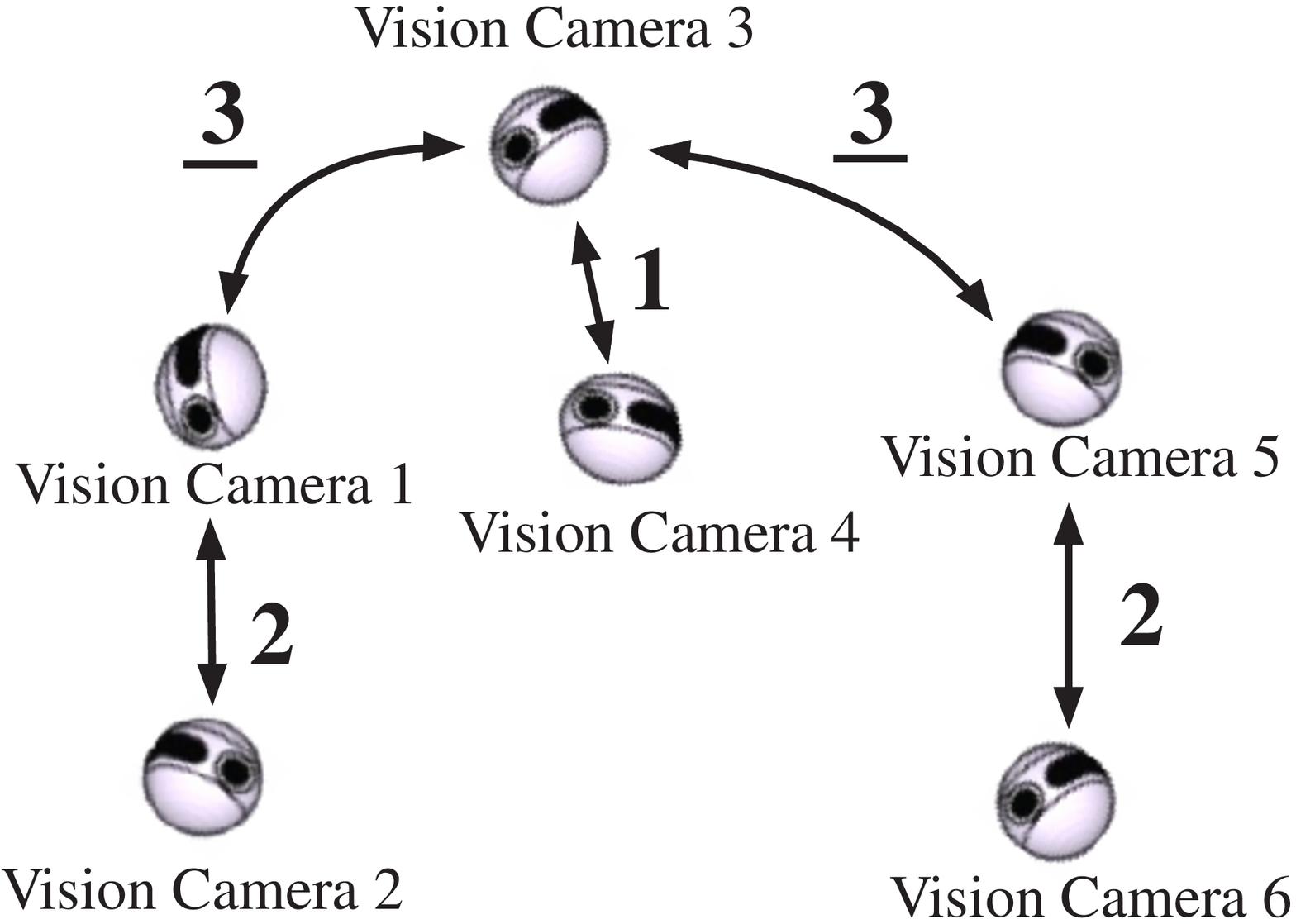}
\end{center}
\end{minipage}
\caption{Left: Communication Graph, Middle:
Tree with Root $1$ Minimizing $\tilde{D}$ 
($\tilde{D}=8$), Right: Tree Minimizing ${D}$ ($W = 3$)}
\label{fig:graph1-2}
\end{center}
\end{figure}

We also denote by $G_u$
the undirected graph produced by replacing all the directed edges of
$G$ by the undirected ones.
Let ${\mathcal T}(i_0)$ be the set of all 
spanning trees over $G_u$ with a root $i_0 \in \V$
and we consider an element $G_T = (\V, \E_T)\in {\mathcal T}(i_0)$.
Let the path from $i_0$ to a node $i\in \V$ along with the tree $G_T$
be denoted by $P_{G_T}(i) = (v_0,\cdots, v_{d_{G_T}(i)}),\ v_0 = i_0,\ v_{d_{G_T}(i)} = 
i,\ (v_l,v_{l+1}) \in \E_T\ {\forall l}\in \{0, \cdots,d_{G_T}(i)-1\}$,
where $d_{G_T}(i)$ denotes the length of the path $P_{G_T}(i)$.
We also define 
\[
\delta_{G_T}(E;i) =\left\{
\begin{array}{ll}
1,&\mbox{if the path } P_{G_T}(i) \mbox{ includes edge } E\\
0,&\mbox{otherwise}
\end{array}
\right.
\]
for any $E \in \E_T$.
By using the above notations, we define
\begin{equation}
 W := \min_{i_0 \in \V} D(i_0),\ D(i_0) := \min_{G_T \in {\mathcal T}(i_0)}\tilde{D}(G_T),\
\tilde{D}(G_T) := \max_{E\in \E_T}\sum_{i\in \V}\delta_{G_T}(E;i)d_{G_T}(i).
\label{eqn:defD}
\end{equation}

For example, let us consider the communication graph in Fig. \ref{fig:graph1-2}(Left).
Suppose that we choose $i_0 = 1$ and build a tree depicted in the middle
figure of Fig. \ref{fig:graph1-2}, where the number at around 
each edge is the value of $\sum_{i\in \V}\delta_{G_T}(E;i)d_{G_T}(i)$.
Namely, $\tilde{D}$ is equal to $8$ for the tree and it is
actually minimal for all spanning trees in ${\mathcal T}(1)$.
However, choosing another node as a root can reduce the value of $\tilde{D}$.
Indeed, as illustrated in the right figure of Fig. \ref{fig:graph1-2},
a tree with $i_0 = 3$ achieves $\tilde{D} = 3$, which is the minimal 
$D(i_0)$ among all the choices of the
 root $i_0$.

\subsection{Average on $SO(3)$ and $SE(3)$}
\label{sec:2.4}

In this paper, the tuple of the 
relative rigid body motion (\ref{eqn:RRBM}),
the visual measurement (\ref{eqn:f_i}) and the communication
structure (\ref{eqn:neighbor}) is called a visual sensor network.
The objective of this paper is to present
a cooperative estimation mechanism for the visual sensor networks 
meeting the following requirements simultaneously:
{\it Averaging for static objects}, which means each camera $i$ estimates a pose 
close to an average of $\{g_{io_j}\}_{j \in \V},\ g_{io_j} := g_{wi}^{-1}g_{wo_j}$,
{\it Tracking to moving objects}, which means the estimates track the moving average pose
within a bounded tracking error.

Let us now introduce
the following mean $g^*$ on $SE(3)$
as an average of target poses $\{g_{wo_j}\}_{j \in \V}$.
\begin{eqnarray}
&&\hspace{-.7cm}g^* = (p^*,\es) := \arg\min_{g \in SE(3)}\sum_{j \in \V}
\psi(g^{-1}g_{wo_j}),
\label{eqn:average}
\end{eqnarray}
where the function $\psi$ is defined for any $g = (p, \e)\in SE(3)$ as
\begin{eqnarray}
&&\psi(g) := \frac{1}{2}\|I_4 - g\|_F^2 =
\frac{1}{2}\|p\|^2 + \phi(\e),\
\phi(\e) := \frac{1}{2}\|I_3 - \e\|_F^2
= \tr(I_3 - \e)
\label{eqn:psi}
\end{eqnarray}
and $\|M\|_F$ is the matrix Frobenius norm of matrix $M$.
Hereafter, we also use the notation 
\[
 g_{i}^* = (p_i^*, \esi):= \arg\min_{g_i\in SE(3)}\sum_{j\in \V}\psi(g_i^{-1}g_{io_j})
= g_{wi}^{-1}g^*. 
\]

The position average $p^*$ is equal to
the arithmetic mean
$p^* = \frac{1}{n}\sum_{j\in \V}p_{wo_j}$
of target positions $\{p_{wo_j}\}_{j\in\V}$ and the orientation
average $\es$ is
a so-called Euclidean mean \cite{M_SIAM02} 
of target orientations $\{\ewoj\}_{j\in\V}$ defined by
\begin{equation}
\es := \arg\min_{\e\in\SO(3)}\sum_{j\in\V}\phi(\me\ewoj).
\label{eqn:euclidean_mean}
\end{equation}
It is known \cite{M_SIAM02} that the Euclidean mean $\es$ is given by
\begin{equation}
\es(t) = {\rm Proj}\left(S(t)\right),\ 
S(t):=\frac{1}{n}\sum_{j\in \V}\ewoj(t).
\label{proj}
\end{equation}
Here, ${\rm Proj}(M)$ is the orthogonal projection of $M\in \R^{3\times 
3}$
onto $SO(3)$, which is given by $U_MV_M^T$ for the matrix
$M$ with singular value decomposition $M=U_M\Sigma V_M^T$ \cite{M_SIAM02}.

\begin{remark}
\label{rem:1}
Just computing the Euclidean mean
is not so difficult even in a distributed fashion
if we have prior knowledge that the target object is static.
Indeed, the matrix $S$ is computed by using the consensus protocol
under appropriate assumptions on the graph \cite{OFM_IEEE07}
and the operation ${\rm Proj}$ can be locally executed.
However, such a scheme works only for static objects
and never embodies tracking nature for moving target objects.
The objective here is to present an estimation mechanism without
using any prior knowledge and any decision-making process on
whether the targets are static or moving.
\end{remark}




\section{Networked Visual Motion Observer}
\label{sec:3}

\subsection{Visual Motion Observer} 
\label{sec:3.1}

In this subsection, we consider the problem that
vision camera $i$ estimates the target object motion
$g_{io_i}$ from the visual measurements $f_i$
without considering communication.
For the purpose, we introduce the {\it visual motion observer}
presented in \cite{TCST07}.


We first prepare a model of the rigid body motion (\ref{eqn:RRBM})
similarly to the Luenberger observer as
\begin{eqnarray}
\dot{\bar{g}}_{io_i}=-\hat{V}_{wi}^{b}\bar{g}_{io_i}+\bar{g}_{io_i}\hat{u}_{ei},\label{eqn:EsRRBM}
\end{eqnarray}
where $\bar{g}_{io_i}=(\bar{p}_{io_i},e^{\hat{\bar{\xi}}\bar{\theta}_{io_i}})$
is the estimate of the actual relative pose $g_{io_i}$.
The input $u_{ei}=(v_{uei},\ \omega_{uei})$ is to be
determined to drive the estimated value $\bar{g}_{io_i}$
to the actual $g_{io_i}$.

In order to establish the estimation error system,
we define the estimation error between the estimated value
$\bar{g}_{io_i}$ and the actual relative rigid body motion $g_{io_i}$ as
$g_{ei}=(p_{ei}, e^{\hat{\xi}\theta_{ei}})
:=\bar{g}_{io_i}^{-1}g_{io_i}$.
Using the notations
$e_{R}(e^{\hat{\xi}\theta}):=\sk(e^{\hat{\xi}\theta})^{\vee}$
and $\sk(e^{\hat{\xi}\theta}) 
:= \frac{1}{2}(e^{\hat{\xi}\theta}-e^{-\hat{\xi}\theta})$,
the vector representation of the estimation error $g_{ei}$ is given by
\begin{eqnarray}
e_{ei}:=E_R(g_{ei}),\ E_R(g_{ei}) :=
\left[\begin{array}{cc}
p_{ei}^{T}&e_{R}^{T}(e^{\hat{\xi}\theta_{ei}})\end{array}\right]^{T}.
\label{eqn:eei}
\end{eqnarray}

Once the estimate $\bar{g}_{io_i}$ is determined,
the estimated measurement $\bar{f}_{i}$ is
also computed by (\ref{eqn:f_i}).
Let us now define the visual measurement error as $f_{ei}:=f_{i}(g_{io_i})-\bar{f}_i(\bar{g}_{io_i})$.
Then, the measurement error vector $f_{ei}$ can be approximately given by 
$f_{ei}=J_i(\bar{g}_{io_i})e_{ei}$
\cite{TCST07}, where $J_i(\bar{g}_{io_i}):SE(3)\rightarrow\R^{2m\times 6}$ is
the well-known image Jacobian. 
Now, if $m\geq4$, the image Jacobian has the full column rank
and the estimation error vector $e_{ei}$ is reconstructed
as
\begin{eqnarray}
e_{ei}=J_i^{\dagger}(\bar{g}_{io_i})f_{ei}, \label{eqn:e_e}
\end{eqnarray}
where $\dagger$ denotes the pseudo-inverse.

Differentiating $g_{ei}=\bar{g}_{io_i}^{-1}g_{io_i}$ with respect to time
and using (\ref{eqn:RRBM}) and (\ref{eqn:EsRRBM}), we obtain
the estimation error system
\begin{eqnarray}
\dot{g}_{ei}=-\hat{u}_{ei}g_{ei}+g_{ei}\hat{V}_{wo_i}^{b}\label{eqn:vmo_update}.
\end{eqnarray}
Fig.~\ref{fig:EstErrSys} shows the block diagram of the
 system (\ref{eqn:vmo_update}).
\begin{figure}[t]
\centering
\begin{minipage}{7.5cm}
{\includegraphics[width=7.5cm]{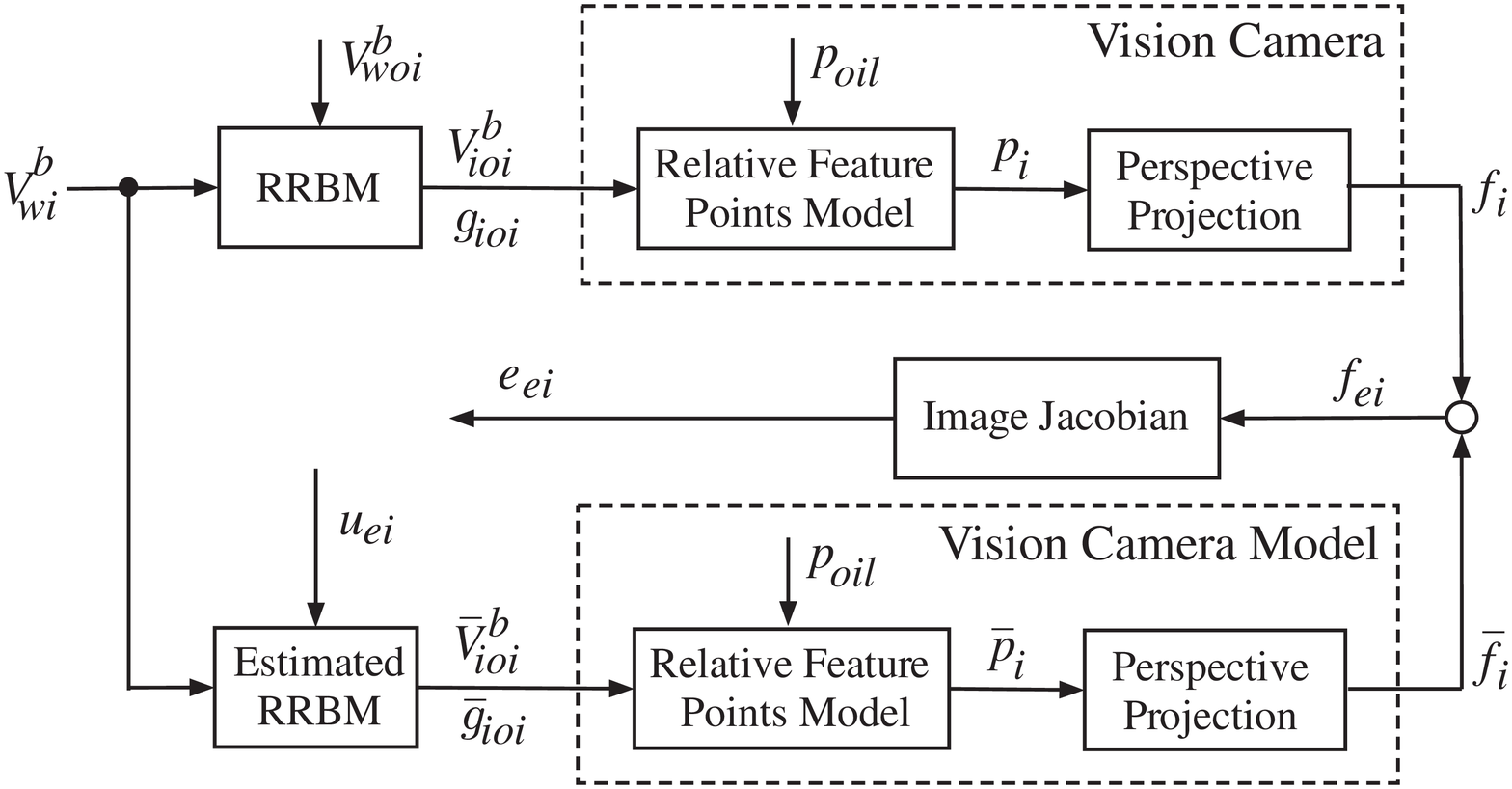}}
\caption{Estimation Error System}
\label{fig:EstErrSys}
\end{minipage}
\hspace{5mm}
\begin{minipage}{7.5cm}
{\includegraphics[width=7.5cm]{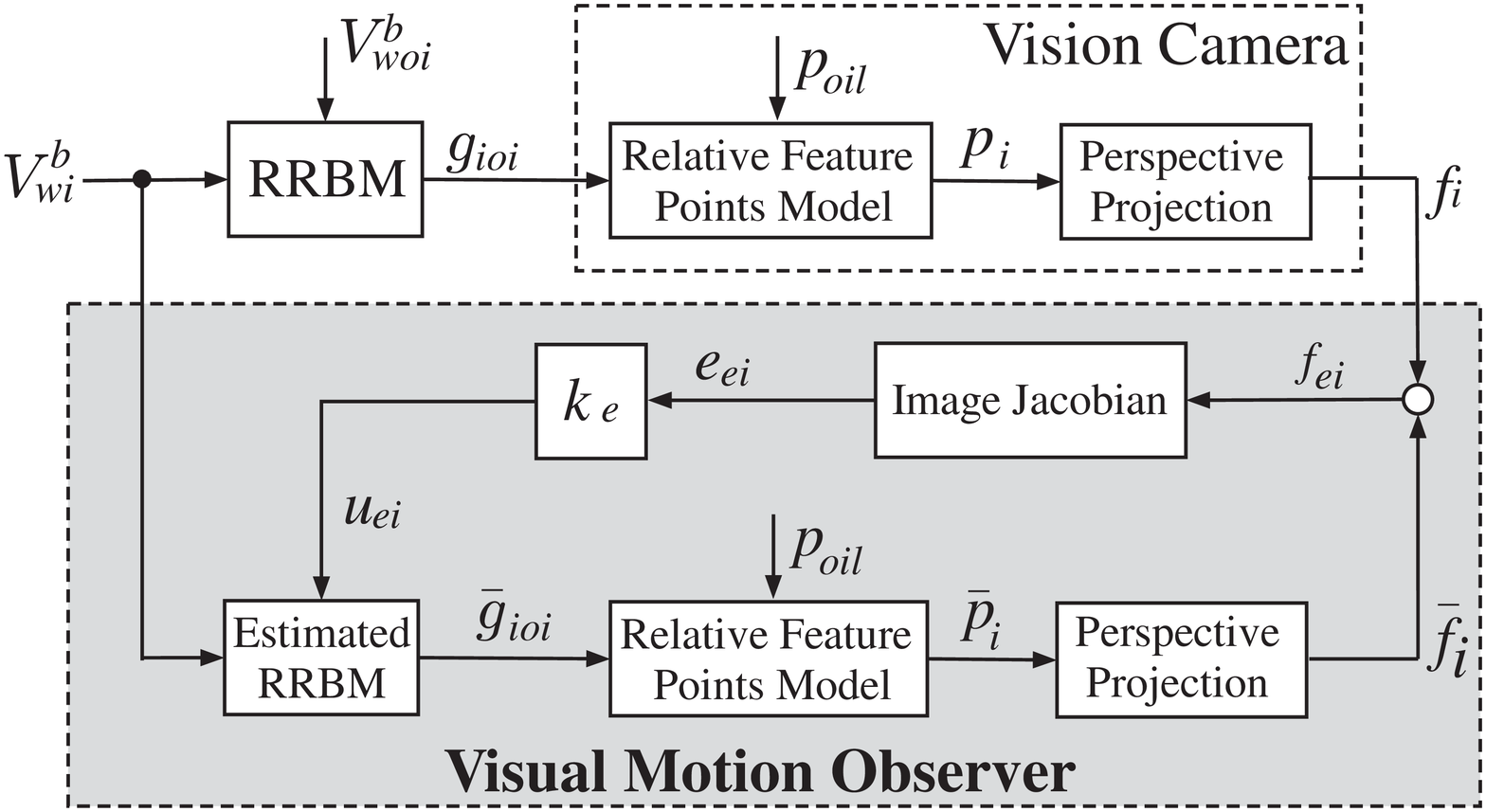}}
\caption{Visual Motion Observer}
\label{fig:vmo}
\end{minipage}
\end{figure}
The paper \cite{TCST07} proves that
if $V_{wo_i}^{b}=0$, then 
the estimation error system
(\ref{eqn:vmo_update}) is passive from the input $u_{ei}$ to the
output $-e_{ei}$.

Based on passivity-based control theory,
we close the loop by using the input
\begin{equation}
 u_{ei} = -k_{e} (-e_{ei}) = k_{e} e_{ei},\ k_{e} > 0.
\label{eqn:vmo_input}
\end{equation}
Then, the resulting total estimation mechanism 
formulated as
\begin{eqnarray}
\mbox{Visual Motion Observer: }
\left\{
\begin{array}{lr}
\dot{\bar{g}}_{io_i}=-\hat{V}_{wi}^{b}\bar{g}_{io_i}+\bar{g}_{io_i}\hat
 {u}_{ei}&\cdots (\ref{eqn:EsRRBM})\\
e_{ei}=J_i^{\dagger}(\bar{g}_{io_i})f_{ei}&\cdots (\ref{eqn:e_e})\\
u_{ei} = k_{e} e_{ei}&\cdots (\ref{eqn:vmo_input})
\end{array}
\right.
\label{eqn:visual_motion_observer}
\end{eqnarray}
is called {\it visual motion observer} \cite{HF_MSC10},
whose block diagram is illustrated in Fig.~\ref{fig:vmo}.
In terms of the mechanism, we immediately obtain the following facts 
from passivity.
\begin{fact}\cite{TCST07}
\label{vmo_theorem}\mbox{}
(i) If $V_{wo_i}^{b}=0$, then the equilibrium point $e_{ei} = 0$
for the closed-loop system (\ref{eqn:vmo_update}) with
(\ref{eqn:vmo_input}) is asymptotically stable.
(ii) Given a positive scalar $\nu_i$, 
if $k_{e}$ satisfies 
$k_{e} - \frac{1}{2\nu_i^2} - \frac{1}{2} > 0$,
then the system (\ref{eqn:vmo_update}) and
(\ref{eqn:vmo_input}) with input $V^b_{wo_i}$
and output $e_{ei}$ has $L_2$-gain smaller than $\nu_i$.
\end{fact}
Item (i) means 
the visual motion observer leads the estimate 
$\bar{g}_{io_i}$ to the actual 
$g_{io_i}$ for a static object.
Item (ii) implies that the observer
also works for a moving target object, and
the parameter $\nu_i$ is an index on estimation accuracy
when the  observer is applied to a moving target.


\subsection{Networked Visual Motion Observer}
\label{sec:3.3}

The objective of this paper is to 
achieve averaging,
while preserving the tracking nature of 
the visual motion observer.
For this purpose, this subsection presents a cooperative estimation mechanism
under the assumption of (i) each vision camera knows
relative pose $g_{ij} = g_{wi}^{-1}g_{wj}$
with respect to neighbors $j \in \N_i$
and (ii) all the vision cameras are static, i.e. $V^b_{wi} = 0\ {\forall i}\in \V$.

Under $V^b_{wi} = 0$, the relative rigid body motion (\ref{eqn:RRBM}) 
is simply given by 
$\dot{g}_{io_i}=g_{io_i}\hat{V}_{wo_i}^{b}$.
Accordingly, the update procedure in (\ref{eqn:visual_motion_observer}) 
is reformulated as
\begin{eqnarray}
\dot{\bar{g}}_{io_i} = \bar{g}_{io_i}\hat{u}_{ei},\ 
u_{ei} = k_{e}e_{ei}.
\label{eqn:vmo_update_static}
\end{eqnarray}
Then, the following proposition holds in terms of
the procedure (\ref{eqn:vmo_update_static}).
\begin{prop}\cite{mm}
The update procedure (\ref{eqn:vmo_update_static})
is a gradient decent algorithm on $SE(3)$
for the potential function
$\psi(\bar{g}_{io_i}g_{io_i})$, where the function $\psi$
is defined in (\ref{eqn:psi}).
\end{prop}

Let us now view $\psi(\bar{g}^{-1}_{io_i}g_{io_i}) = \psi(\bar{g}^{-1}_{wo_i}g_{wo_i})$
as the local objective function to be minimized by vision camera $i$.
Then, we see that the group objective (\ref{eqn:average})
is given by the sum of the local objective functions for all $i \in \V$.
Note that each vision camera does not know the local objective of
the other vision cameras.
Under such a situation
computing a solution minimizing the global objective function 
by using local negotiations is called {\it multi-agent optimization problem}
and \cite{NO_TAC09} presents an update rule of the local estimates of 
the solution to produce approximate solutions to the global objective
combining the gradient decent algorithm
of the local objective function and the consensus protocol \cite{OFM_IEEE07}.
The present cooperative estimation mechanism is inspired by
the algorithm but the consensus protocol cannot be executed on $SE(3)$.
We thus instead use a pose synchronization law presented in \cite{TCST09},
which is also based on passivity of rigid body motion.
\begin{figure}[t]
\centering
{\includegraphics[width=10cm]{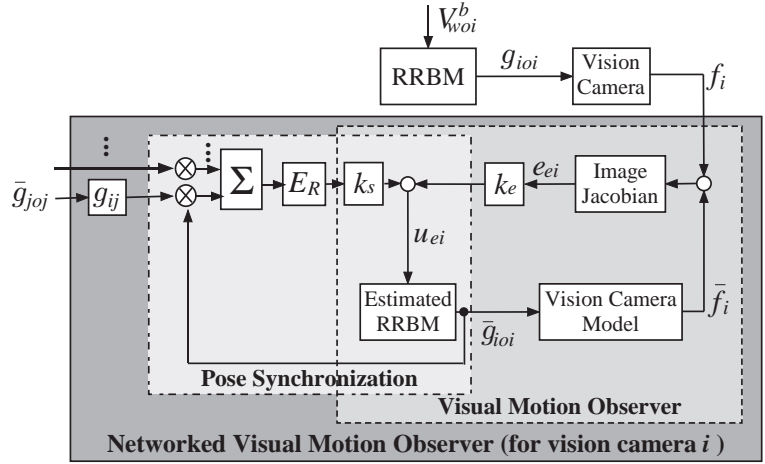}}
\caption{Networked Visual Motion Observer}
\label{fig:ce}
\end{figure}

We next present an update rule of the estimates
$\bar{g}_{io_i}$ so as to estimate the average $g_i^*$.
Each vision camera $i$ first gains the estimates $\bar{g}_{jo_j}$
from $j \in \N_i$ as messages.
Now, by multiplying known information
$g_{ij}$ from left, each vision camera $i$
gets $\bar{g}_{io_j}:=g_{ij}\bar{g}_{jo_j}$
for all $j \in \N_i$.
Using the information,
the estimate $\bar{g}_{io_i}$ is updated according to 
(\ref{eqn:EsRRBM}) with
\begin{eqnarray}
{u}_{ei} = k_ee_{ei} + 
k_s\sum_{j\in \N_i}
E_R(\bar{g}_{io_i}^{-1}\bar{g}_{io_j}),\ k_e > 0,\ k_s > 0.
\label{eqn:ce_update}
\end{eqnarray}
Since $e_{ei}$ is reconstructed from the visual
measurement $f_i$ by (\ref{eqn:e_e})
and $\bar{g}_{io_j}$ is obtained through communication 
as stated above,
the update procedure (\ref{eqn:ce_update})
is implementable.

The present input (\ref{eqn:ce_update}) 
consists of the visual feedback term $k_e e_{ei}$ 
and the mutual feedback term $k_s \sum_{j\in \N_i} E_R(\bar{g}_{io_i}^{-1}\bar{g}_{io_j})$,
where the former is inspired by the visual motion observer \cite{TCST07}
and the latter is by the pose synchronization
law \cite{TCST09}.
Indeed,  without the second term,
the update rule (\ref{eqn:ce_update})
is the same as that of the visual motion observer
(\ref{eqn:vmo_update_static}).
In addition, without the visual feedback, the update procedure 
(\ref{eqn:ce_update}), namely ${u}_{ei} = k_s\sum_{j\in \N_i}
E_R(\bar{g}_{io_i}^{-1}\bar{g}_{io_j})$,
is essentially equivalent to the passivity-based pose synchronization
law \cite{TCST09} of a group of rigid bodies 
with states 
$\bar{g}_{wo_i} := g_{wi}\bar{g}_{io_i}$.
Thus, under appropriate assumptions, each state $\bar{g}_{wo_i}$
would converge to a state satisfying $\bar{g}_{wo_i} = \bar{g}_{wo_j}\ 
{\forall i,j}\in \V$
as time goes to infinity without the visual feedback term. 

In other words, the visual motion observers are networked by
the mutual feedback term in the total 
estimation mechanism formulated as
\begin{eqnarray}
\mbox{
Networked VMO: }
\left\{
\begin{array}{ll}
\dot{\bar{g}}_{io_i}=\bar{g}_{io_i}\hat
 {u}_{ei}&\cdots (\ref{eqn:EsRRBM})\\
e_{ei}=J_i^{\dagger}(\bar{g}_{io_i})f_{ei}&\cdots (\ref{eqn:e_e})\\
{u}_{ei} = k_ee_{ei} + 
k_s\sum_{j\in \N_i}
E_R(\bar{g}_{io_i}^{-1}\bar{g}_{io_j})&\cdots (\ref{eqn:ce_update})
\end{array}
\right.
{\forall i}\in \V,
\label{eqn:net_visual_motion_observer}
\end{eqnarray}
where VMO is an acronym for Visual Motion Observer.
This is why the estimation mechanism is called 
{\it networked visual motion observer}.
The block diagram of the total system
of vision camera $i$ is illustrated in Fig. \ref{fig:ce}.


\section{Averaging Performance Analysis}
\label{sec:4}

In this section, we derive ultimate estimation accuracy
of the average $g_i^*$ achieved by
the networked visual motion observer
(\ref{eqn:net_visual_motion_observer})
under the following assumption.
\begin{assumption}
\label{ass:2}\mbox{}\\
(i) The target objects are static, i.e. $V^b_{wo_i} = 0\ {\forall 
	   i}\in \V$.\\
(ii) There exists a pair $(i,j) \in \V\times \V$ such that
$p_{wo_i} \neq p_{wo_j}$ and $\ewoi \neq \ewoj$.\\
(iii) $\mesi\eioi > 0$ for all $i \in \V$.
\footnote{Throughout this paper, we refer to a real matrix $M$,
which is not necessarily symmetric, as a
positive definite (positive semi-definite) matrix
if and only if $x^{T}Mx > 0$
($x^{T}Mx \geq 0$) for all nonzero vector $x$.}
\end{assumption}
The moving target objects will be investigated in Section \ref{sec:6}.
The item (ii) is assumed in order to avoid a meaningless problem
such that $g_{wo_i}=g_{wo_j}\ {\forall i,j}\in \V$.
Indeed, under the situation, it is straightforward to prove
convergence of the estimates to the common pose.
In terms of the item (iii), we see that
if $\mewoi\ewoj>0$ for all $i,j\in \V$,
then the following inequality holds.
\begin{eqnarray}
\phi(\mesi\eioi) \leq \phi_m := 
\max_{i,j\in \V}\phi(\mewoi\ewoj)\ {\forall i}\in \V
\label{ind_mean}
\end{eqnarray}
Inequality (\ref{ind_mean}) implies that
if $\mewoi\ewoj > 0 \ {\forall i,j}\in \V$
($\phi_m$ is smaller than $2$), 
then  (iii) is satisfied.
Thus, (iii) can be checked if set-valued prior 
information on the target orientations,
i.e. an upper bound of $\phi_m$ is available.

\subsection{Definition of Averaging Performance}
\label{sec:4.1}

In this subsection, we introduce a notion of 
approximate averaging.
For this purpose, we define the following sets for any positive parameter
$\varepsilon$.
\begin{eqnarray}
\hspace{-1cm}&&\Omega_p(\varepsilon):=\left\{(\bar{p}_{io_i})_{i\in \V}\left|\ 
\frac{1}{2}\sum_{i\in \V}\|\bar{p}_{io_i} - p^*_i\|^2\leq \varepsilon \rho_p\right.\right\},\
\rho_p := \frac{1}{2}\sum_{i\in \V}\|{p}_{io_i} - p^*_i\|^2
\label{eqn:ome_p}\\
\hspace{-1cm}&&\Omega_R(\varepsilon):=\left\{(\ebioi)_{i\in \V}\left|\ 
\sum_{i\in \V}\phi(\mesi\ebioi)\leq \varepsilon \rho_R\right.\right\},\
\rho_R := \sum_{i\in \V}\phi(\mesi\eioi)
\label{eqn:ome_R}
\end{eqnarray}

Let us now define $\varepsilon$-level 
averaging performance to be met by the estimates
$\bar{g}_{io_i} = (\bar{p}_{io_i}, \ebioi)$.
\begin{definition}
\label{def:1}
Given target poses $(g_{io_i})_{i\in \V}$,
position estimates $(\bar{p}_{io_i})_{i\in \V}$
are said to achieve $\varepsilon$-level 
averaging performance for a scalar $\varepsilon>0$
if there exists a finite $T$
such that 
$(\bar{p}_{io_i}(t))_{i\in \V} \in \Omega_p(\varepsilon)\ 
 {\forall t}\geq T$
and the orientation estimates $(\ebioi)_{i\in \V}$
are said to achieve $\varepsilon$-level 
averaging performance
if there exists a finite $T$ such that
$(\ebioi(t))_{i\in \V} \in \Omega_R(\varepsilon)\ {\forall t}\geq T$.
\end{definition}

In the absence of communication, each vision camera $i$
acquires no information on the target objects $o_j,\ j\neq i$.
Under the situation, what each vision camera can do is
to produce as an accurate estimate of the relative pose $g_{io_i}$ as possible.
Namely, the parameters $\rho_p$ and $\rho_R$ specify
the best performance of average estimation 
in the absence of communication.
More specifically, since the visual motion observer (\ref{eqn:visual_motion_observer})
correctly estimates the static target object pose
$g_{io_i}$ (Fact 1), the parameters $\rho_p$ and $\rho_R$ indicate 
the average estimation accuracy in the absence of the mutual feedback term of
$u_{ei}$ in (\ref{eqn:ce_update}).
Namely, the parameter $\varepsilon$ is an indicator 
of improvement of average estimation accuracy by inserting the mutual 
feedback term $k_s \sum_{j\in \N_i}E_R(\bar{g}_{io_i}^{-1}\bar{g}_{io_j})$.


\subsection{Auxiliary Results}
\label{sec:4.8}

In this subsection, we give some results 
necessary for proving the main result of this section.

\begin{lemma}
\label{lem_journal:1}
Suppose that the estimates $(\bar{g}_{io_i})_{i \in \V}$ are updated 
by the networked visual motion observer
(\ref{eqn:net_visual_motion_observer}).
Then, under Assumptions \ref{ass:1} and \ref{ass:2}
and $\mebioi\esi > 0\ 
{\forall t}\geq 0$, 
for all $c > 0$, there exists a finite $\tau(c)$ such that
$\phi(\mesi\ebioi) \leq \phi(\mes\ewol) + c \ {\forall t}\geq \tau(c),\
i \in \V$, 
where $h := \arg\max_{j\in \V}\phi(\mes\ewoj)$.
\end{lemma}
\begin{proof}
See Appendix \ref{app:2}.
\end{proof}
Lemma \ref{lem_journal:1} implies that the individual estimate $\ebioi$
gets closer to the average $\esi$ at least than the object
with the farthest orientation from the average.
In addition, the proof of this lemma also means that the set
\begin{equation*}
 {\mathcal S} = \{(\ebioi)_{i\in \V}|\ \mebioi\esi > 0\ {\forall 
i}\in \V\}
\end{equation*}
is positively invariant for the total system 
(\ref{eqn:net_visual_motion_observer})
under Assumption 2.
Namely, if $\mebioi\esi > 0$ is satisfied at the initial time,
then $\mebioi\esi > 0$ holds for all subsequent time. 

We next have the following lemma.
\begin{lemma}
\label{lem_journal:2}
Suppose that the estimates $(\bar{g}_{io_i})_{i \in \V}$ are updated 
by the networked visual motion observer
(\ref{eqn:net_visual_motion_observer}).
Then, under Assumptions \ref{ass:1} and \ref{ass:2},
if the initial estimates satisfy $(\ebioi(0))_{i\in \V} \in {\mathcal S}$,
both of the estimates $(\bar{p}_{io_i})_{i \in \V}$ 
and $(\ebioi)_{i\in \V}$ achieve
$1$-level averaging performance.
\end{lemma}
\begin{proof}
See Appendix \ref{app:3}.
\end{proof}
This lemma is proved by using the energy functions
\begin{eqnarray}
\ \ U_p := \frac{1}{2}\sum_{i \in \V}\|p_i^* - \bar{p}_{io_i}\|^2
= \frac{1}{2}\sum_{i \in \V}\|p^* - \bar{p}_{wo_i}\|^2, \
{U}_R := \sum_{i \in \V}\phi(\mesi\ebioi)
= \sum_{i \in \V}\phi(\mes\ebwoi)
\nonumber
\end{eqnarray}
which are defined by the sum of individual error
between the average and the estimate.
The functions $U_p\geq 0$ and $U_R\geq 0$ are equal to $0$ if and only if 
$\bar{p}_{io_i} = p_i^*$ and
$\ebioi = \esi\ {\forall i}\in \V$ respectively.
The selection of the energy function is inspired by
one of our previous works on pose synchronization \cite{TCST09}
whose framework is originally presented in \cite{CS_BK06}.

%
Lemma \ref{lem_journal:2} means that the average estimation 
as a group in the presence of communication
is at least more accurate than the case in the absence of communication.
However, this lemma does not say how accurate estimates of the average
the networked visual motion observer produces.

%
%

From Lemmas \ref{lem_journal:1} and \ref{lem_journal:2},
the estimates $(\bar{p}_{io_i})_{i\in \V}$ and 
$(\ebioi)_{i\in \V}$ settle into
$\Omega_p(1)$ and
${\mathcal S}^R_1 := {\mathcal S} \cap \Omega_R(1)$
in finite time, respectively.
Let us now define the following subsets of $\Omega_p(1)$ and ${\mathcal S}^R_1$.
\begin{eqnarray}
&&{\mathcal S}^p_2(k) := \Big\{(\bar{p}_{io_i})_{i\in \V}\in 
\Omega_p(1)\Big|\sum_{i\in \V}\sum_{j \in \N_i} 
\frac{1}{2}\|\bar{p}_{wo_i} - \bar{p}_{wo_j}\|^2
\geq k\rho_p\Big\},\
\nonumber\\
&&{\mathcal S}^R_2(k) := \Big\{(\ebioi)_{i\in \V}\in 
{\mathcal S}^R_1\Big|\beta \sum_{i\in \V}\sum_{j \in \N_i}
\phi(\mebwoi\ebwoj) 
\geq k\rho_R
\Big\},\
\nonumber\\
&&{\mathcal S}^p_3(k,\varepsilon) := 
\Omega_p(1)\setminus ({\mathcal S}^p_2(k) 
\cup \Omega_p(\varepsilon)),\
{\mathcal S}^R_3(k,\varepsilon) := 
{\mathcal S}^R_1\setminus ({\mathcal S}^R_2(k) 
\cup \Omega_R(\varepsilon))
\nonumber
\end{eqnarray}
for some $\varepsilon \in [0,1)$, 
where $\beta := 1 - \sqrt{2(\phi(\mes\ewol) + c)}$ and
$k = k_{e}/k_{s}$.
Images of the subsets on the position space are
depicted in Fig. \ref{set}.
We see from the figure that
\begin{eqnarray}
&&\hspace{-1cm}\Omega_p(1) \setminus 
({\mathcal S}^p_2(k)\cup {\mathcal S}^p_3(k,\varepsilon))\subseteq \Omega_p(\varepsilon),\
{\mathcal S}^R_1 \setminus 
({\mathcal S}^R_2(k)\cup {\mathcal S}^R_3(k,\varepsilon))\subseteq \Omega_R(\varepsilon).
\label{eqn:subset2}
\end{eqnarray}
\begin{figure}[t]
\begin{center}
\begin{minipage}{3.5cm}
\includegraphics[width=3.5cm]{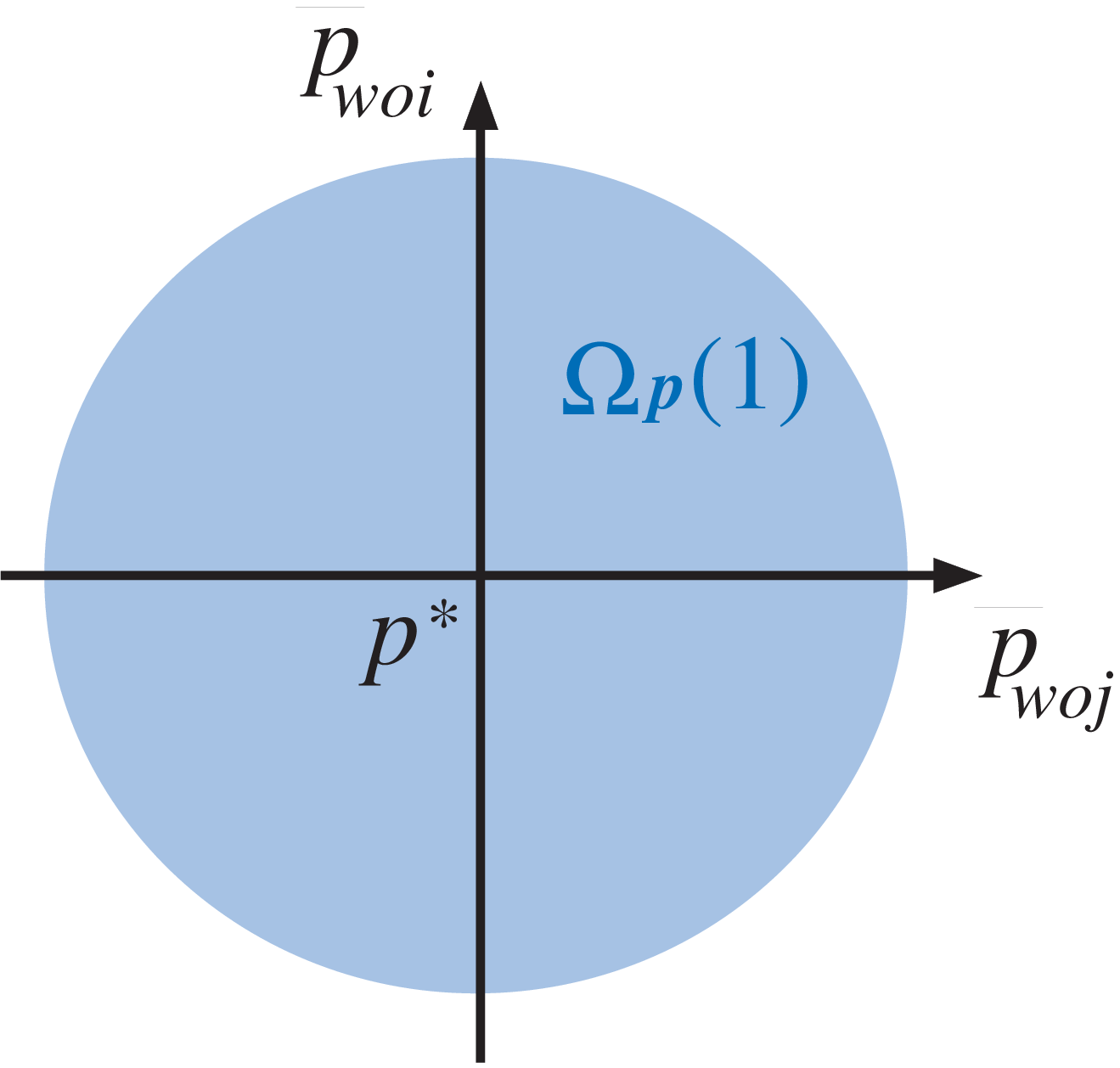}
\end{minipage}
\hspace{5mm}
\begin{minipage}{3.5cm}
\includegraphics[width=3.5cm]{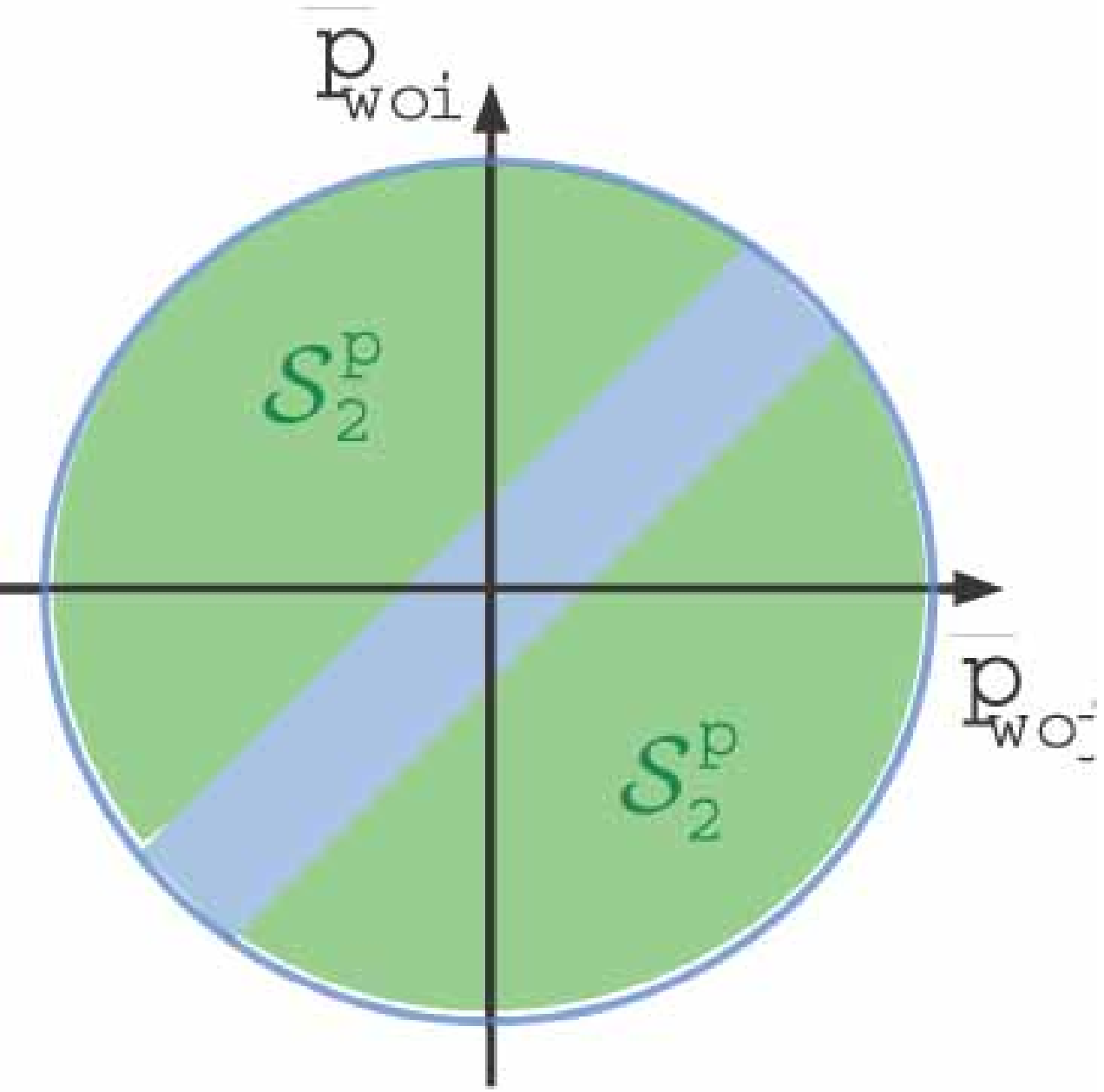}
\end{minipage}
\hspace{5mm}
\begin{minipage}{3.5cm}
\includegraphics[width=3.5cm]{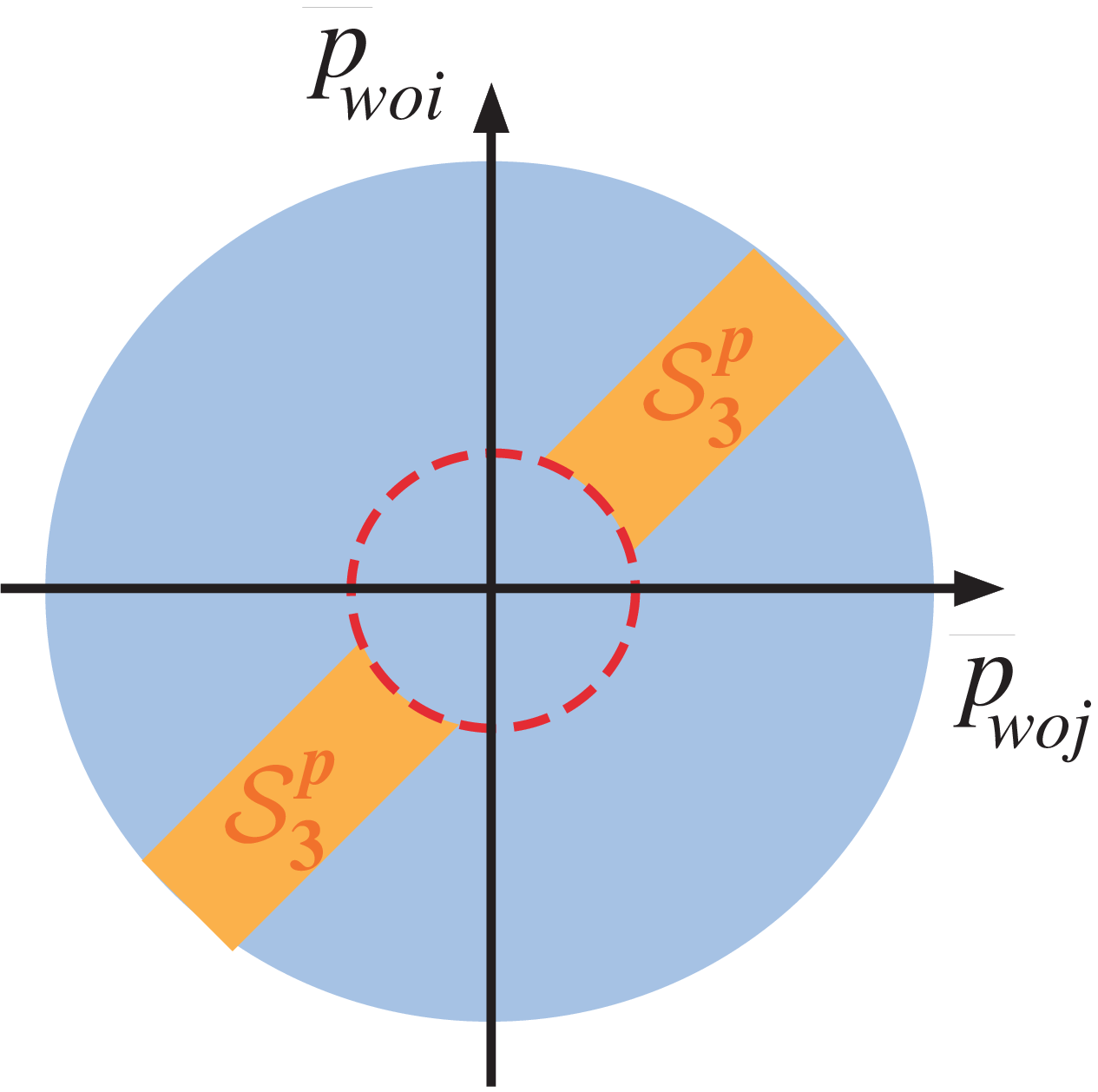}
\end{minipage}
\hspace{5mm}
\begin{minipage}{3.5cm}
\includegraphics[width=3.5cm]{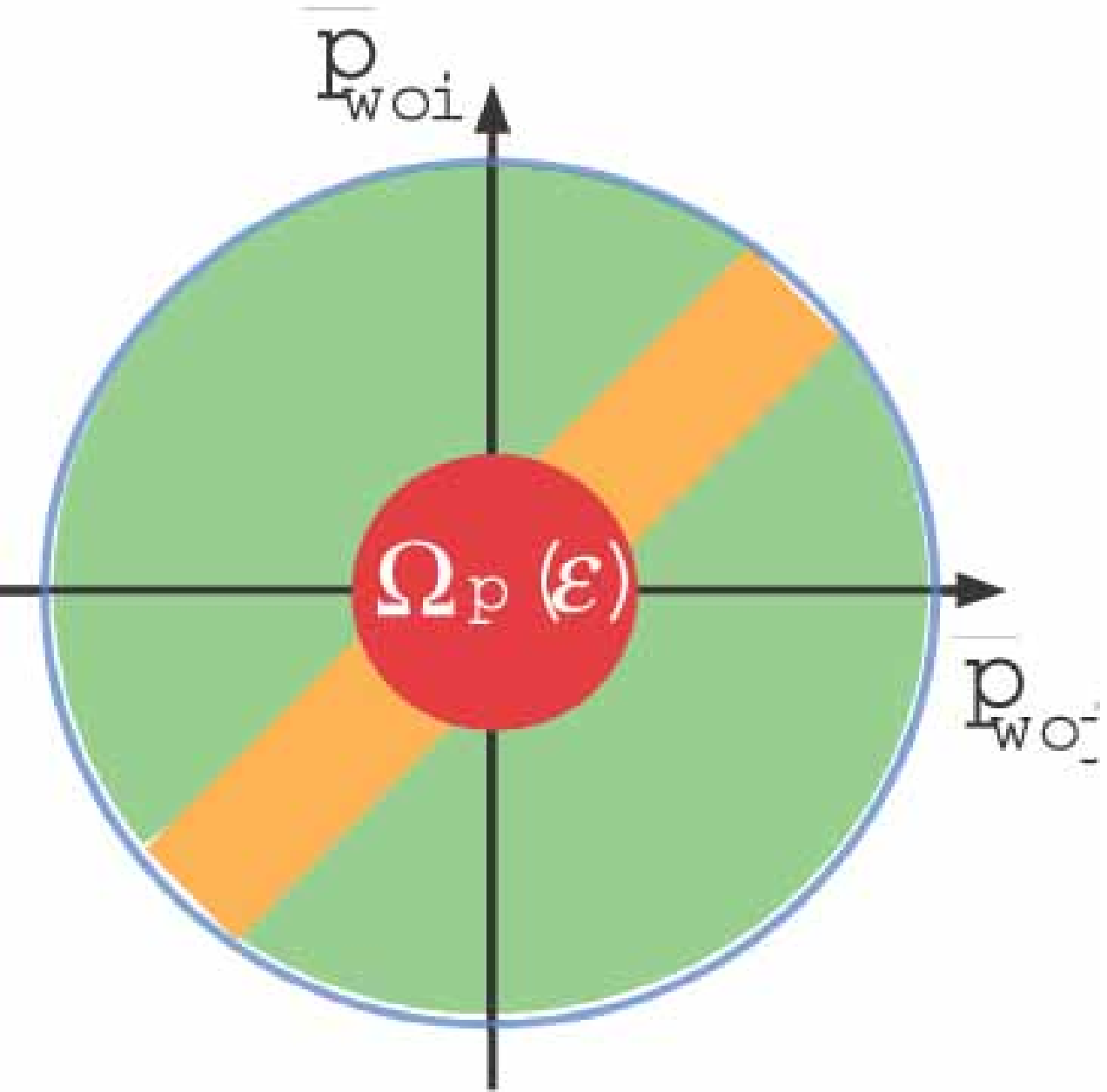}
\end{minipage}
\end{center}
\caption{Images of Each Subsets (Position)}
\label{set}
\end{figure}
In terms of the subsets ${\mathcal S}^p_2(k)$ and 
${\mathcal S}^R_2(k)$, we have the following lemma.
\begin{lemma}
\label{lem_journal:3}
Suppose that all the assumptions in Lemma \ref{lem_journal:2} hold
and $\beta > 0$.
Then, the time derivative of $U_p$ and $U_R$ along with 
the trajectories of (\ref{eqn:net_visual_motion_observer})
are strictly negative
as long as $(\bar{p}_{io_i})_{i\in \V}\in {\mathcal S}^p_2(k)$
and $(\ebioi)_{i\in \V}\in {\mathcal S}^R_2(k)$
respectively,
at least after the time $\tau(c)$.
\end{lemma}
\begin{proof}
See Appendix \ref{app:4}.
\end{proof}
From (\ref{ind_mean}), $\beta$ can be estimated
by set-valued prior information on
the target orientations i.e. $\phi_m$.


\subsection{Averaging Performance}
\label{sec:4.2}

%
We are now ready to state the main result of this section
on averaging accuracy attained by the networked visual motion observer
(\ref{eqn:net_visual_motion_observer}).
\begin{theorem}
\label{thm_journal:1}
Suppose that all the assumptions in Lemma \ref{lem_journal:2} hold.
Then, for any $\epsilon \in (0, 1)$,
position estimates $(\bar{p}_{io_i})_{i\in \V}$ achieve $\varepsilon_p$-level averaging performance
with 
\begin{eqnarray}
\varepsilon_p
= \left\{
\begin{array}{ll}
1 - (1-\epsilon)\left(
1 - \sqrt{kW}\right)^2 &
\mbox{if }k\leq 1/W\\
1&\mbox{otherwise}
\end{array}
\right.,
\label{16a}
\end{eqnarray}
and orientation estimates $(\ebioi)_{i\in \V}$ achieve 
$\varepsilon_R$-level averaging performance
with 
\begin{eqnarray}
\varepsilon_R
= \left\{
\begin{array}{ll}
1 - (1-\epsilon)\left(
\sqrt{\beta} - \sqrt{kW}\right)^2, 
&\mbox{if }
k\leq \beta/W,\ \beta > 0\\
1,&\mbox{otherwise}
\end{array}
\right.,
\label{16b}
\end{eqnarray}
where $W$ is defined in (\ref{eqn:defD}).
\end{theorem}
\begin{proof}
See Appendix \ref{app:6}.
\end{proof}

Suppose that $\epsilon$ is taken sufficiently close to $0$.
Then, we see that both of the parameters
$\varepsilon_p$ and $\varepsilon_R$ become small
as the term $\sqrt{kW}$ approaches to $0$.
Note that if we use a sufficiently small $k$ ($k_s\gg k_e$) 
in (\ref{eqn:ce_update}),
the term is approximated by $0$.
Here, we see an essential difference between
the position and orientation estimates.
The definition of $\varepsilon_p$ with $\epsilon \approx 1$
indicates that we can get arbitrarily accurate
estimation of the average $p_i^*$ by choosing a sufficiently small $k$.
In contrast, we see from the definition of $\varepsilon_R$ that
an offset associated with $\sqrt{\beta}(< 1)$ occurs
for the orientation estimates regardless of the parameter $k$.
From the definition of $\beta := 1 - \sqrt{2(\phi(\mes\ewol) + c)}$, if
the target object's orientation $\ewol$ 
is sufficiently close to the average $\mes$, i.e.
if $\ewoi$ and $\ewoj$ are close among all $i,j\in \V$ enough to
approximate all the orientations by  
matrices on a tangent vector space of $SO(3)$ at $\ewoi$,
then it becomes close to $0$ and the average is accurately estimated 
by the networked visual motion observer (\ref{eqn:net_visual_motion_observer}).
Otherwise, the accuracy might degrade, though it is more accurate 
at least than the case in the absence of communication.

\section{Tracking Performance Analysis}
\label{sec:6}

In this section, we analyze the tracking performance
of the estimates $\{\bar{g}_{io_i}\}_{i\in \V}$ to the average $g_i^*$
for moving targets when the networked visual motion observer is applied 
to the visual sensor networks under the following assumption.
\begin{assumption}
\label{ass:3}\mbox{}\\
(i) The target body velocities $V^b_{wo_i}(t),\ i\in \V$ are continuous in 
	   $t$ and bounded as
\begin{equation}
\|v^b_{wo_i}(t)\|_2^2\leq \bar{w}_p^2,\
\|\omega^b_{wo_i}(t)\|^2 \leq \bar{w}_R^2\ {\forall i\in \V},\ t\geq 0.
\label{dist}
\end{equation}
(ii) For all $t\geq 0$, there exists $(i(t),j(t)) \in \V\times \V$ such that
$p_{wo_{i(t)}}\neq p_{wo_{j(t)}}$ and 
$e^{\hat{\xi}\theta_{wo_{i(t)}}}\neq e^{\hat{\xi}\theta_{wo_{j(t)}}}$.\\
(iii)
$\mewoj(t)\ewoi(t) > 0$ for all $i,j \in \V$ and $t\geq 0$.
\end{assumption}

\subsection{Description of Average Motion}
\label{sec:6.1}

In this subsection, we first formulate the motion of
the average $g^* = (p^*,\es)$.
The behavior of the position average $p^*$ is
clearly described by 
\begin{equation}
\dot{p}^* = \es v^{b,*},\ v^{b,*} := \mes\left(\frac{1}{n}\sum_{i\in \V} \ewoi v^b_{woi}\right)
\label{pesb}
\end{equation}
from the definition of $p^* = \frac{1}{n}\sum_{i\in \V}p_{wo_i}$.
Meanwhile, the trajectory of the orientation average $\es$ 
described by (\ref{proj}) satisfies the following lemma.
\begin{lemma}
\label{lem:9}
Under Assumption \ref{ass:3}, the average $\es$ is continuously differentiable.
\end{lemma}
\begin{proof}
From the polar decomposition, we get $S(t) = \es(t) P_S(t)$ \cite{M_SIAM02},
where $S(t) = \frac{1}{n}\sum_{i\in \V}\ewoi$ and $P_S^2(t) = S^T(t)S(t)$.
Under Assumption \ref{ass:3}(iii), 
we have $\mewoj(t) S(t) > 0$
and hence $P_S(t)$ is invertible
for all ${t}\geq 0$.
Thus, the average $\es$ is given by $\es(t) = S(t)P_S^{-1}(t)$.
From (\ref{eqn:RRBM}), the matrices $S(t)$ and $P_S(t)$
are clearly differentiable from their definitions and hence 
$\des$ is well defined.
Moreover, from Assumption \ref{ass:3}(i),
both of $\dot{S}(t)$ and $\dot{P}_S(t)$ are continuous and 
$\frac{d}{dt}(P_S^{-1}) = P_S^{-1}(t)\dot{P}_SP_S^{-1}$ is also continuous,
which implies that $\des(t)$ is also continuous.
Hence, the average $\es$ is continuously differentiable.
This completes the proof.
\end{proof}

Moreover, since $\es(t) \in SO(3)$ holds for all $t\geq 0$,
the derivative ${\des}$ has to satisfy
$\des \in T_{\es}SO(3)$, 
where $T_{\es}SO(3) := \{\es X |\ X \in so(3)\}$ is the tangent
space of the manifold $SO(3)$ at $\es$.
Namely, the trajectory of the Euclidean mean $\es$ is described by
the differential equation
$\des = \es \hat{\omega}^{b,*}$
with some body velocity $\hat{\omega}^{b,*}\in so(3)$.

We next clarify a relation between
velocities $V^{b,*} := (v^{b,*}, \omega^{b,*})$ 
 and $V^b_{wo_i} = 
(v^b_{wo_i}, \omega^b_{wo_i}),\ {i\in \V}$.
We first define
$w_p := (v^b_{wo_i})_{i\in \V}$ and $w_R := (\omega^b_{wo_i})_{i\in \V}$.
Since it is easy from (\ref{pesb}) to obtain 
$\|v^{b,*}\|^2 \leq \|w_p\|^2/n$,
we mention only a relation between $\omega^{b,*}$ and $w_R$ in the following.
\begin{lemma}
\label{lem:8}
Suppose that the target orientations $(\ewoi)_{i\in \V}$ satisfy
\begin{equation}
\left\|\es(t) - S(t)\right\|_F\leq \gamma\ 
 {\forall t}\geq 0
\label{ac2}
\end{equation}
for some $\gamma > 0$.
Then, the following inequality holds.
\begin{equation}
 \|\omega^{b,*}(t)\|^2 < \frac{\mu^2(\gamma)}{n}\|w_R(t)\|^2,\
\mu(\gamma) := \frac{\sqrt{2}}{\sqrt{2}-\gamma}
\label{ine}
\end{equation}
\end{lemma}
\begin{proof}
See Appendix \ref{app:8}
\end{proof}
Though we omit the proof, $\left\|\es(t) - S(t)\right\|_F$
is also upper bounded by $\phi_m$ and hence
$\gamma$ is estimated by prior information on the target orientations.

\subsection{Tracking Performance}
\label{sec:6.2}

Let us consider the whole networked system $\Sigma_{track}$ 
consisting of the relative rigid body motion
(\ref{eqn:RRBM}) for all $i\in \V$ and the networked visual motion observer 
(\ref{eqn:net_visual_motion_observer}).
Here, we regard the collections of body velocities of
the target objects $(V^b_{wo_i})_{i\in \V}$,
i.e. $w = (w_p,w_R)$, as the 
external disturbance to $\Sigma_{track}$ and evaluate the 
error between the estimates $(\bar{g}_{io_i})_{i\in \V}$ and the average $g^*_i$
in the presence of the disturbance $w$.
Namely, we let the error $(\{g_i^*\}^{-1}\bar{g}_{io_i})_{i\in \V}$
be the output signal of $\Sigma_{track}$.

Unlike the static objects case, 
$\rho_p = \frac{1}{2}\sum_{i\in \V}\|p_{io_i} - p^*_i\|^2$ and
$\rho_R = \sum_{i\in \V}\phi(\mesi\eioi)$  
are also time-varying.
We thus define the parameters
\begin{eqnarray}
\rho'_p := \sup_{t} \rho_p(t),\
\rho'_R := \sup_{t} \rho_R(t)
\nonumber
\end{eqnarray}
assuming $\rho_p' < \infty$ and redefine the sets $\Omega_p'$
and $\Omega_R'$ by just replacing $\rho_p$ and $\rho_R$
in (\ref{eqn:ome_p}) and (\ref{eqn:ome_R}) by
$\rho_p'$ and $\rho_R'$, respectively.
The parameters $\rho'_p$ and $\rho'_R$ are the suprimum of the distance from
the estimate to the average when $g_{io_i}$ is correctly
estimated and hence they are also indicators of
the best average estimation performance in the absence of communication.
Note however that the visual motion observer 
(\ref{eqn:visual_motion_observer}) cannot correctly estimate
$g_{io_i}$ as long as the object is moving with unknown velocity.

The problem to be considered here is
redefined as follows.
\begin{definition}
\label{def:3}
The position estimates $(\bar{p}_{io_i})_{i\in \V}$
are said to achieve $\varepsilon$-level 
tracking performance for a positive scalar $\varepsilon$
if there exists a finite $T$
such that 
$(\bar{p}_{io_i}(t))_{i\in \V} \in \Omega_p'(\varepsilon)\ 
{\forall w}\in {\mathcal W} \mbox{ and }
t \geq T$,
where ${\mathcal W}$ is the set of the
disturbance signal $w(\cdot)$
satisfying Assumption \ref{ass:3}.
Similarly, the estimates $(\ebioi)_{i\in \V}$
are said to achieve $\varepsilon$-level 
tracking performance
if there exists a finite $T$ such that
$(\ebioi(t))_{i\in \V} \in \Omega_R'(\varepsilon)\ 
{\forall w}\in {\mathcal W} \mbox{ and } {t}\geq T$.
\end{definition}

In terms of the tracking performance defined above, we have the following theorem.
\begin{theorem}
\label{thm_journal:3}
Suppose that the estimates $\bar{g}_{io_i}$ are updated 
according to 
(\ref{eqn:net_visual_motion_observer}).
Under Assumptions 1 and 3, 
if $k_e > \mu^2(\gamma)$ for $\gamma$ satisfying (\ref{ac2})
and $(\ebioi(t))_{i\in \V}\in {\mathcal S}\ 
 {\forall t}\geq 0$, 
the estimates $(\bar{p}_{io_i})_{i\in \V}$
and $(\ebioi)_{i\in \V}$ 
achieve $\varepsilon_p'$ and $\varepsilon_R'$-level
tracking performances respectively with
\begin{eqnarray}
\varepsilon_p' := 1 + \frac{1}{k_e-1} + \frac{\bar{w}_p^2}{\rho'_p(k_e-1)},\
\varepsilon_R' := 1 + \frac{\mu^2(\gamma)}{k_e-\mu^2(\gamma)} + \frac{\bar{w}_R^2}{\rho'_R(k_e-\mu^2(\gamma))}.
\nonumber
\end{eqnarray}
\end{theorem}
\begin{proof}
See Appendix \ref{app:20}.
\end{proof}
This theorem implies that the networked visual motion observer works
even for moving target objects.
We also see that the ultimate error between the estimates
and the average gets small
as the visual feedback gain $k_e$ becomes large, which is a natural conclusion
from the form of (\ref{eqn:ce_update}).

In summary, we have the following conclusion on the gain selection.
In order to achieve a good averaging performance, 
we should make the mutual feedback gain $k_s$ large
relative to the visual feedback gain $k_e$.
In order to achieve a good tracking performance,
the visual feedback gain $k_e$ should be absolutely large.
Namely, the best selection is to make both
gains $k_e$ and $k_s$ large while the mutual feedback gain $k_s$ is much larger than
the visual feedback gain $k_e$.
However, the size of $k_s$ is in general restricted by
the communication rate due to limitation in
standard feedback control theory.
Then, a trade-off occurs between 
averaging and tracking performances, i.e.
if we set a large $k_e$, a good tracking performance
is achieved at the cost of a poor averaging performance
and vice versa.

\section{Simulation}
\label{sec:7}

\begin{figure}[t]
\begin{center}
\begin{minipage}{5.5cm}
\begin{center}
\includegraphics[width=5.5cm]{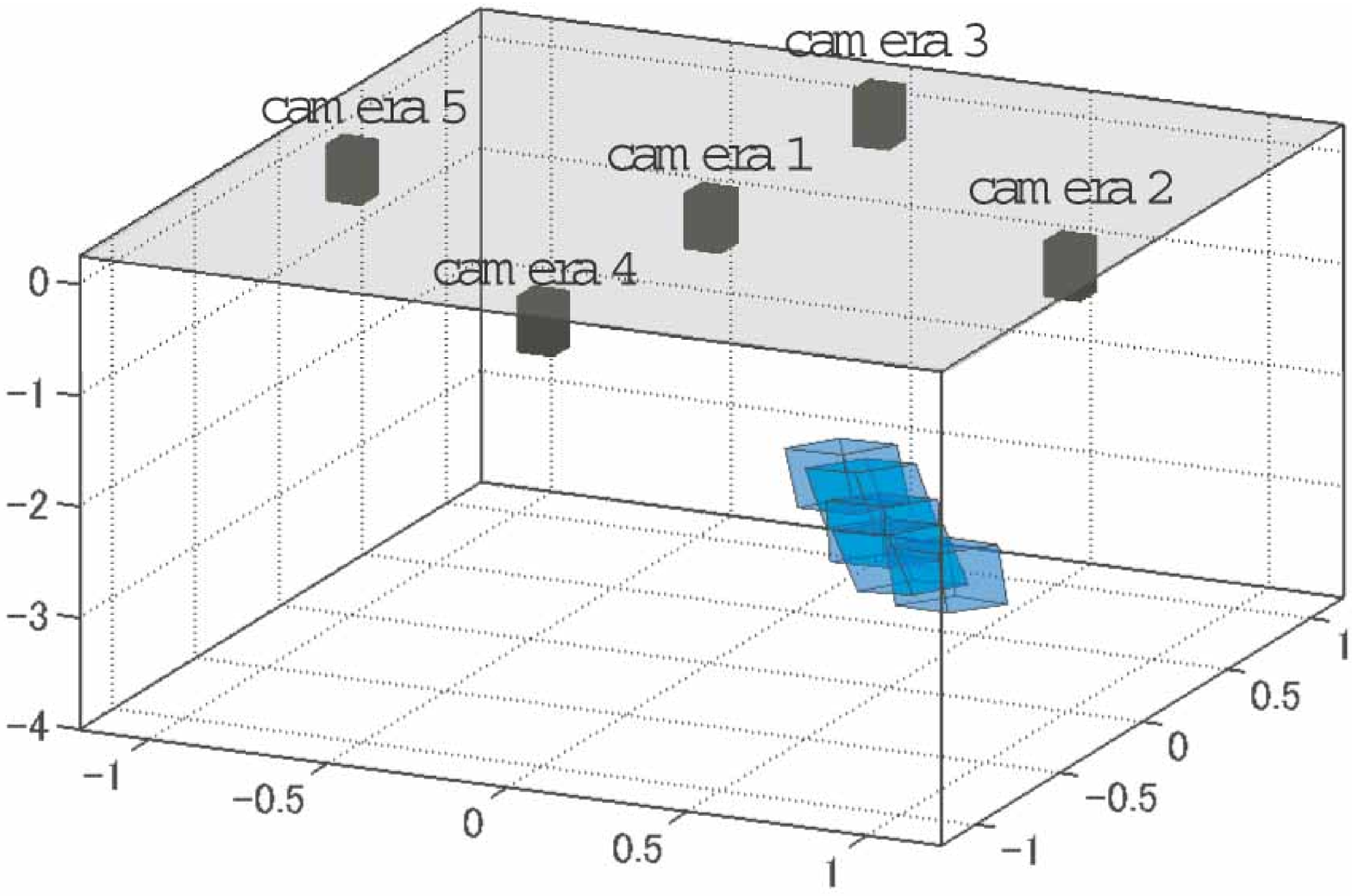}
\caption{Overview}
\label{fig:ov}
\end{center}
\end{minipage}
\hspace{1cm}
\begin{minipage}{3.5cm}
\begin{center}
\includegraphics[width=3.5cm]{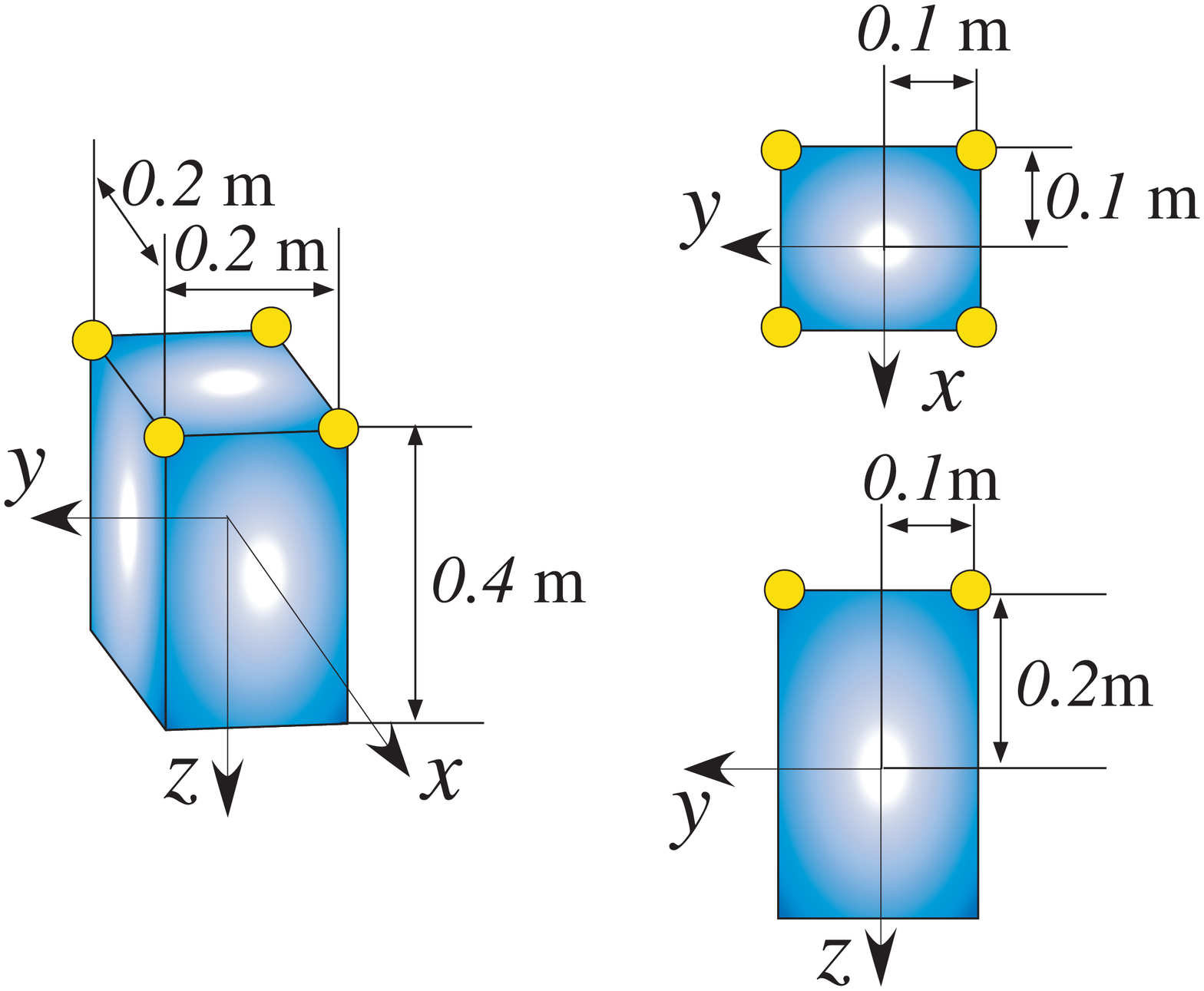}
\caption{Feature Points}
\label{fig:features}
\end{center}
\end{minipage}
\hspace{1cm}
\begin{minipage}{4.5cm}
\begin{center}
\includegraphics[width=4.5cm]{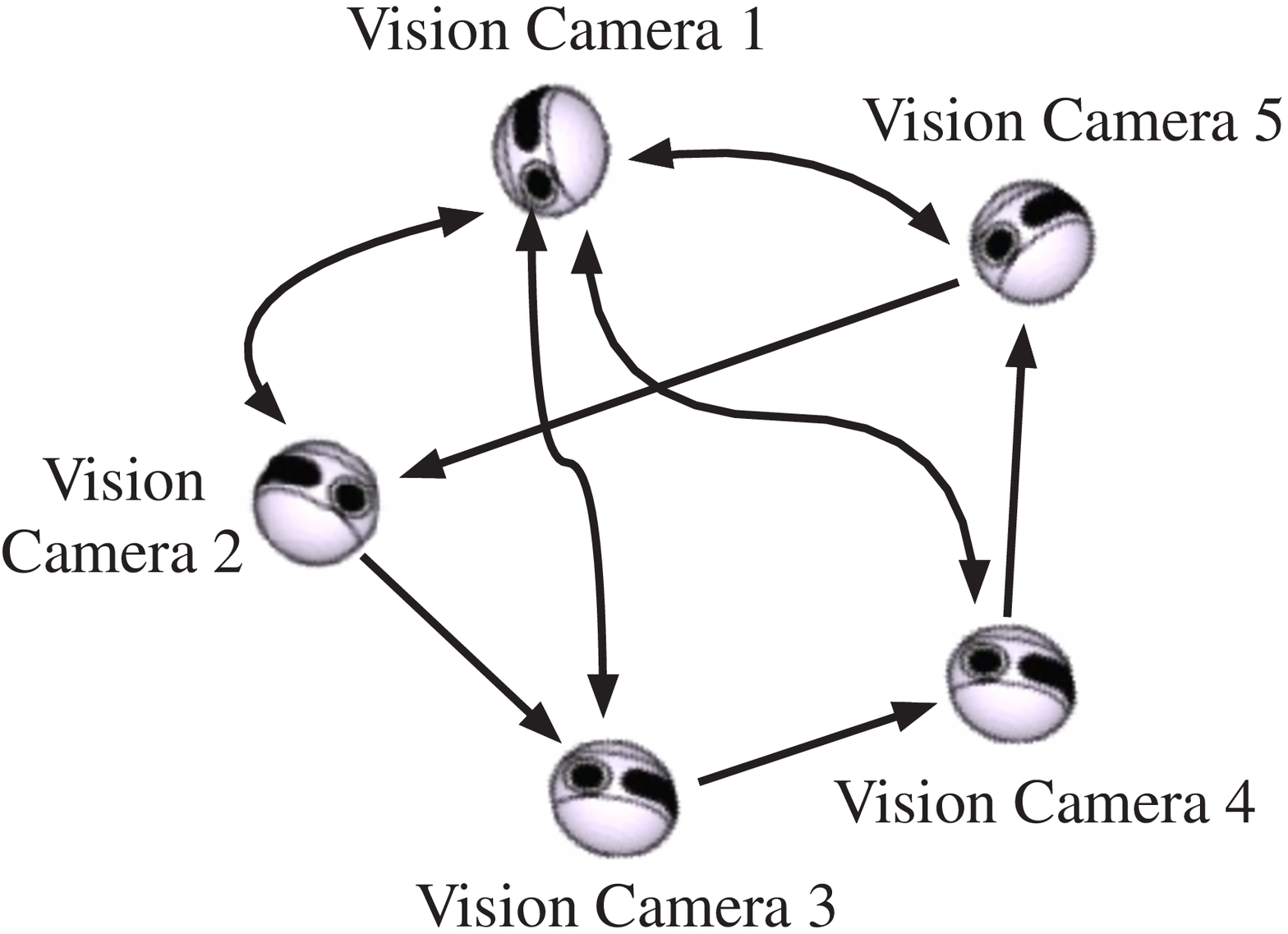}
\caption{Communication Graph}
\label{fig:graph1}
\end{center}
\end{minipage}
\end{center}
\end{figure}
We finally demonstrate the effectiveness of the networked visual motion observer
and validity of the theoretical results through simulation.
Throughout this section, we consider the situation where
five pin-hole type vision cameras with focal length $0.01$[m]
see a group of target objects.
We identify the frame of camera 1 with the world frame
and let 
$p_{w2} = [
1\
0\
0
]^T,\
p_{w3} = [
0\
1\
0
]^T,\
p_{w4} = [
0\
-1\
0
]^T,\
p_{w5} = [
-1\
0\
0
]^T$,\
and $e^{\hat{\xi}\theta_{wi}} = I_3,\ {\forall i}\in \{2,3,4,5\}$.
The overview of the setting is illustrated in Fig. \ref{fig:ov},
where blue boxes represent the initial configuration of target objects
with 
$p_{wo_1}=[
0.12\
0.55\
-2.78
],\
p_{wo_2}=
[
0.22\
0.48\
-2.85
],\
p_{wo_3}=
[
0.33\
0.33\
-2.97
],\
p_{wo_4}=
[
0.42\
0.23\
-3.08
],\
p_{wo_5}=
[
0.56\
0.12\
-3.15
]$
and
$
\xi\theta_{wo_1}=
[
   -0.30\
   -0.30\
   -0.30
],\
\xi\theta_{wo_2}=
[
   -0.30\
   -0.40\
   -0.40
],\
\xi\theta_{wo_3}=
[
   -0.40\
   -0.30\
   -0.30
],\
\xi\theta_{wo_4}=
[
   -0.30\
   -0.40\
   -0.30
],\
\xi\theta_{wo_5}=
[
   -0.30\
   -0.30\
   -0.40
].$
All the targets have four feature points
whose positions relative to the object frame 
are illustrated in Fig. \ref{fig:features}.
We use the points projected onto the image plane
as visual measurements $f_i$.
The communication structure is depicted in 
Fig. \ref{fig:graph1} with $W = 1$.

\begin{figure}[th]
\begin{center}
\begin{minipage}{4.5cm}
\begin{center}
\includegraphics[width=4.5cm]{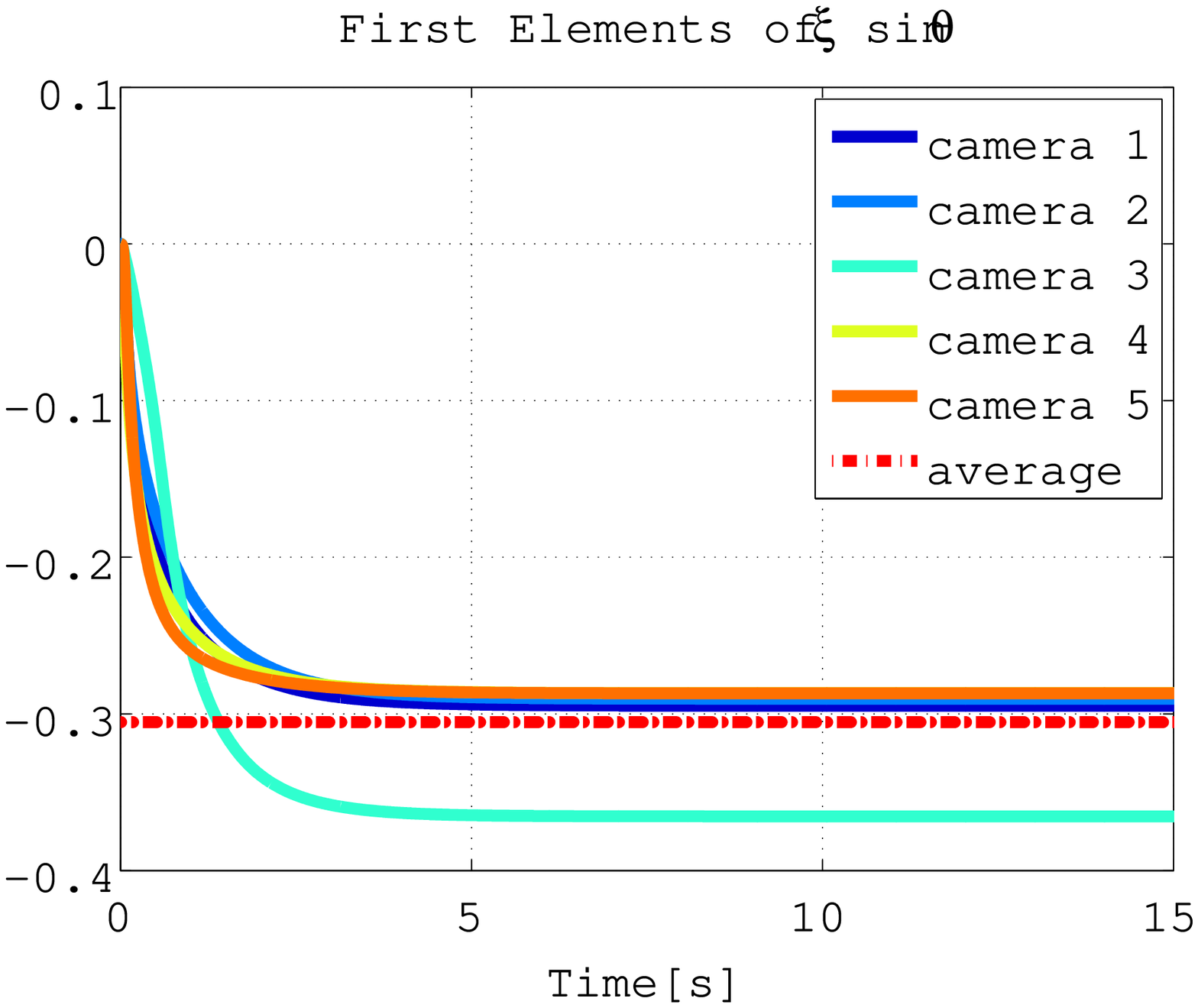}
\end{center}
\end{minipage}
\hspace{5mm}
\begin{minipage}{4.5cm}
\begin{center}
\includegraphics[width=4.5cm]{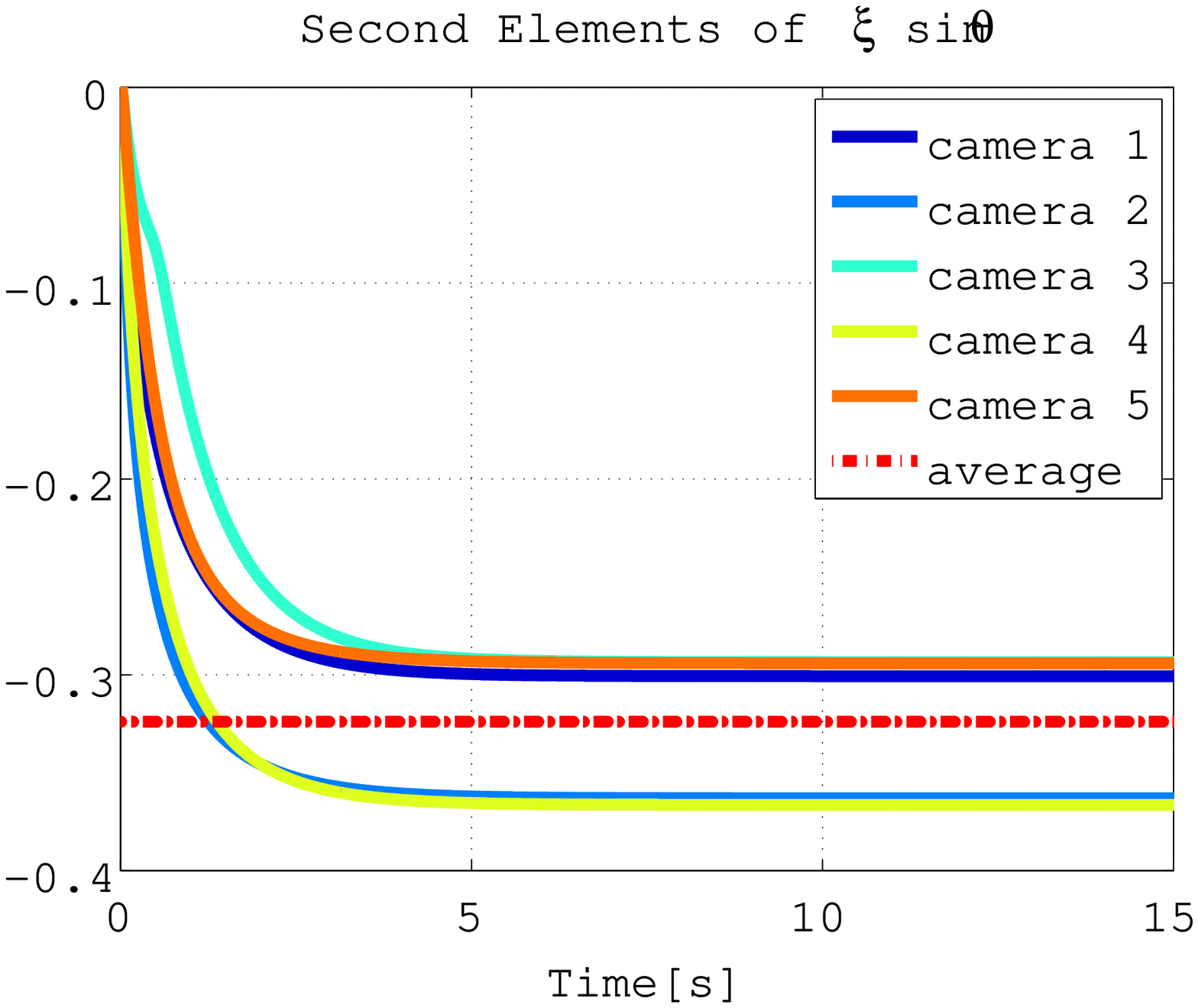}
\end{center}
\end{minipage}
\hspace{5mm}
\begin{minipage}{4.5cm}
\begin{center}
\includegraphics[width=4.5cm]{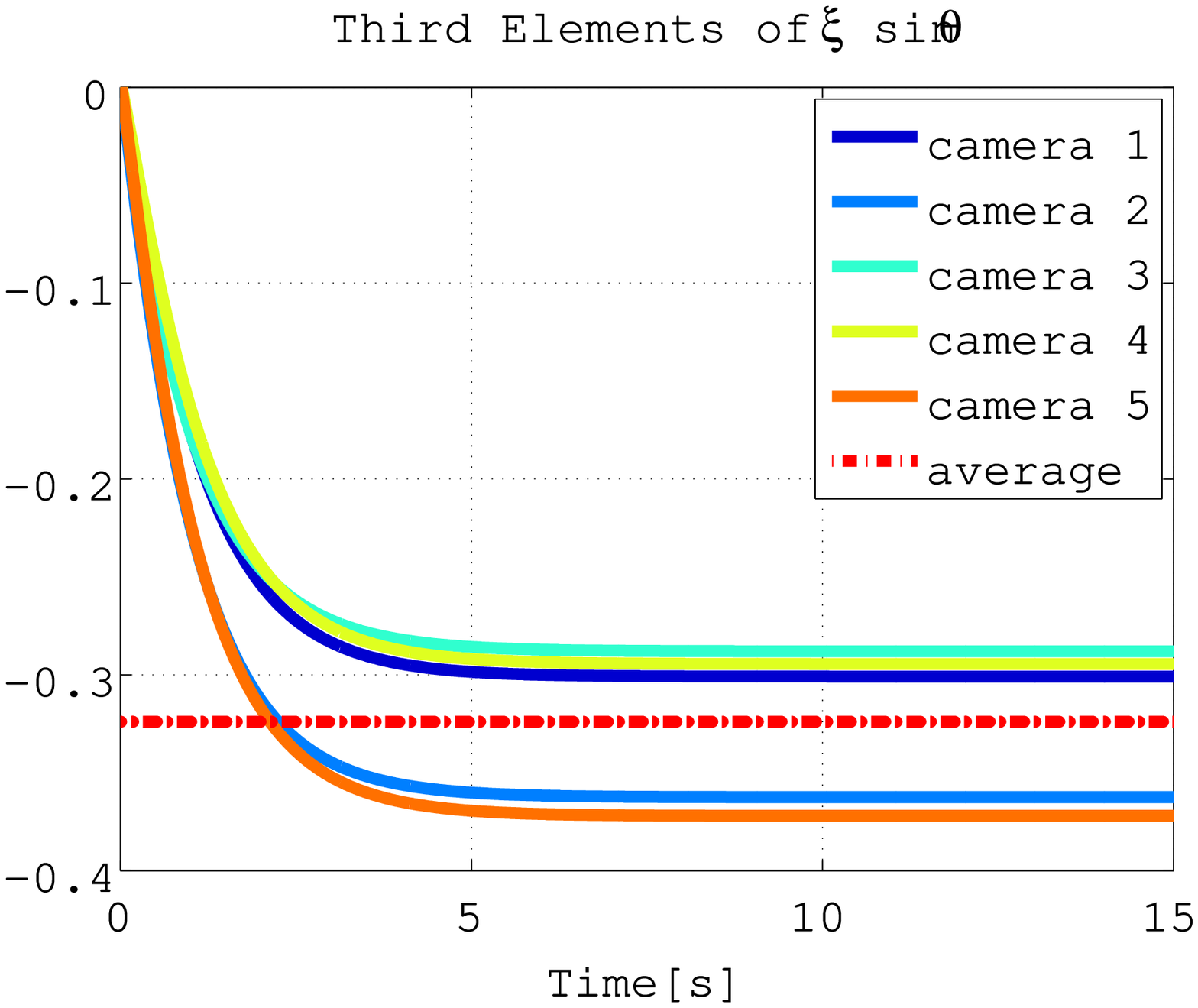}
\end{center}
\end{minipage}
\caption{Time Responses of Each Element of $\bar{\xi}\sin\bar{\theta},\ i=1,\cdots,5$ 
 (Static: $k_s = 0.1$)}
\label{fig:res1_2}
\end{center}
\begin{center}
\begin{minipage}{6cm}
\begin{center}
\includegraphics[width=6cm]{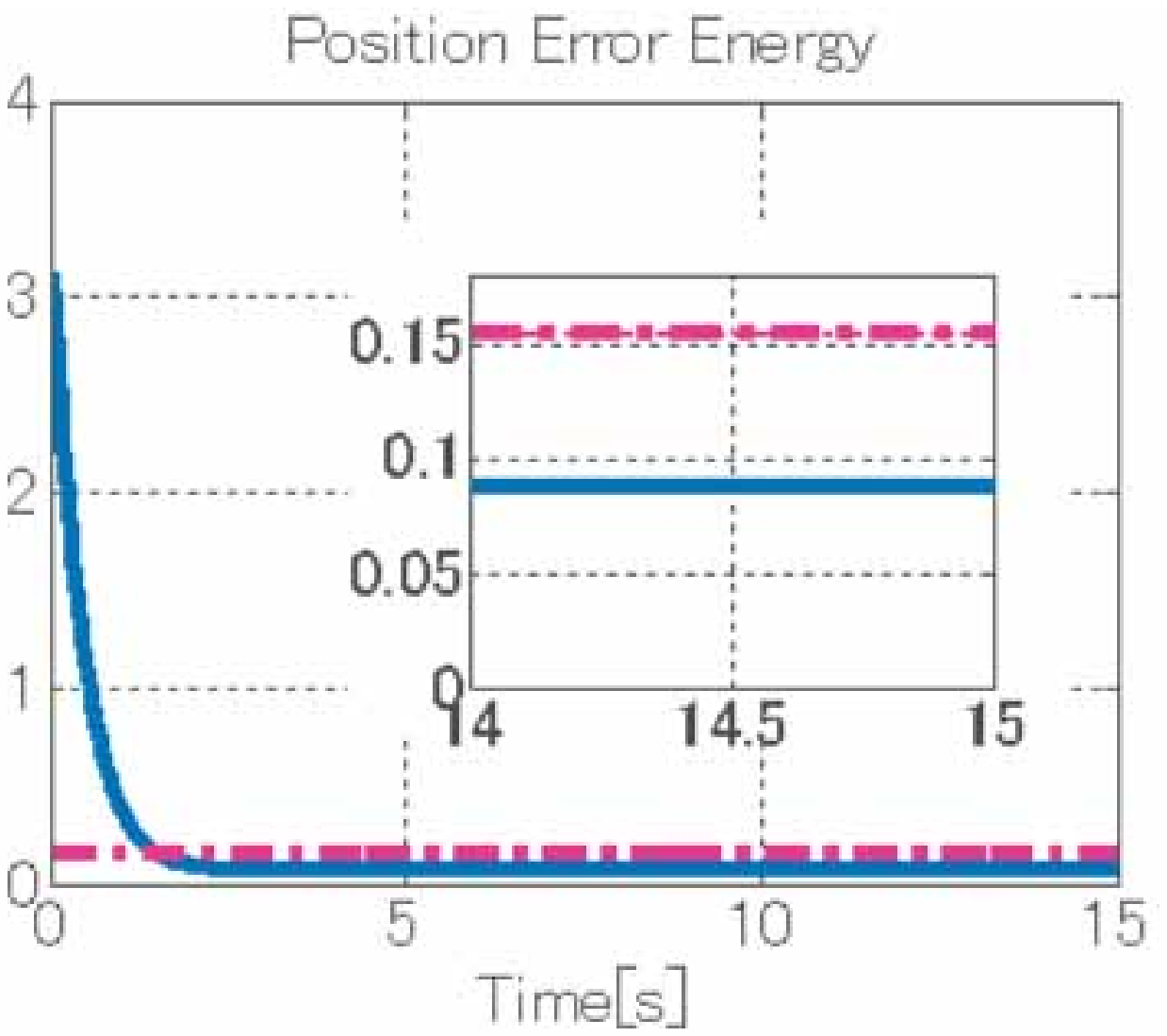}
\end{center}
\end{minipage}
\hspace{1.5cm}
\begin{minipage}{6cm}
\begin{center}
\includegraphics[width=6cm]{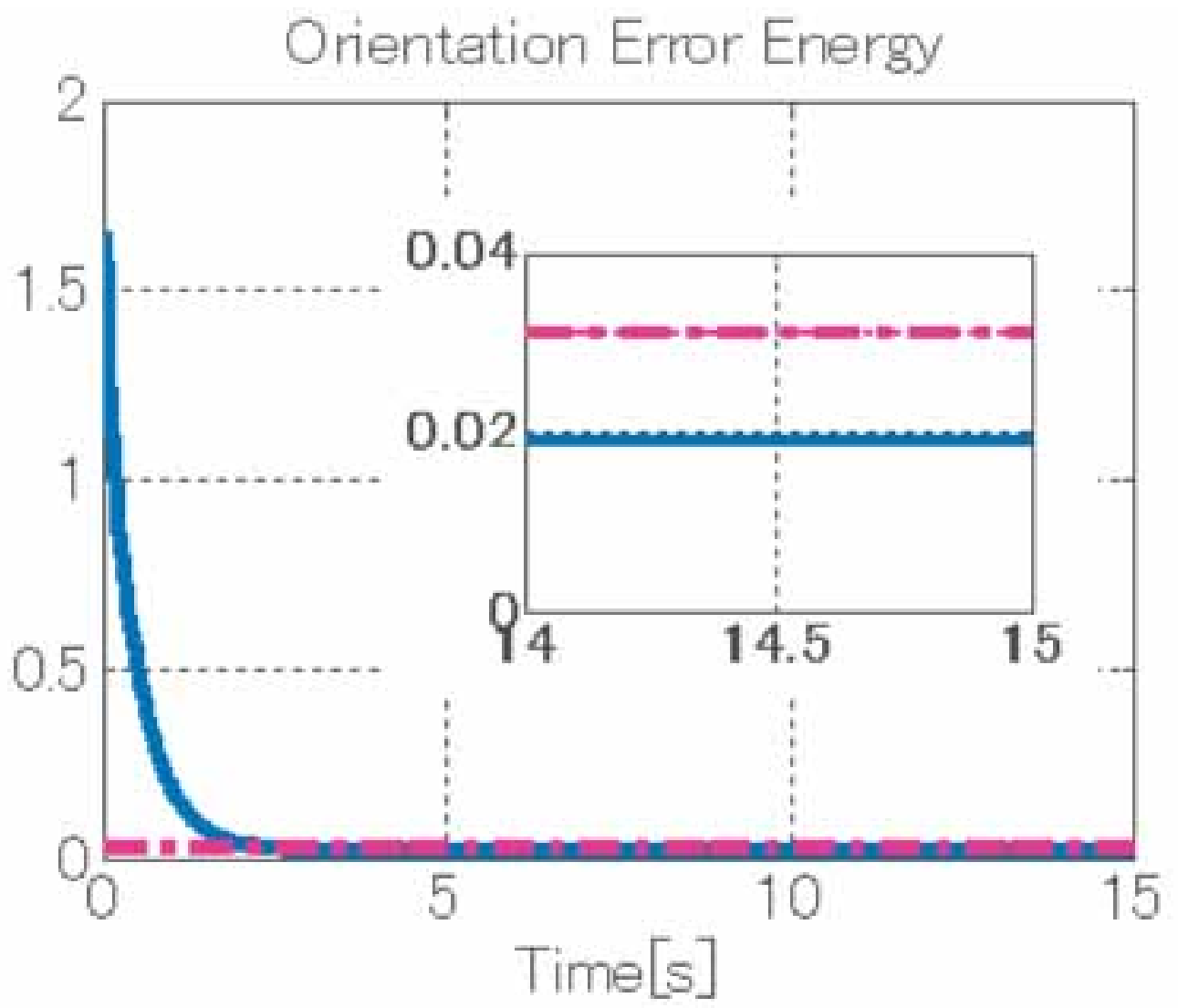}
\end{center}
\end{minipage}
\caption{Time Response of $U_p$ (Left) and $U_R$ (Right) (Static: $k_s = 0.1$)}
\label{fig:res1_8}

\end{center}
\begin{center}
\begin{minipage}{6cm}
\begin{center}
\includegraphics[width=6cm]{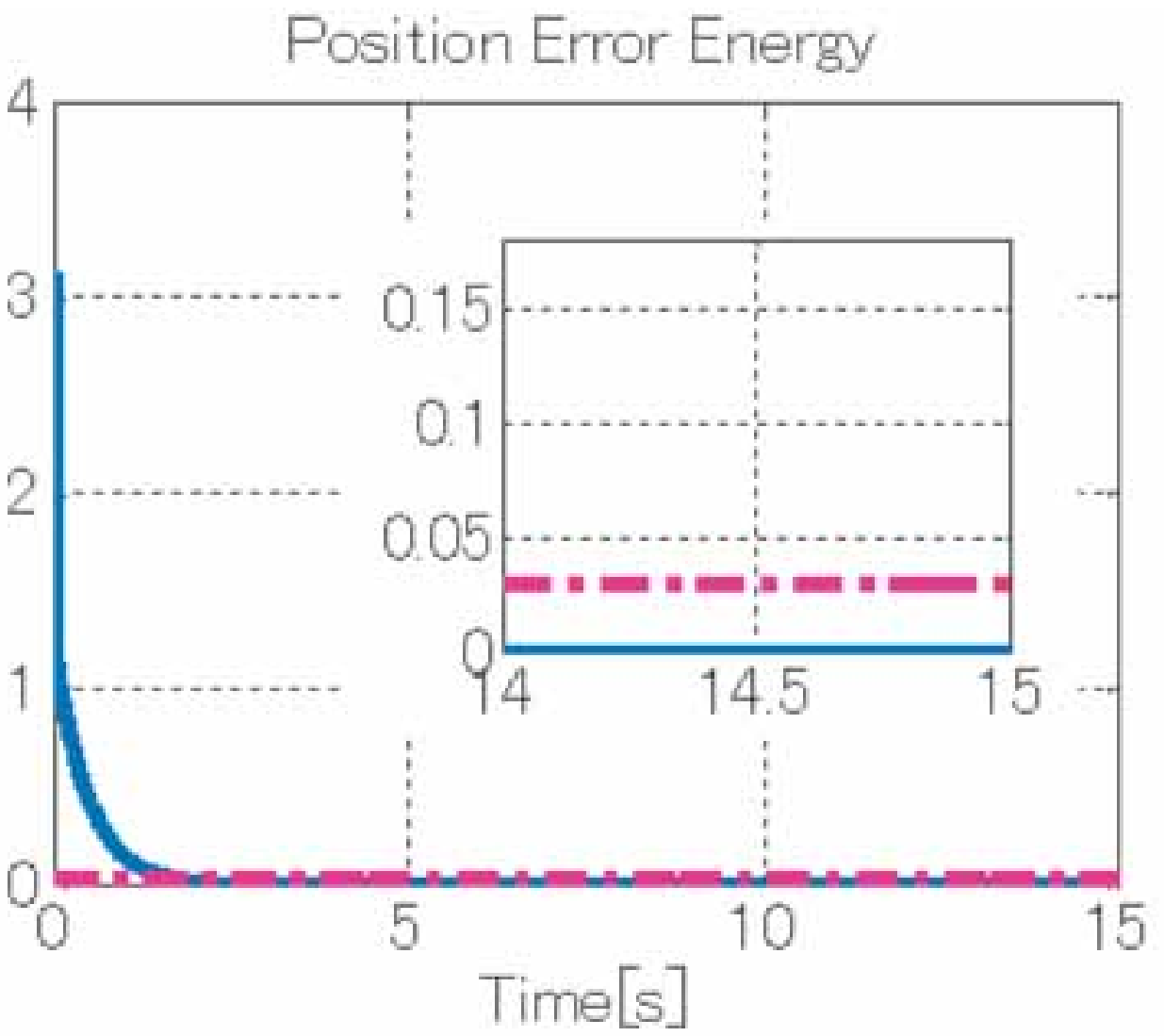}
\end{center}
\end{minipage}
\hspace{1.5cm}
\begin{minipage}{6cm}
\begin{center}
\includegraphics[width=6cm]{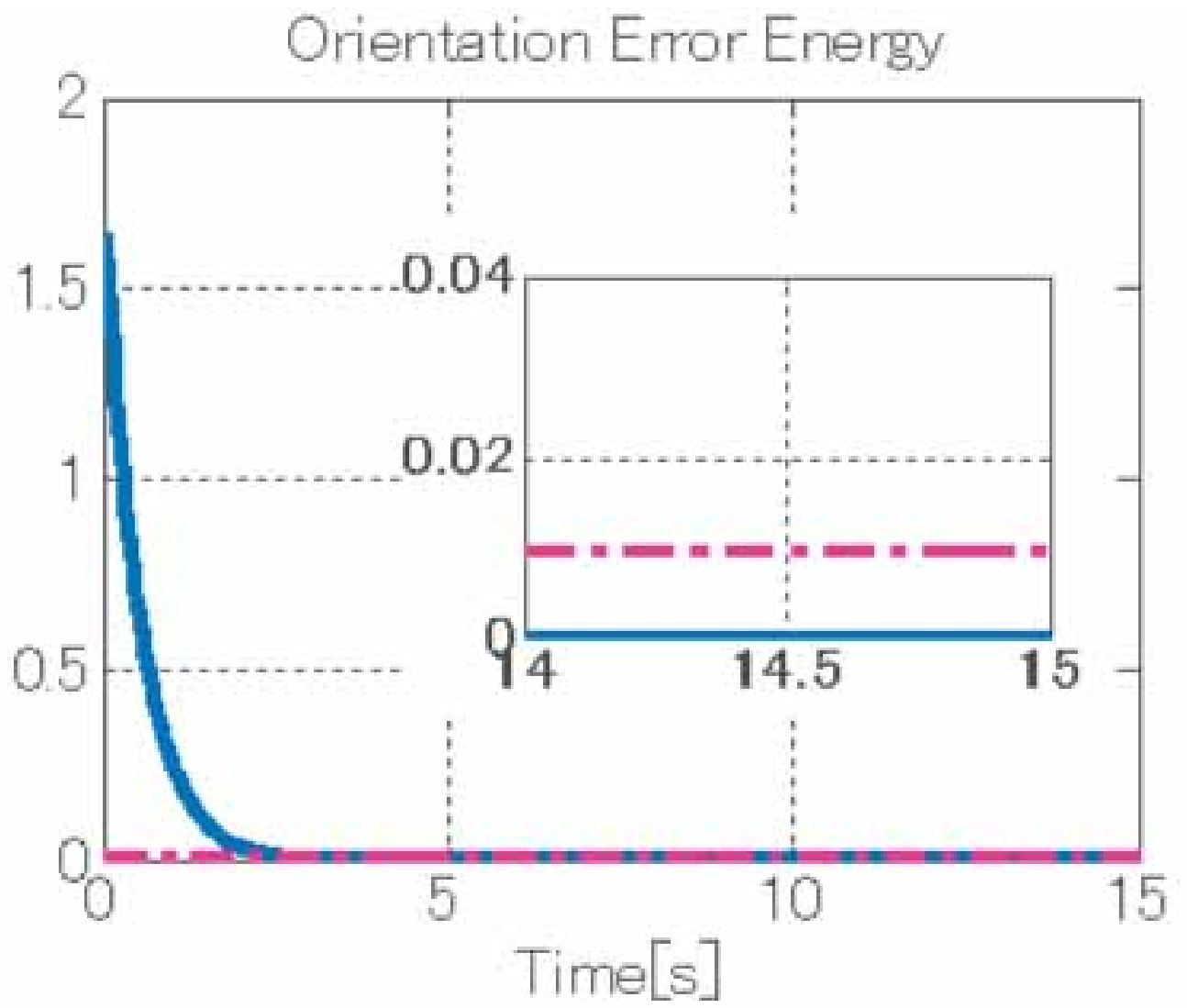}
\end{center}
\end{minipage}
\end{center}
\caption{Time Response of $U_p$ (Left) and $U_R$ (Right) (Static: $k_s = 100$)}
\label{fig:res2_8}
\end{figure}

In the first scenario, we consider static target objects and
demonstrate validity of Theorem \ref{thm_journal:1}.
Then, the average $g^* = (p^*, \es)$ is given by
$p^* = [
    0.33\
    0.36\
   -2.96
],\
\xi\theta^* = 
[
   -0.32\
   -0.34\
   -0.34
]$.
For the configuration of the target objects,
the parameter $\beta$ is given by about $\beta = 0.86$.
Throughout this section, we let the initial estimates
be $\bar{p}_{io_i}(0) = [0\ 0\ 2.5]^T$ and
$\ebioi(0) = I_3$.

We first employ the gains $k_e = 1$ and $k_s = 0.1$ ($k=10$).
Then, the parameters $\varepsilon_p$ and $\varepsilon_R$
in Theorem \ref{thm_journal:1} are given by $\varepsilon_p = \varepsilon_R = 1$.
Fig. \ref{fig:res1_2} illustrates the time responses of
orientation estimates of all vision cameras produced by the networked visual motion observer,
where the red dash-dotted lines 
represent each element of the average $\xi^*\sin\theta^*$.
We see from the figures that there exist gaps between the average and the estimates
for all elements.
The error functions $U_p$ and $U_R$ are depicted by blue curves in Fig.
\ref{fig:res1_8} respectively,
where red dash dotted lines represent $\varepsilon_p \rho_p$ and $\varepsilon_R \rho_R$.
Namely, Theorem \ref{thm_journal:1} implies that the blue curve eventually
takes lower values than the value indicated by the dash dotted line
and we see that it is really achieved as expected.

\begin{figure}[t]
\centering
\begin{minipage}{4.5cm}
\begin{center}
\includegraphics[width=4.5cm]{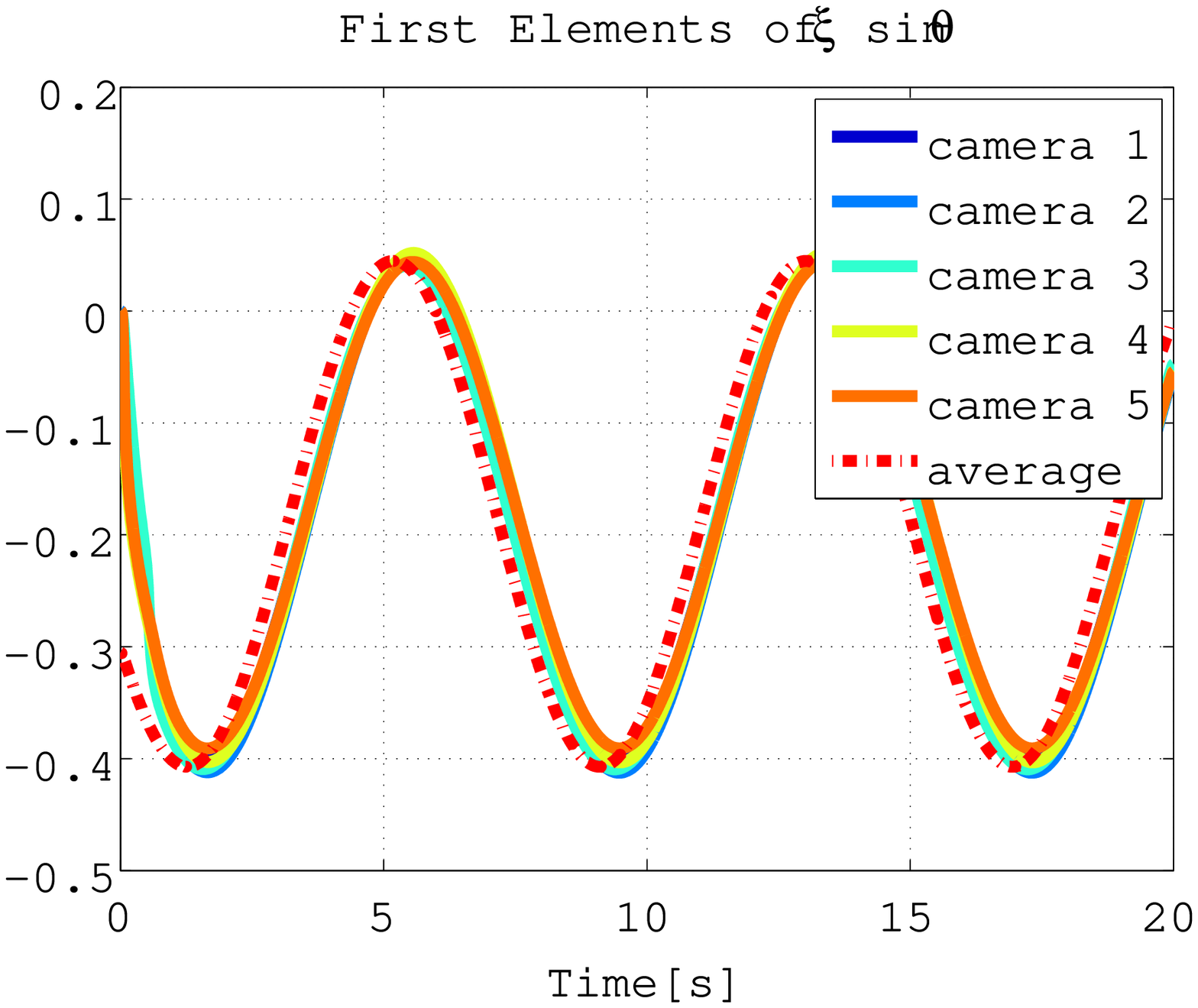}
\end{center}
\end{minipage}
\hspace{5mm}
\begin{minipage}{4.5cm}
\begin{center}
\includegraphics[width=4.5cm]{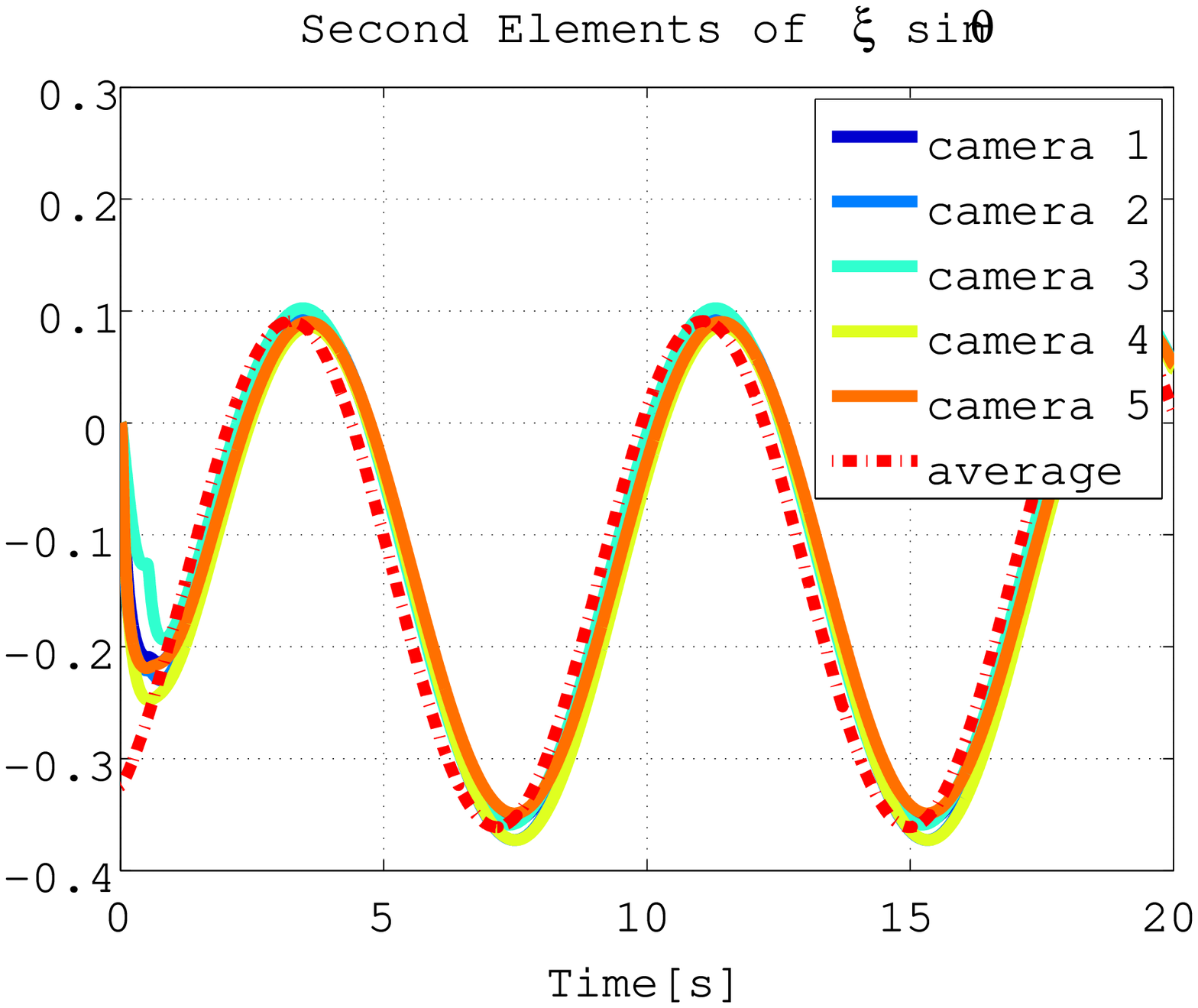}
\end{center}
\end{minipage}
\hspace{5mm}
\begin{minipage}{4.5cm}
\begin{center}
\includegraphics[width=4.5cm]{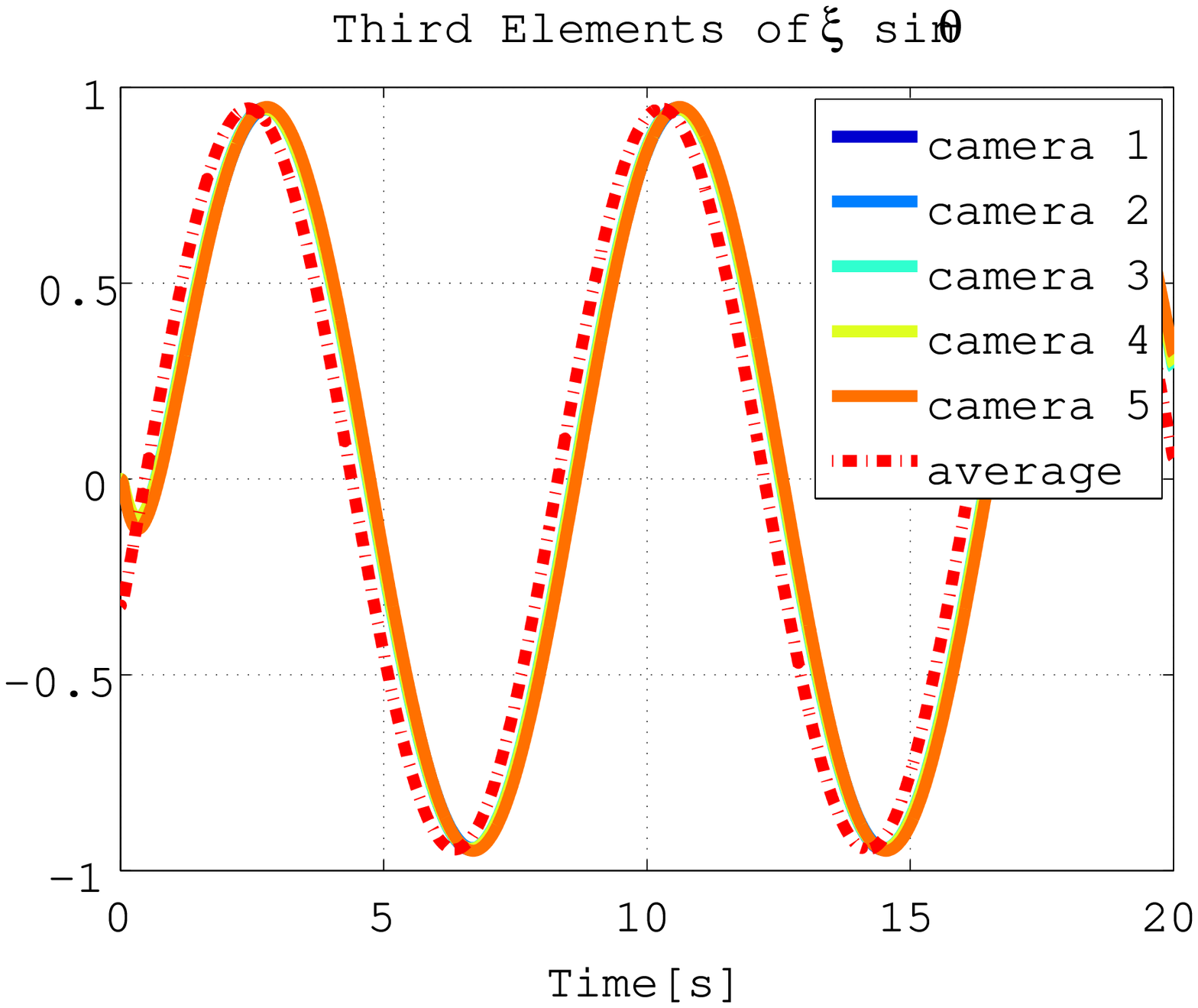}
\end{center}
\end{minipage}
\caption{Time Responses of Each Element of 
$\bar{\xi}\sin\bar{\theta},\ i=1,\cdots,5$ 
 (Moving: $k_e = 3$)}
\label{fig:res3_2}
\begin{center}
\begin{minipage}{6cm}
\begin{center}
\includegraphics[width=6cm]{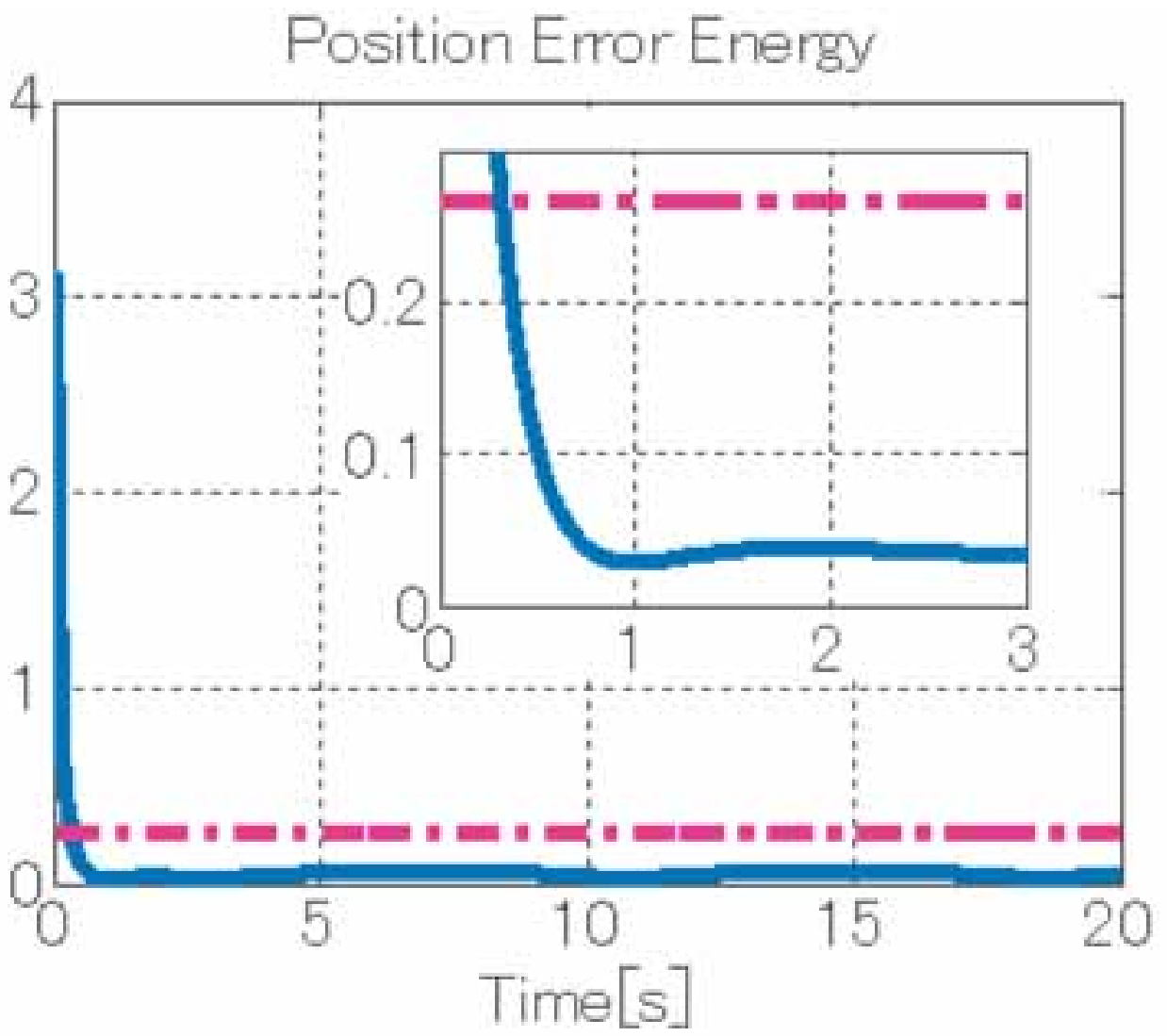}
\end{center}
\end{minipage}
\hspace{1.5cm}
\begin{minipage}{6cm}
\begin{center}
\includegraphics[width=6cm]{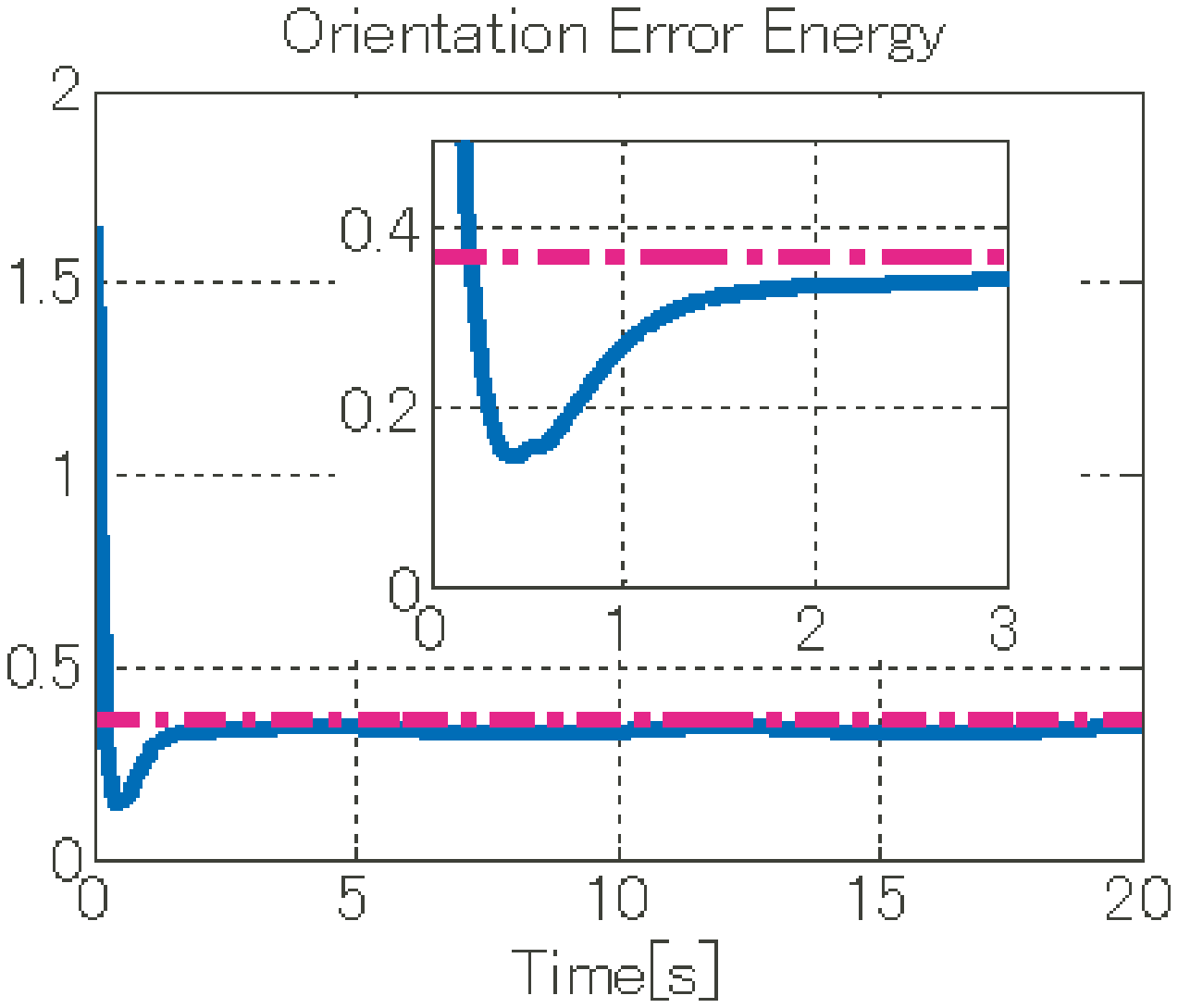}
\end{center}
\end{minipage}
\caption{Time Responses of $U_p$ (Left) and $U_R$ (Right) (Moving: $k_e = 3$)}
\label{fig:res3_8}

\begin{minipage}{6cm}
\begin{center}
\includegraphics[width=6cm]{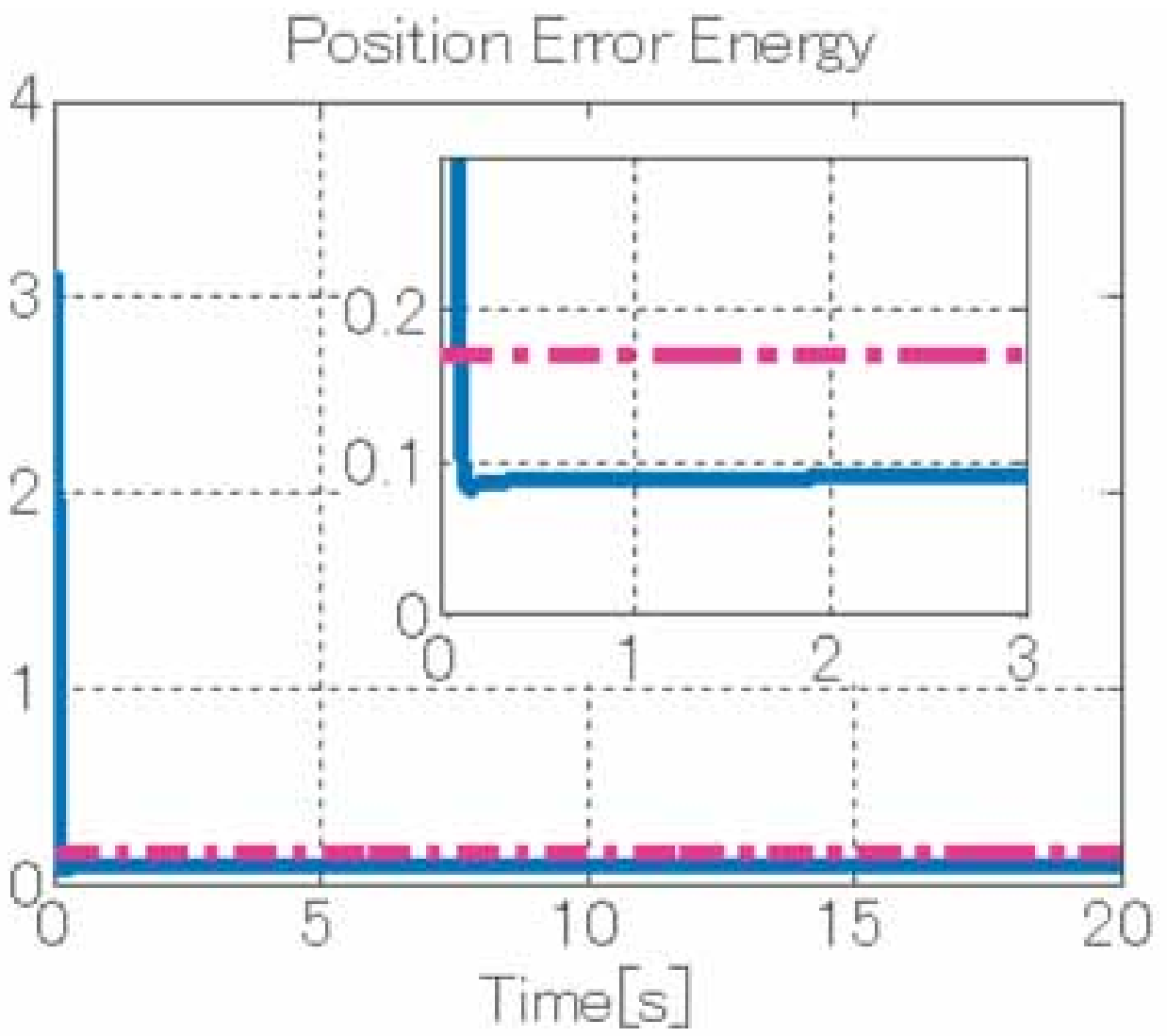}
\end{center}
\end{minipage}
\hspace{1.5cm}
\begin{minipage}{6cm}
\begin{center}
\includegraphics[width=6cm]{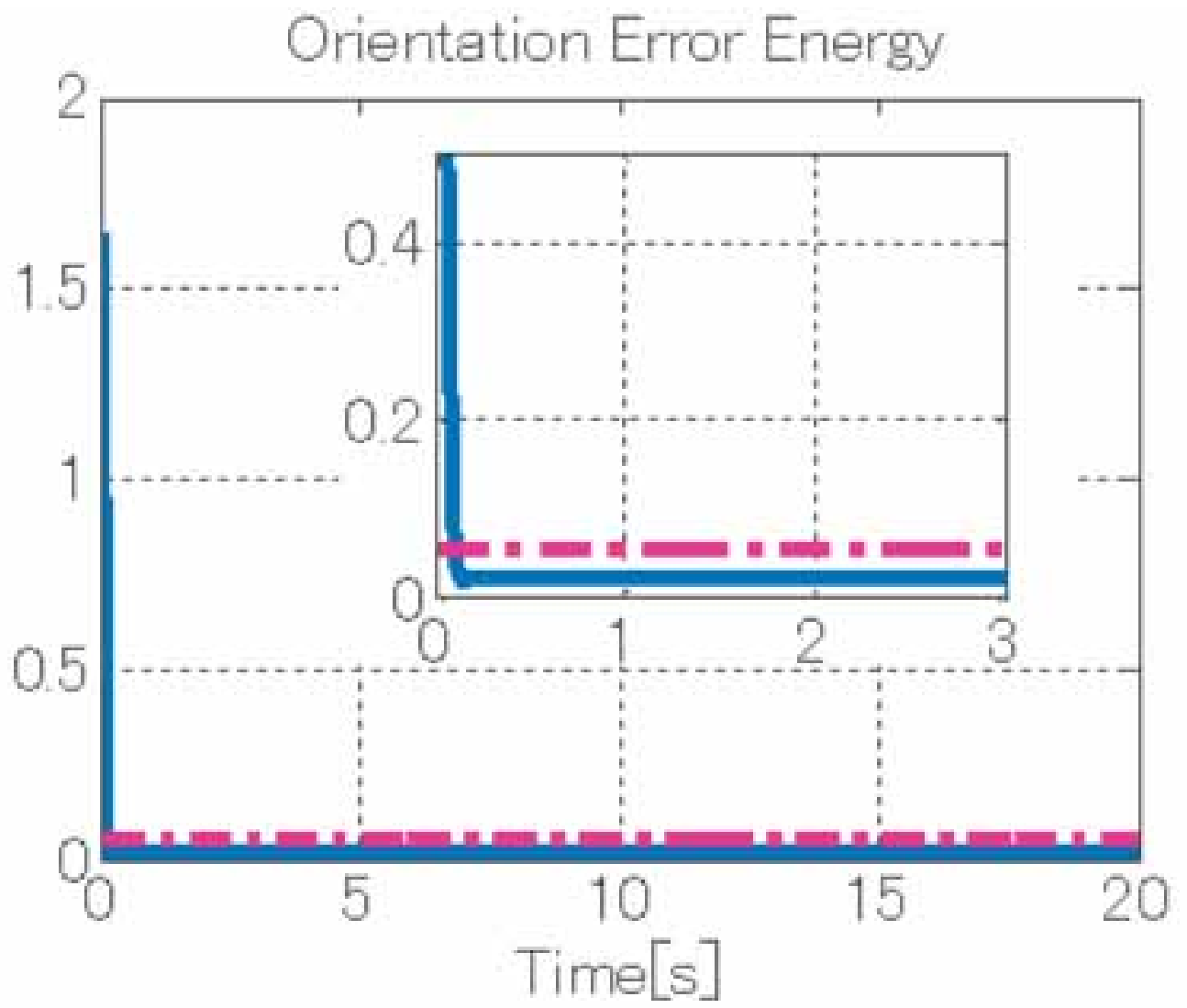}
\end{center}
\end{minipage}
\caption{Time Responses of $U_p$ (Left) and $U_R$ (Right) (Moving: $k_e = 30$)}
\label{fig:res3_11}
\end{center}
\end{figure}

We next let $k_e = 1$ and $k_s = 100$ ($k=0.01$).
Then, we have $\varepsilon_p = 0.19,\ \varepsilon_R = 0.31$ for sufficiently 
small $\epsilon$ and $c$.
Fig. \ref{fig:res2_8} illustrates the time responses of $U_p$ and $U_R$.
We see from the figures that
the estimates of all vision cameras become much closer to the average
than the case of a small mutual feedback gain $k_s = 0.1$.
Fig. \ref{fig:res2_8} also
indicates that the error functions $U_p$ and $U_R$ ultimately take
lower values than the right-hand side of (\ref{16a}) and (\ref{16b}) respectively.
Namely, it turns out as predicted that a small $k = k_e/k_s$ results in a good averaging performance.

In the second scenario, we consider moving target objects with
constant body velocities
$V^b_{wo_i} = \begin{bmatrix}
0.2&0&0&0&0&0.8
\end{bmatrix}^T\ \ {\forall i} = 1,\cdots,5$
and the same initial states as the above static case.
For the targets, we apply the networked visual motion observer
with $k_e = k_s = 3$,
where we let the initial estimates be the same as the above static object case.
Then the time responses of orientation estimates are depicted
in Fig. 
\ref{fig:res3_2},
where red dash dotted curves describe the average motion of the target orientations.
We see from the figures that the estimates track the moving average
within bounded errors and the networked observer
also works for a dynamic problem.

The responses of $U_p$ and $U_R$ are illustrated 
in Fig. \ref{fig:res3_8},
where the dash-dotted lines show $\varepsilon_p' \rho_p'$ and $\varepsilon_R' \rho_R'$.
As shown in Theorem \ref{thm_journal:3}, both of $U_p$ and $U_R$
ultimately take values smaller than $\varepsilon_p' \rho_p'$ and $\varepsilon_R' \rho_R'$
respectively.
Their counterparts for $k_e=30,\ k_s = 3$ are shown in Fig. \ref{fig:res3_11},
which also illustrate validity of Theorem \ref{thm_journal:3}.
We also see that a large $k_e$ achieves a better tracking performance
than a smaller $k_e$, which supports validity of the analysis at the 
end of Section \ref{sec:6}.

Experimental verifications on a testbed are omitted in this paper but shown in 
\cite{ACC11,CDC11}.

\section{Conclusions}
\label{sec:8}

This paper has presented a novel cooperative estimation mechanism
for visual sensor networks.
We have considered the situation where multiple smart vision cameras with
computation and communication capability
see a group of target objects.
We first have presented an estimation mechanism called networked
visual motion observer to meet 
two requirements, averaging and tracking.
Then, we have derived an upper bound of the ultimate error between
the actual average and the estimates produced by the present methodology.
Moreover, we have derived an upper bound of the ultimate error
from the estimates to the average when the target objects are moving.
Finally, the effectiveness of the present mechanism
has been demonstrated through simulation.

The authors would like to express sincere appreciation to 
Prof. Francesco Bullo and Prof. Kenji Hirata for their invaluable
suggestions and advices.



\appendices

\section{Proof of Lemma \ref{lem_journal:1}}
\label{app:2}

In the proof, 
we use the following lemma.
\begin{lemma}\cite{TCST09}
\label{lem:A_1}
For any matrices $R_1, R_2, R_3\in \SO(3)$, the inequality
\begin{eqnarray}
 \frac{1}{2}\tr(R_1^TR_2 - R_1^TR_3R_2^TR_3) \geq 
\phi(R_1^TR_3) - \phi(R_1^TR_2) 
+ 
\lambda_{min}(\sym(R_1^TR_3))\phi(R_3^TR_2)
\nonumber
\end{eqnarray}
holds, where $\sym(M) := \frac{1}{2}(M+M^T)$ and
$\lambda_{min}(M)$ is the minimal eigenvalue
of matrix $M$.
\end{lemma}
%
%
The time evolution of the orientation estimate $\ebioi$
in (\ref{eqn:EsRRBM}) with $V^b_{wi} = 0$ and
(\ref{eqn:ce_update}) is given by
\begin{eqnarray}
&&\hspace{-1cm}\debioi = \ebioi\hat{\omega}_{uei},\
{\omega}_{uei} =
k_e e_R(\mebioi\eioi) + k_s
\sum_{j \in {\N}_i}e_R(\mebioi\ebioj),
\label{eqn:ori_update_i}
\end{eqnarray}
which is independent of evolution of the position estimate
$\bar{p}_{io_i}$.
Multiplying $e^{\hat{\xi}\theta_{wi}}$ to (\ref{eqn:ori_update_i}) from left,
we have the following equation describing evolution of the estimate
$\ebwoi$ relative to 
$\Sigma_w$.
\begin{eqnarray}
&&\hspace{-1cm}\debwoi = \ebwoi\hat{\omega}_{uei},\
{\omega}_{uei} =
k_e e_R(\mebwoi\ewoi) + k_s
\sum_{j \in {\N}_i}e_R(\mebwoi\ebwoj)
\label{eqn:ori_update}
\end{eqnarray}

Let us now consider the energy function
\begin{equation*}
U := \phi(\mesiota \ebiotaoiota) = \phi(\mes \ebwoiota),\
\iota(t) := \arg\max_{i\in \V}\phi(\mesi\ebioi(t)).
\end{equation*}
Then, the time derivative of $U$ along with the trajectories of (\ref{eqn:ori_update})
is given by
\begin{eqnarray}
&&\hspace{-.6cm}\dot{U} = 2e^T_R(\mes \ebwoiota)\omega_{ue\iota}
= -\tr \left(\sk(\mes \ewoiota)\hat{\omega}_{ue\iota}\right),
\label{eqn:dev_U}
\end{eqnarray}
where we use the relation $a^Tb = -\frac{1}{2}\tr(\hat{a}\hat{b})$.
Substituting (\ref{eqn:ori_update}) into (\ref{eqn:dev_U}) yields
\begin{eqnarray}
&&\hspace{-.6cm}\dot{U} 
= -\tr \Big(k_e\Big(\mes \ewoiota - 
\mes\ebwoiota\mewoiota\ebwoiota\Big)\nonumber\\
&&\hspace{2cm}
+ k_s\Big(\sum_{j\in \N_{\iota}}\mes \ebwoiota - 
\mes\ebwoiota\mebwoj\ebwoiota
\Big)
\Big).
\label{eqn:dev_U_2}
\end{eqnarray}
From Lemma \ref{lem:A_1}, (\ref{eqn:dev_U_2})
is rewritten as $\dot{U} \leq -\left(k_eF_1 + k_s F_2\right)$,
where
\begin{eqnarray}
&&\hspace{-.6cm}F_1:= \phi(\mes\ebwoiota)
- \phi(\mes\ewoiota) 
+ \sigma\phi(\mebwoiota\ewoiota),\ \sigma := \lambda_{min}(\sym(\mes\ebwoiota)),
\nonumber\\
&&\hspace{-.6cm}F_2:= \sum_{j\in\N_{\iota}}\Big(
\phi(\mes\ebwoiota) - \phi(\mes\ebwoj)
+\sigma\phi(\ebwoiota\ebwoj)\Big).
\nonumber
\end{eqnarray}
From the definition of the index $\iota$, the inequality
$\phi(\mes\ebwoiota) \geq \phi(\mes\ebwoj)
\ {\forall j}\in \V$
holds and hence we obtain
$F_2\geq \sigma\sum_{j\in\N_i}\phi(\ebwoiota\ebwoj)$.
Thus, the inequality
\begin{eqnarray}
&&\hspace{-.5cm}\dot{U} \leq -\Big(k_e\phi(\mes\ebwoiota)
- k_e\phi(\mes\ewoiota)
+ \sigma\Big(k_e\phi(\mebwoiota\ewoiota)+k_s\sum_{j\in\N_i}\phi(\ebwoiota\ebwoj)\Big)\Big).
\nonumber
\end{eqnarray}
is true.
From the assumption of $\mes\ebwoi > 0\ {\forall i}\in \V$,
we have $\sigma > 0$ and the inequality
\begin{eqnarray}
\dot{U} \leq -k_e(\phi(\mes\ebwoiota)
- \phi(\mes\ewoiota))
\leq -k_e(\phi(\mes\ebwoiota)
- \phi(\mes\ewol)).
\nonumber
\end{eqnarray}
holds. Thus, if $\phi(\mes\ebwoiota)- \phi(\mes\ewol) > c$,
then $\dot{U} \leq -c k_e$ is true. 
Namely, there exists a finite
$\tau(c)$ such that $\ebwoiota$ satisfies
$\phi(\mes\ebwoiota) - \phi(\mes\ewol)< c\ {\forall t} \geq \tau(c)$
and, from the definition of $\iota$, we also have
$\phi(\mes\ebwoi) - \phi(\mes\ewol)< c\ {\forall t} \geq \tau(c)$
for all $i \in \V$.
This completes the proof.

\section{Proof of Lemma \ref{lem_journal:2}}
\label{app:3}

In the proof, we use the energy functions
\begin{eqnarray}
\ \ U_p := \frac{1}{2}\sum_{i \in \V}\|p_i^* - \bar{p}_{io_i}\|^2
= \frac{1}{2}\sum_{i \in \V}\|p^* - \bar{p}_{wo_i}\|^2,\
U_R := \sum_{i \in \V}\phi(\mesi\ebioi)
= \sum_{i \in \V}\phi(\mes\ebwoi).
\nonumber
\end{eqnarray}

We first consider evolution of the position estimates 
$(\bar{p}_{io_i})_{i \in \V}$
and then show its counterpart with respect to orientation estimates 
$(\ebioi)_{i\in \V}$ separately.
The time evolution of the position estimate $\bar{p}_{io_i}$
in (\ref{eqn:EsRRBM}) with $V^b_{wi} = 0$ and
(\ref{eqn:ce_update}) is described by
$\dot{\bar{p}}_{io_i} = 
k_e (p_{io_i} - \bar{p}_{io_i}) + k_s
\sum_{j \in {\N}_i}(\bar{p}_{io_j} - \bar{p}_{io_i})$.
Since the cameras are static, the evolution of $\bar{p}_{wo_i}$
relative to the world frame $\Sigma_w$ is given by
\begin{eqnarray}
\dot{\bar{p}}_{wo_i} &\!\!\!=\!\!\!& e^{\hat{\xi}\theta_{wi}}\dot{\bar{p}}_{io_i}
=k_e (p_{wo_i} - \bar{p}_{wo_i}) + k_s
\sum_{j \in {\N}_i}(\bar{p}_{wo_j} - \bar{p}_{wo_i}),
\label{eqn:p_update}
\end{eqnarray}
which is independent of evolution of
the orientation estimates 
(\ref{eqn:ori_update}).

If we define $\bar{q}_i := \bar{p}_{wo_i} - p^*$ and
${q}_i := {p}_{wo_i} - p^*$,
the time derivative of $U_p$ along with the trajectories of
(\ref{eqn:p_update}) is given by
\begin{eqnarray}
\dot{U}_p &\!\!\!=\!\!\!& \sum_{i\in \V}\left(
k_e \bar{q}_i^T (q_i - \bar{q}_i) 
+ k_s
\sum_{j \in {\N}_i}\bar{q}_i^T (\bar{q}_j - \bar{q}_i)
\right)
\nonumber\\
&\!\!\!=\!\!\!& \frac{1}{2}\sum_{i\in \V}
\Big(k_e(\|q_i\|^2 - \|\bar{q}_i\|^2 - \|q_i 
 - \bar{q}_i\|^2)
+k_s\sum_{j \in {\N}_i}(\|\bar{q}_j\|^2 - \|\bar{q}_i\|^2 - \|\bar{q}_j 
 - \bar{q}_i\|^2)
\Big).
\nonumber
\end{eqnarray}
Since $\sum_{i\in \V}\sum_{j \in \N_i}\|\bar{q}_j\|^2 - \|\bar{q}_i\|^2 = 
0$ holds under Assumption \ref{ass:1} \cite{TCST09}, we obtain
\begin{eqnarray}
\dot{U}_p &\!\!\!=\!\!\!& \frac{1}{2}\sum_{i\in \V}\Big(k_e(\|q_i\|^2 - \|\bar{q}_i\|^2 - \|q_i 
 - \bar{q}_i\|^2)
-k_s\sum_{j \in {\N}_i}\|\bar{q}_j 
 - \bar{q}_i\|^2\Big).
\label{eqn:lem2_1}
\end{eqnarray}
We see from (\ref{eqn:lem2_1}) that if $(\bar{p}_{io_i})_{i\in \V}\in \Omega_p(1)$
then 
\begin{eqnarray}
 \dot{U}_p \leq 
-\frac{1}{2}\sum_{i\in \V}\Big(k_e\|q_i 
 - \bar{q}_i\|^2
+k_s\sum_{j \in {\N}_i}\|\bar{q}_j 
 - \bar{q}_i\|^2\Big)
\label{eqn:lem2_2}
\end{eqnarray}
holds. From Assumption \ref{ass:2}, $\sum_{i\in \V}\|q_i 
 - \bar{q}_i\|^2$ and $\sum_{i\in \V}\sum_{j \in {\N}_i}\|\bar{q}_j 
 - \bar{q}_i\|^2$ are never equal to $0$ simultaneously and
hence the right-hand side of (\ref{eqn:lem2_2})
is strictly negative.
Thus, the trajectories of the position estimates 
$(\bar{p}_{io_i})_{i\in \V}$
along with (\ref{eqn:p_update}) settle into the set $\Omega_p(1)$
in finite time.

%

The time derivative of $U_R$ along the trajectories of
(\ref{eqn:ori_update}) is given by
\begin{eqnarray}
 \dot{U}_R = 2\sum_{i\in\V} e^T_R(\mes\ebwoi){\omega}_{uei},
= -\sum_{i\in \V}\tr\left(\sk(\mes\ebwoi)\hat{\omega}_{uei}\right).
\label{eqn:lem2_3}
\end{eqnarray}
Substituting (\ref{eqn:ori_update}) into (\ref{eqn:lem2_3}) yields
\begin{eqnarray}
&&\hspace{-.6cm}
\dot{U}_R= -\sum_{i\in \V}\tr(k_e \Phi_1 + k_s \Phi_2),
\label{eqn:lem2_6}\\
&&\hspace{-.6cm}
\Phi_1 := \frac{1}{2}(\mes\ewoi - \mes\ebwoi\mewoi \ebwoi),\nonumber\\
&&\hspace{-.6cm}
\Phi_2 := \frac{1}{2}\sum_{j \in {\N}_i}\left(\mes\ebwoj 
- \mes\ebwoi\mebwoj\ebwoi\right).
\nonumber
\end{eqnarray}

We first consider the term $\sum_i\tr(\Phi_2)$ in (\ref{eqn:lem2_6}).
From Lemma \ref{lem:A_1}, the following inequality holds.
\begin{eqnarray}
&&\hspace{-.6cm}\sum_{i\in \V}\tr(\Phi_2) \geq 
\sum_{i\in\V}\sum_{j \in {\N}_i} \left\{
\phi(\mes\ebwoi) - \phi(\mes\ebwoj)
+\sigma_i\phi(\mebwoi\ebwoj)\right\},
\label{eq28}
\end{eqnarray}
where $\sigma_i:=\lambda_{min}(\sym(\mes\ebwoi))$.
Assumption \ref{ass:1} implies that
$\sum_{i\in\V}\sum_{j \in {\N}_i} \phi(\mes\ebwoi) - \phi(\mes\ebwoj) = 0$
\cite{TCST09}
and hence (\ref{eq28}) is rewritten as
\begin{eqnarray}
\sum_{i\in \V}\tr(\Phi_2) \geq \sum_{i\in\V}\sum_{j \in {\N}_i}
\sigma_i\phi(\mebwoi\ebwoj).
\label{eqn:lem2_4}
\end{eqnarray}

We next consider the term $k_e \sum_{i\in\V}\tr(\Phi_1)$ in (\ref{eqn:lem2_6}).
Applying Lemma \ref{lem:A_1} again to the term yields
\begin{eqnarray}
&&\hspace{-.6cm}\sum_{i\in\V}\tr(\Phi_1)
\geq \sum_{i\in\V} \left\{\phi(\mes\ebwoi) - \phi(\mes\ewoi)
+ \sigma_i \phi(\mebwoi\ewoi)\right\}.
\label{eqn:lem2_5}
\end{eqnarray}
Substituting (\ref{eqn:lem2_4}) and (\ref{eqn:lem2_5}) into 
(\ref{eqn:lem2_6}) yields
\begin{eqnarray}
&&\dot{U}_R\leq -
\sum_{i\in\V}\Big(k_e\phi(\mes\ebwoi) - k_e\phi(\mes\ewoi)
\nonumber\\
&&\hspace{5cm}+ \sigma_i\Big(k_e\phi(\mebwoi\ewoi)+
k_s\sum_{j \in {\N}_i}\phi(\mebwoi\ebwoj)\Big)\Big).
\label{eqn:lem2_7}
\end{eqnarray}
If $(\ebioi)_{i\in \V} \notin \Omega_R(1)$
is true, (\ref{eqn:lem2_7}) is rewritten as
\begin{eqnarray}
&&\hspace{-.5cm} \dot{U}_R\leq - 
\sum_{i\in\V}\sigma_i \Big(k_e\phi(\mebwoi\ewoi)
+k_s\sum_{j \in {\N}_i}\phi(\mebwoi\ebwoj)\Big).
\label{eqn:lem2_8}
\end{eqnarray}
Note that, from the assumption of $\mes\ebwoi > 0$, we have $\sigma_i > 0$.
Since both of the terms $\sum_{i\in \V}\sigma_i\phi(\mebwoi\ewoi)$ and 
$\sum_{i\in \V}\sigma_i
\sum_{j \in {\N}_i}\phi(\mebwoi\ebwoj)$ are never equal to
$0$ under Assumption \ref{ass:2},
the right-hand side of (\ref{eqn:lem2_8}) is strictly negative.
This implies that the trajectories of the estimates $(\ebioi)_{i\in \V}$
converge to the set $\Omega_R(1)$ in finite time.

\section{Proof of Lemma \ref{lem_journal:3}}
\label{app:4}

%
%
Suppose that
$\phi(\mes\ebwoi) < \phi(\mes\ewol) + c$ holds true for some $c > 0$.
Then, from Hoff-man-Wielandt's perturbation theorem \cite{EF}, we have
%
%
\begin{eqnarray}
\!\!&&\left|\lambda_{min}(\sym(\mes\ebwoi)) - \lambda_{min}(\sym(\mes\es))\right|
\leq \left\|\sym(\mes\ebwoi) - \sym(\mes\es)\right\|_F
\nonumber\\
\!\!&&\hspace{3cm}\leq
\|\mes(\ebwoi - \es)\|_F = \|\ebwoi - \es\|_F
< 
\sqrt{2(\phi(\mes\ewol) + c)}.
\nonumber
\end{eqnarray}
This immediately means 
\begin{equation}
\lambda_{min}(\sym(\mes\ebwoi)) \geq \beta := 1 - \sqrt{2(\phi(\mes\ewol) + c)}.
\label{eqn:lem3_1}
\end{equation}

%
%
From Lemma \ref{lem_journal:1} and (\ref{eqn:lem3_1}),
Inequality (\ref{eqn:lem2_7})  is rewritten as
\begin{eqnarray}
&&\dot{U}_R\leq -
\sum_{i\in \V}\Big(
k_e\phi(\mes\ebwoi) - k_e\phi(\mes\ewoi)
\nonumber\\
&&\hspace{3cm}+ \beta\Big(k_e
\phi(\mebwoi\ewoi)+
k_s\sum_{j \in {\N}_i}\phi(\mebwoi\ebwoj)\Big)\Big)
\label{eqn:lem3_2}
\end{eqnarray}
at least after the time $\tau(c)$.
If  
$(\ebioi)_{i\in \V}\in {\mathcal S}^R_2(k)$ holds true, then we have
\begin{eqnarray}
\dot{U}_R\leq -k_e
\sum_{i\in\V}\Big(
\phi(\mes\ebwoi) 
+ \beta\phi(\mebwoi\ewoi)\Big) 
\label{eqn:lem3_3}
\end{eqnarray}
at least after the time $\tau(c)$.
Under Assumption 2, the right-hand side of (\ref{eqn:lem3_3})
is strictly negative.

In terms of $U_p$, from (\ref{eqn:lem2_1}),
if $(\bar{p}_{io_i})_{i\in \V}\in {\mathcal S}^p_2(k)$,
we have
$\dot{U}_p = -\frac{k_e}{2}\sum_{i\in \V}\Big(\|\bar{q}_i\|^2 + \|q_i 
 - \bar{q}_i\|^2\Big)$,
whose right-hand side is strictly negative under Assumption \ref{ass:2}.
These complete the proof.

\section{Proof of Theorem \ref{thm_journal:1}}
\label{app:6}

We first consider evolution of the position estimates 
$(\bar{p}_{io_i})_{i\in \V}$ described by (\ref{eqn:p_update}).
The case not satisfying $k \leq 1/W$
is already proved in Lemma \ref{lem_journal:2} and hence
we consider the case such that
$k \leq 1/W$ is satisfied.
Lemmas \ref{lem_journal:2} and \ref{lem_journal:3} indicate that
$\dot{U}_p < 0$ holds in $\Omega_p(1)\cup {\mathcal S}^p_2(k)$.
Namely, from the inclusion (\ref{eqn:subset2}),
we have $\dot{U}_p < 0$ except for the region $\Omega_p({\varepsilon_{p}})$
if $\dot{U}_p < 0$ holds in the region ${\mathcal S}_3^p(k,\varepsilon_p)$.
If it is true, the trajectories along with (\ref{eqn:p_update})
settle into the set $\Omega_p({\varepsilon_{p}})$ in finite time.
It is thus sufficient to prove that $\dot{U}_p$ is strictly negative
for all $(\bar{p}_{io_i})_{i\in \V}\in {\mathcal S}_3^p(k,\varepsilon_p)$.

Equation (\ref{eqn:lem2_1}) is rewritten as
\begin{eqnarray}
\dot{U}_p &\!\!\!=\!\!\!& -\frac{k_e}{2}\sum_{i\in \V}\Big(-\|{q}_i\|^2 + \|\bar{q}_i\|^2 + (1-\epsilon)\|{q}_i 
 - \bar{q}_i\|^2\Big) - a_p,
\label{eqn:thm1_9}
\end{eqnarray}
where $a_p := \frac{1}{2}\sum_{i\in \V}\Big(k_e\epsilon
\|q_i - \bar{q}_i\|^2
+k_s\sum_{j \in {\N}_i}\|\bar{q}_j - \bar{q}_{i}\|^2\Big)$ is strictly positive under Assumption \ref{ass:2}.
Now, for any $\alpha\in(0,1)$ and $j^*\in \V$, we have
\begin{eqnarray}
\|q_i - \bar{q}_i\|^2 &\!\!\!\geq\!\!\!& \alpha
\|q_i - \bar{q}_{j^*}\|^2
-\frac{\alpha}{1-\alpha} 
\|\bar{q}_{i} - \bar{q}_{j^*}\|^2.
\label{eq9}
\end{eqnarray}
Let $j^*$ be a node satisfying $j^* = \arg\min_{i_0}D(i_0)$
and $G_T^* = (\V, \E_T^*) \in {\mathcal T}(j^*)$ be a graph satisfying
$G_T^* = \arg\min_{G_T \in {\mathcal T}(j^*)}\tilde{D}(G_T)$,
where $D$ and $\tilde{D}$ are defined in (\ref{eqn:defD}).
Then, we obtain
\begin{eqnarray}
\ \ \|\bar{q}_{i} - \bar{q}_{j^*}\|^2
= \Big\|
\sum_{l \in \{0,\cdots, d_{G^*_T}(i)-1\}}(\bar{q}_{v_l(i)} - \bar{q}_{v_{l+1}(i)})
\Big\|^2
\leq d_{G^*_T}(i) 
\sum_{l \in \{0,\cdots, d_{G^*_T}(i)-1\}}
\|\bar{q}_{v_l(i)} - \bar{q}_{v_{l+1}(i)}\|^2,
\nonumber
\end{eqnarray}
where $(v_0(i),\cdots,v_{d_{G^*_T}(i)-1}(i))$
is the path from root $j^*$ to node $i$
along tree $G_T^*$.
Namely, 
\begin{eqnarray}
\sum_{i\in \V} \|\bar{q}_{i} - \bar{q}_{j^*}\|^2
\leq \sum_{i\in \V}d_{G^*_T}(i) 
\sum_{l \in \{0,\cdots, d_{G^*_T}(i)-1\}}
\|\bar{q}_{v_l(i)} - \bar{q}_{v_{l+1}(i)}\|^2.
\label{eq5}
\end{eqnarray}
holds. For any edge $E = (v^1, v^2)$ of $G_T^*$, the coefficient
of $\|\bar{q}_{v^1} - \bar{q}_{v^2}\|^2$
in the right hand side of (\ref{eq5}) is given by 
$\sum_{i\in \V}\delta_{G^*_T}(E;i)d_{G^*_T}(i)$,
which is upper-bounded by $\tilde{D}(G^*_T) = W$.
We thus have
\begin{eqnarray}
\sum_{i\in \V}\|\bar{q}_{i} - \bar{q}_{j^*}\|^2  \leq W
\sum_{E = (v^1,v^2) \in \E_T^*}
\|\bar{q}_{v^1} - \bar{q}_{v^2}\|^2 
\leq W \sum_{i\in \V}\sum_{j\in \N_i}\|\bar{q}_{i} - \bar{q}_{j}\|^2.
\label{eq6}
\end{eqnarray}
The latter inequality of (\ref{eq6}) holds because 
$G^*_T$ is a subgraph of $G_u$.
Since $(\bar{p}_{io_i})_{i\in \V}\in {\mathcal S}_3^p(k,\varepsilon_p)$,
the inclusion $(\bar{p}_{io_i})_{i\in \V} \notin {\mathcal S}^p_2(k)$ holds and 
hence
\begin{equation}
\sum_{i\in \V}\sum_{j\in \N_i}\|\bar{q}_{i} - \bar{q}_{j}\|^2 = 
\sum_{i\in \V}\sum_{j\in \N_i}\|\bar{p}_{wo_i} - \bar{p}_{wo_j}\|^2
\leq 2k\rho_p.
\label{eq7}
\end{equation}
Moreover, the following 
inequality holds from the definition of the average $p^*$.
\begin{eqnarray}
\sum_{i \in \V} \|q_i - \bar{q}_{j^*}\|^2 = 
\sum_{i \in \V} \|p_{wo_i} - \bar{p}_{wo_{j^*}}\|^2
\geq
\sum_{i \in \V} \|p_{wo_i} - p^*\|^2 = 2\rho_p
\label{eq11}
\end{eqnarray}
From (\ref{eq9}), (\ref{eq6}), (\ref{eq7}) and (\ref{eq11}),
equation (\ref{eqn:thm1_9}) is rewritten as
\begin{eqnarray}
\dot{U}_p &\!\!\!\leq\!\!\!& {k_e}\left\{-\frac{1}{2}\Big(\sum_{i\in \V}\|\bar{q}_i\|^2\Big) + 
\Big(1- (1-\epsilon)\Big(\alpha - 
\frac{kW\alpha}{1-\alpha}\Big)\Big)\rho_p\right\} - a_p.
\label{eqn:thm1_11}
\end{eqnarray}
If $(\bar{p}_{io_i})_{i\in \V}\in {\mathcal S}_3^p(k,\varepsilon_p)$,
then $(\bar{p}_{io_i})_{i\in \V}\notin \Omega_p(\varepsilon_p)$
and hence (\ref{eqn:thm1_11}) is rewritten as 
\begin{eqnarray}
\dot{U}_p &\!\!\!\leq\!\!\!& {k_e}\left\{
(1-\epsilon)\Big(\alpha - 
\frac{kW\alpha}{1-\alpha}-
\left(1 - \sqrt{kW}\right)^2
\Big)\right\}\rho_p - a_p.
\nonumber
\end{eqnarray}
Under the assumption that $k\leq 1/W$, 
the inequality $\alpha - \frac{kW\alpha}{1-\alpha}\leq
\left(1 - \sqrt{kW}\right)^2$ holds for any $\alpha \in (0, 1)$
and hence $\dot{U}_p \leq -a_p < 0$.
This completes the proof of the former half of the theorem.

We next consider the evolution of the orientation estimates
$(\ebioi)_{i\in \V}$ described by (\ref{eqn:ori_update}).
The case not satisfying $k \leq \beta/W$ or $\beta > 0$
is already proved in Lemma \ref{lem_journal:2}.
We thus consider the case such that
$k \leq \beta/W$ and $\beta > 0$ hold.
We first note that the set ${\mathcal S}= \{(\ebioi)_{i\in \V}|\ \mebioi\esi > 0\ {\forall 
i}\in \V\}$ is a positively invariant set
from Lemma \ref{lem_journal:1} and hence trajectories of 
$(\ebioi)_{i\in \V}$ starting from ${\mathcal S}$
never gets out of ${\mathcal S}$.
Lemmas \ref{lem_journal:2} and \ref{lem_journal:3} also prove that,
in the region ${\mathcal S}$,
$\dot{U}_R < 0$ holds if $(\ebioi)_{i\in \V} \in ({\mathcal S}\setminus \Omega_R(1))
\cup {\mathcal S}^R_2(k)$
at least after the time $\tau(c)$.
Namely, as long as $\dot{U}_R < 0$ is true in the region ${\mathcal S}_3^R(k,\varepsilon_R)$,
the inequality $\dot{U}_R < 0$ holds 
except for the region $\Omega_R({\varepsilon_{R}})$
from the inclusion (\ref{eqn:subset2}), which
means the trajectories along with (\ref{eqn:ori_update})
settle into the set $\Omega_R({\varepsilon_{R}})$ in finite time.
It is thus sufficient to prove that 
$\dot{U}_R$ is strictly negative for all
$(\ebioi)_{i\in \V}\in {\mathcal S}_3^R(k,\varepsilon_R)$.

We first notice that if we define
$a_R := \beta\sum_{i\in\V}\Big(k_e\epsilon\phi(\mebwoi\ewoi) 
+k_s\sum_{j \in {\N}_i}\phi(\mebwoi\ebwoj)\Big)$,
$a_R$ is strictly positive under Assumption \ref{ass:2}.
Using the parameter $a_R$, (\ref{eqn:lem3_2}) is rewritten as
\begin{eqnarray}
\dot{U}_R&\!\!\!\leq\!\!\!& -k_e
\sum_{i\in\V}\Big(
- \phi(\mes\ewoi) + 
\phi(\mes\ebwoi) 
+ \beta(1 - \epsilon)\phi(\mebwoi\ewoi)
\Big)-a_R.
\label{eqn:thm1_2}
\end{eqnarray}
We thus consider the former three terms of the right hand side
of Inequality (\ref{eqn:thm1_2}).
We first have 
\begin{eqnarray}
\phi(\mebwoi\ewoi) &\!\!\!\geq\!\!\!& \alpha\phi(\mebwojs\ewoi) 
-\frac{\alpha}{1-\alpha} \phi(\mebwojs\ebwoi)
\label{eq4}
\end{eqnarray}
for any $\alpha\in(0,1)$ and $j^* \in \V$.
Again, let $j^*$ be a node satisfying $j^* = \arg\min_{i_0}D(i_0)$
and $G^*_T = (\V, \E_T^*) \in {\mathcal T}(j^*)$ be a graph satisfying
$G^*_T = \arg\min_{G_T \in {\mathcal T}(j^*)}\tilde{D}(G_T)$.
Then, the inequality
\begin{eqnarray}
\phi(\mebwojs\ebwoi) \leq d_{G^*_T}(i) 
\sum_{l \in \{0,\cdots, d_{G^*_T}(i)-1\}}
\phi(e^{-\hat{\bar{\xi}}\bar{\theta}_{wo_{v_l(i)}}}e^{\hat{\bar{\xi}}\bar{\theta}_{wo_{v_{l+1}(i)}}})
\end{eqnarray}
holds from the definition of the energy function $\phi$ and hence
\begin{eqnarray}
\sum_{i\in \V}\phi(\mebwojs\ebwoi) \leq \sum_{i\in \V}d_{G_T^*}(i) 
\sum_{l \in \{0,\cdots, d_{G^*_T}(i)-1\}}
\phi(e^{-\hat{\bar{\xi}}\bar{\theta}_{wo_{v_l(i)}}}e^{\hat{\bar{\xi}}\bar{\theta}_{wo_{v_{l+1}(i)}}}).
\label{eq1}
\end{eqnarray}
Similarly to the case of position estimates,
(\ref{eq1}) is rewritten as
\begin{eqnarray}
\sum_{i\in \V}\phi(\mebwojs\ebwoi) \leq W
\sum_{E = (v^1,v^2) \in \E_T}
\phi(e^{-\hat{\bar{\xi}}\bar{\theta}_{wo_{v^1}}}e^{\hat{\bar{\xi}}\bar{\theta}_{wo_{v^2}}})
\leq W \sum_{i\in \V}\sum_{j\in \N_i}\phi(\mebwoi\ebwoj).
\label{eq2}
\end{eqnarray}
Since $(\ebioi)_{i\in \V}\in {\mathcal S}_3^R(k,\varepsilon_R)$,
the inclusion $(\ebioi)_{i\in \V} \notin {\mathcal S}^R_2(k)$ holds and hence
\begin{equation}
\sum_{i\in \V}\sum_{j\in \N_i}\phi(\mebwoi\ebwoj)
\leq \frac{k\rho_R}{\beta}
\label{eq3}
\end{equation}
is true.
We next focus on $\phi(\mebwojs\ewoi)$ in (\ref{eq4}).
From the definition of the average $\es$ (\ref{eqn:euclidean_mean}), 
\begin{eqnarray}
\sum_{i\in \V} \phi(\mebwojs\ewoi) \geq 
\sum_{i\in \V} \phi(\mes\ewoi) = \rho_R
\label{eq10}
\end{eqnarray}
holds for any $\ebwoj \in SO(3)$.
Substituting (\ref{eq4}), (\ref{eq2}), (\ref{eq3})
and (\ref{eq10}) into inequality (\ref{eqn:thm1_2}) yields
\begin{eqnarray}
\dot{U}_R\leq -
k_e\Big\{\Big(\sum_{i\in\V}
\phi(\mes\ebwoi)\Big) 
- 
\Big(1 - (1 - \epsilon)\Big(\alpha\beta - 
\frac{kW\alpha}{1-\alpha}\Big)
\Big)\rho_R\Big\}-a_R.
\label{eqn:thm1_5}
\end{eqnarray}
If $(\ebioi)_{i\in \V}\in {\mathcal S}^R_3(k,\varepsilon_R)$,
then $(\ebioi)_{i\in \V}\notin \Omega(\varepsilon_R)$ and hence
(\ref{eqn:thm1_5}) is rewritten by
\begin{eqnarray}
\dot{U}_R\leq 
k_e
(1 - \epsilon)\Big(\alpha\beta - 
\frac{kW\alpha}{1-\alpha} - \left(\sqrt{\beta}-\sqrt{kW}\right)^2\Big)
\rho_R
-a_R.
\label{eqn:thm1_6}
\end{eqnarray}
Let us now notice that, under $k \leq \beta/W$,
$\alpha\beta - 
\frac{kW\alpha}{1-\alpha}\leq
\left(\sqrt{\beta}-\sqrt{kW}\right)^2$
holds true for any $\alpha \in (0, 1)$ and hence
$\dot{U}_R \leq - a_R < 0$.
This completes the proof of the latter half of the theorem.


\section{Proof of Lemma \ref{lem:8}}
\label{app:8}
Suppose that $S(t)$ moves from $S(t) = S$ to $S(t+t_{\Delta}) 
= S + \Delta S$.
We also describe $\es(t+t_{\Delta})$ as
$\es(t+t_{\Delta}) = \es(t) + \Delta \es,\ \Delta \es := \proj(S + \Delta 
 S) - \es(t)$.
Then, if $\|\Delta S\|_F^2\leq \bar{s}$ is true for some $\bar{s}$, 
it is proved in \cite{S93} that under (\ref{ac2})
\begin{equation}
\sup_{\Delta S} \|\Delta \es\|^2_F \leq b := 4 (1-(1-\mu^2(\gamma)\bar{s}/2)^{1/2}) 
 < \mu^2(\gamma)\bar{s}.
\label{eqn:lem8_1}
\end{equation}
The hypothesis of $\|\Delta \es\|^2_F \geq \mu^2(\gamma)\|\Delta S\|_F^2$
contradicts (\ref{eqn:lem8_1}) and hence
$\|\Delta \es\|^2_F < \mu^2(\gamma) \|\Delta S\|^2_F$.
From continuous differentiability of the average $\es$,
we also get 
\begin{eqnarray}
\|\omega^{b,*}\|^2 &\!\!\!=\!\!\!& \|\des\|_F^2
= \|\lim_{t_{\Delta}\rightarrow 0}\Delta \es/t_{\Delta}\|_F^2
= \lim_{t_{\Delta}\rightarrow 0}
\|\Delta \es/t_{\Delta}\|^2_F 
\nonumber\\
&\!\!\!<\!\!\!& \lim_{t_{\Delta}\rightarrow 0} 
\mu^2(\gamma) \|\Delta S/t_{\Delta}\|^2_F
= \mu^2(\gamma)\|\lim_{t_{\Delta}\rightarrow 0} \Delta S/t_{\Delta}\|^2_F
= \mu^2(\gamma)\|\dot{S}\|^2_F .
\nonumber
\end{eqnarray}
It is clear that $n\|\dot{S}(t)\|_F^2 \leq \|w(t)\|^2$ holds and hence
(\ref{ine}) is true.

\section{Proof of Theorem \ref{thm_journal:3}}
\label{app:20}

We first consider the statement in terms of the position estimates.
The time derivative of $U_p$ along the trajectories of the system $\Sigma_{track}$
is given by
\begin{eqnarray}
\dot{U}_p = \sum_{i\in \V}\bar{q}_i^T(\ebwoi v_{ue} - \es v^{b,*}).
\label{eq20}
\end{eqnarray}
From Lemma \ref{lem_journal:2}, we obtain
\begin{eqnarray}
\sum_{i\in \V}\bar{q}_i^T \ebwoi v_{ue} <
\frac{k_e}{2}\sum_{i\in \V}(\|q_i\|^2 - \|\bar{q}_i\|^2) \leq
k_e \rho_p' - \frac{k_e}{2}\sum_{i\in \V}\|\bar{q}_i\|^2
\label{eq21}
\end{eqnarray}
under Assumptions \ref{ass:1} and \ref{ass:3}.
In addition, under Assumption \ref{ass:3}, the second term of (\ref{eq20}) satisfies
\begin{eqnarray}
\hspace{-.8cm}&&- \sum_{i\in \V}\bar{q}_i^T\es v^{b,*} = \frac{1}{2}\sum_{i\in \V}\Big(
-\|\bar{q}_i - \es v^{b,*}\|^2 + \|\bar{q}\i\|^2 + \|\es v^{b,*}\|^2 \Big)
\leq \frac{1}{2}\sum_{i\in \V}\Big(\|\bar{q}_i\|^2 + \|v^{b,*}\|^2 \Big)
\nonumber\\
\hspace{-.8cm}&&\hspace{1cm}\leq \frac{1}{2}\Big(\sum_{i\in \V}\|\bar{q}_i\|^2\Big) + \frac{n}{2}\|v^{b,*}\|^2 
\leq \frac{1}{2}\Big(\sum_{i\in \V}\|\bar{q}_i\|^2\Big) + \frac{1}{2}\|w_p\|^2 
\leq \frac{1}{2}\Big(\sum_{i\in \V}\|\bar{q}_i\|^2\Big) + \frac{1}{2}\bar{w}_p^2.
\label{eq22}
\end{eqnarray}
Substituting (\ref{eq21}) and (\ref{eq22}) into (\ref{eq20}) yields
\begin{eqnarray}
\dot{U}_p < \sum_{i\in \V}
k_e \rho_p' - (k_e-1)\Big(\frac{1}{2}\sum_{i\in \V}\|\bar{q}_i\|^2\Big)
+ \frac{1}{2}\bar{w}_p^2.
\label{eq23}
\end{eqnarray}
Now, we see from (\ref{eq23}) and the definition of $\varepsilon_p'$
that $\dot{U}_p < 0$ as long as $(\bar{p}_{io_i})_{i\in \V}\notin \Omega_p'(\varepsilon_p')$.
Hence, the function $U_p(t)$ 
is monotonically strictly decreasing in the region
and there exists a finite time $T$ such that
$(\bar{p}_{io_i})_{i\in \V} \in \Omega_p'(\varepsilon_p')\ {\forall t}\geq T$.

We next consider the evolution of orientation estimates.
The time derivative of $U_R$ along the trajectories of the system $\Sigma_{track}$
is given by
\begin{eqnarray}
\dot{U}_R =
2\sum_{i\in \V}e^T_R(\mes \ebwoi)(\omega_{uei}-\omega^{b,*}).
\label{eqn:thm3_1}
\end{eqnarray}
From Lemma \ref{lem_journal:2}, we obtain
\begin{eqnarray}
2\sum_{i\in \V}e^T_R(\mes \ebwoi)\omega_{uei} &\!\!\!<\!\!\!& -k_e \sum_{i\in \V}\Big(\phi(\mes \ebwoi) - \phi(\mes\ewoi)\Big)
\nonumber\\
&\!\!\!\leq\!\!\!& k_e\rho'_R - k_e \sum_{i\in \V}\phi(\mes \ebwoi)
\label{eqn:thm3_2}
\end{eqnarray}
under the assumption of $(\ebioi(t))_{i\in \V}\in {\mathcal S}\ 
 {\forall t}\geq 0$ and Assumptions \ref{ass:1} and \ref{ass:3}.
We also have
\begin{eqnarray}
\hspace{-1cm}&&  
-2\sum_{i\in \V}e^T_R(\mes \ebwoi)\omega^{b,*}\nonumber\\
\hspace{-1cm}&& 
= -\sum_{i\in \V}\Big(\frac{1}{\mu^2(\gamma)}\|\mu^2(\gamma)e_R(\mes \ebwoi) + \omega^{b,*}\|^2
- \mu^2(\gamma)\|e_R(\mes \ebwoi)\|^2\Big) + \frac{n}{\mu^2(\gamma)} \|\omega^{b,*}\|^2
\nonumber\\
\hspace{-1cm}&&
\leq \sum_{i\in \V}\Big(\mu^2(\gamma)\|e_R(\mes \ebwoi)\|^2\Big) 
+ \frac{n}{\mu^2(\gamma)}\|\omega^{b,*}\|^2
\leq \sum_{i\in \V}\Big(\mu^2(\gamma)\|e_R(\mes \ebwoi)\|^2\Big) + \bar{w}_R^2,
\label{eqn:thm3_3}
\end{eqnarray}
where the last inequality holds from Lemma \ref{lem:8}.
Since $\|e_R(\mes \ebwoi)\|^2 \leq \phi(\mes \ebwoi)$ is true,
substituting (\ref{eqn:thm3_2}) and (\ref{eqn:thm3_3}) 
into (\ref{eqn:thm3_1}) yields
\begin{eqnarray}
 \dot{U}_R < k_e\rho'_R + \bar{w}_R^2-(k_e-\mu^2(\gamma))\sum_{i\in \V}\phi(\mes \ebwoi).
\nonumber
\end{eqnarray}
Now, if $(\ebioi)_{i\in \V} \notin \Omega_R'(\varepsilon_R')$ holds, then $\dot{U}_R < 0$.
Hence, the function $U_R(t)$ 
is monotonically strictly decreasing in the region
and 
this completes the proof.

\nocite{*}

\end{document}